\documentclass[12pt,dvips,twoside,a4paper]{article}
\usepackage[usenames]{color}

\usepackage{fancyhdr}
\usepackage{graphicx}
\usepackage{geometry}

\RequirePackage{amsfonts,amssymb,amsmath,amscd,amsthm}
\RequirePackage{graphicx}
\RequirePackage{xcolor}

\def\figurename{Figure} 
\makeatletter
\renewcommand{\fnum@figure}[1]{\figurename~\thefigure.}
\makeatother

\def\tablename{Table} 
\makeatletter
\renewcommand{\fnum@table}[1]{\tablename~\thetable.}
\makeatother

\makeatletter
\newcommand\footnoteref[1]{\protected@xdef\@thefnmark{\ref{#1}}\@footnotemark}
\makeatother

\usepackage{amsmath}
\usepackage{amssymb}
\usepackage{amsfonts}
\usepackage{amsthm,amscd}

\newtheorem{theorem}{Theorem}[section]
\newtheorem{lemma}[theorem]{Lemma}

\newtheorem{proposition}[theorem]{Proposition}
\theoremstyle{definition}
\newtheorem{definition}[theorem]{Definition}

\theoremstyle{remark}

\newtheorem{assumption}{Assumption}[section]

\numberwithin{equation}{section}

\newcommand{\eq}[1]{(\ref{#1})}

\def\R{\mathbb R}

\def\HH{\mathcal{H}}

\newenvironment{namelist}[1]{%
\begin{list}{}
{
\settowidth{\labelwidth}{#1}
\setlength{\leftmargin}{1.1\labelwidth}}
}{%
\end{list}}

\begin{document}
\title{\bfseries\scshape{Wave Operators and Similarity for Long Range $N$-body Schr\"odinger Operators}}
\author{\normalsize Dedicated to the late Professor Tosio Kato on his 100th birthday\\
\\
\normalsize Hitoshi Kitada\thanks{E-mail address: \tt{kitada@ms.u-tokyo.ac.jp}}\\
\normalsize Department of Mathematics\\
\normalsize Metasciences Academy\\\normalsize Tokyo, Japan
\\}
\date{\normalsize }
\maketitle \thispagestyle{empty} \setcounter{page}{1}

\thispagestyle{fancy} \fancyhead{}
\renewcommand{\headrulewidth}{0pt}

\begin{abstract} \noindent
We consider asymptotic behavior of $e^{-itH}f$ for $N$-body Schr\"odinger operator $H=H_0+\sum_{1\le i<j\le N}V_{ij}(x)$ with long- and short-range pair potentials $V_{ij}(x)=V_{ij}^L(x)+V_{ij}^S(x)$ $(x\in \R^\nu)$ such that $\partial_x^\alpha V_{ij}^L(x)=O(|x|^{-\delta-|\alpha|})$ and $V_{ij}^S(x)=O(|x|^{-1-\delta})$ $(|x|\to\infty)$ with $\delta>0$. Introducing the concept of scattering spaces which classify the initial states $f$ according to the asymptotic behavior of the evolution $e^{-itH}f$, we give a generalized decomposition theorem of the continuous spectral subspace $\HH_c(H)$ of $H$. The asymptotic completeness of wave operators is proved for some long-range pair potentials with $\delta>1/2$ by using this decomposition theorem under some assumption on subsystem eigenfunctions.
\end{abstract}


\vspace{.08in} \noindent \textbf{Keywords}: wave operators, similarity, 
$N$-body Schr\"odinger operator, 
long-range scattering, 
asymptotic behavior, scattering space, extended micro-local analysis

\section{Introduction}\label{Introduction}

We consider $N$ particles $1,2,...,N$ ($N\ge 2$) with mass $m_1,...,m_N$ which move in $\R^\nu$, where $\nu\ge1$ is an integer.
 We assume that there is no external force\footnote{See \cite{HMS}, \cite{MS} for the case when external forces are included. Especially \cite{MS} deals with time-dependent external force.}, so that only possible interaction is the one between two non-identical particles $i$ and $j$. We denote the position of the $i$-th particle by $r_i=(r_{i1},\dots,r_{i\nu}) \in \R^\nu$, and write the relative position between the $i$-th and $j$-th particles as $x_{ij}=r_i-r_j\in \R^\nu$, where ${(i,j)}$ is a pair of particles $i$ and $j$ with $1\le i<j\le N$. Then it will be reasonable to assume that the interaction between $i$-th and $j$-th particles is given by a pair potential $V_{ij}(x_{ij})$ determined by the relative position $x_{ij}$ of $i$ and $j$. 
We use the notation $\langle x\rangle=(1+|x|^2)^{1/2}$ ($x\in \R^d$, $d\ge1$ an integer), and assume the following condition on the pair potentials.

\begin{assumption}\label{potentialdecay} The pair potential 
$V_{ij}(x)$ ($x\in \R^\nu$) is divided into a sum of a real-valued $C^\infty$ function $V_{ij}^L(x)$ (long-range pair potential) and a real-valued measurable function $V_{ij}^S(x)$ (short-range pair potential) such that the following holds for some constants $\delta\in (0,1)$, $C>0$ and $C_\alpha>0$ for any multi-index $\alpha$.
\begin{align}
&|\partial_{x}^\alpha V_{ij}^L(x)|\le C_\alpha\langle x\rangle^{-|\alpha|-\delta},\label{1.3}\\
&|V_{ij}^S(x)|\le C\langle x\rangle^{-1-\delta}.\label{1.4}
\end{align}
\end{assumption}
\noindent

\pagestyle{fancy} \fancyhead{} \fancyhead[EC]{H. Kitada} \fancyhead[EL,OR]{\thepage}
\fancyhead[OC]{Long range $N$-body Schr\"odinger operators}
\fancyfoot{}
\renewcommand\headrulewidth{0.5pt}

Under this assumption with some additional assumption, we will consider the asymptotic behavior as $t\to\pm\infty$ of the evolution $e^{-itH}f$ of the Schr\"odinger operator $H$ defined in a Hilbert space $\HH$ for the state function $f$ in the continuous spectral subspace $\HH_c=\HH_c(H)$ of $H$. In the case of $N$-body Hamiltonian $H$, there is a natural decomposition of $H$ according to cluster decomposition $a=\{C_1,\dots,C_k\}$ of the set $\{1,2,\dots,N\}$ such as $H=H_a+I_a$ with $H_a=H^a+T_a$, where $H^a$ is the subsystem Hamiltonian describing the internal world of each cluster $C_j$ $(1\le j\le k)$, and $T_a$ is the Hamiltonian describing the intercluster free motion. $I_a$ is the sum of the intercluster interactions described by intercluster pair potentials, which include both of the internal coordinates and intercluster coordinates. We denote by $|a|$ the number $k$ of clusters in $a$, i.e. $|a|=k$ for $a=\{C_1,\dots,C_k\}$. The description of asymptotic behavior of $e^{-itH}$ is usually supposed to be given by channel wave operators
\begin{align}
W_a^\pm=\text{s-}\lim_{t\to\pm\infty}e^{itH}J_ae^{-itH_a}P^a,\label{waveoperators}
\end{align}
where $P^a$ is the orthogonal projection onto the eigenspace of $H^a$, and $J_a$ is a suitable time-independent modifier introduced to deal with the long-range tail of the pair potentials. 
If such wave operators exist, they intertwine $H$ and $H_a$ in the following way.
\begin{align}
E_H(B)W_a^\pm=W_a^\pm E_{H_a}(B)\quad\text{for all Borel sets }B\text{ of }\R.\label{intertwiningproperty0}
\end{align}
It has been supposed that the sum of the ranges $\mathcal{R}(W_a^\pm)$ of the channel wave operators
is equal to the continuous spectral subspace $\HH_c(H)$
\begin{align}
\HH_c(H)=\bigoplus_{2\le|a|\le N}\mathcal{R}(W_a^\pm).\label{usualasymptoticcompleteness}
\end{align}
This is called `asymptotic completeness of wave operators,' so that asymptotic behavior of $e^{-itH}f$ for $f\in \HH_c(H)=\bigoplus_{2\le|a|\le N}\mathcal{R}(W_a^\pm)$ is given as a sum of partly free evolutions $J_ae^{-itH_a}g_a$ for some $g_a\in P^a\HH$ ($2\le |a|\le N$). 

Tosio Kato \cite{Kato-wave} calls this situation that the wave operators $W_a^\pm$ give a similarity between $H$ on $\HH_c(H)$ and $H_a$'s on $\HH$. Namely `similarity' means the usual existence and the asymptotic completeness of wave operators, and this problem has been solved in the case $N=2$ by \cite{KK}, \cite{KK2}, \cite{Saito1}, \cite{Saito2}, \cite{Kuroda1}, \cite{Kuroda2}, \cite{Agmon}, \cite{[E1]}\footnote{\label{footnotelabel}historical order} for the short-range potential, and by \cite{Kitada-0}, \cite{KiI}, \cite{KiII}, \cite{[E1LR]}, \cite{[IK]}, \cite{Hor4}\footnoteref{footnotelabel} for the long-range potential. For further references see \cite{[Yaf-4]}. In the case $N\ge 3$, Sigal-Soffer \cite{[SS]} proved the similarity for the short-range pair potentials by the use of usual channel wave operators. Their proof relies on a decomposition of phase space. Later than Sigal-Soffer \cite{[SS]},  Graf \cite{[G]} gave a proof based on an improved decomposition of configuration space. Kitada \cite{[K1]} also gave a proof based on another decomposition of configuration space.

For long-range pair potentials, Enss \cite{[Enss4]} gave a proof of similarity for long-range pair potentials with $\delta>\sqrt{3}-1$ for 3-body case. Wang \cite{[XPWang6]} improved the method of Enss \cite{[Enss4]} to 3-body long-range pair potentials with $\delta>1/2$ whose negative parts decaying like $\langle x\rangle^{-\gamma}$ with $\gamma>2(1-\delta)/\delta$.
G\'erard \cite{[Gerard7]} gave a proof of similarity for the 3-body case with the potentials satisfying \eq{1.3}, $\delta>1/2$, and virial condition: $2V_{ij}(x)+x\cdot\nabla_{x}V_{ij}(x)\le -C\langle x\rangle^{-\delta}$ for $C>0$.
Derezi\'nski \cite{[De]} gave a proof of the asymptotic completeness for general $N$-body case with $\delta>\sqrt{3}-1$. Further references may be found in \cite{[De-Ge]} and \cite{[Yaf-3]}.

In the sense that physically interesting case of Coulomb pair potentials has been completed by the work of Derezi\'nski \cite{[De]} for general $N$-body problem, it might be natural that further investigation has not been done for longer-range potentials.

However, from the mathematical view point, it is an important problem to pursue whether or not the similarity extends to general long-range pair potentials. Further due to the progress after Sigal-Soffer \cite{[SS]}, we now have advanced techniques such as the partition of unity associated with the decomposition of configuration space as the ones in \cite{[G]} and \cite{[K1]} and the extended micro-localizing factor in \cite{KS} introduced for $N=2$. Given those, it would now be an appropriate time to begin with the investigation of general long-range pair potentials with $\delta>0$.

Our strategy to investigate general long-range pair potentials in this paper is firstly to introduce scattering spaces with utilizing the refined decomposition of configuration space. Then we will prove a generalized decomposition theorem of the continuous spectral subspace $\HH_c(H)$ of $H$ by scattering spaces under the following additional assumption. 
\begin{assumption}\label{eigenfunctiondecay}
Every eigenvector $\psi^a$ of any subsystem Hamiltonian $H^a$ $(2\le|a|\le N-1)$ satisfies $\Vert |x^a|\psi^a\Vert<\infty$.
\end{assumption}
As it is known \cite{[FH]} that the nonthreshold eigenvectors decay exponentially, this assumption concerns the threshold eigenvectors.
 We will consider the case $t\to\infty$ in the rest of the paper, since the case $t\to-\infty$ can be treated similarly.

The scattering space $S_a^r$ ($2\le |a|\le N$, $0\le r\le1$) consists of the state functions $f\in\HH$ such that the evolution $e^{-itH}f$ develops into the region where $|x_{ij}|>\sigma t$ ($(ij)\not\le a$) and $|x^a|\le {\mu} t^r$ as $t\to\infty$ for some $\sigma>0$ and any ${\mu}>0$.
In particular for the case $r=1$, we have the following generalized decomposition theorem of $\HH_c(H)$.
\begin{theorem}\label{decompositionintoorthogonalsumofscatteringspaces}
Let Assumptions \ref{potentialdecay} and \ref{eigenfunctiondecay} be satisfied. The continuous spectral subspace $\HH_c(H)$ is decomposed as an orthogonal sum of the scattering spaces $S_a^1$ with $2\le|a|\le N$.
\begin{equation}
\HH_c(H)=\bigoplus_{2\le|a|\le N}S_a^1.\label{decompsotionformula}
\end{equation}
\end{theorem}
In section \ref{existenceofwaveandinversewaveoperators} we will introduce in wave operator \eq{waveoperators} an auxiliary factor $P_a^{\varepsilon}(t)$ as in \eq{wavemicrolocalizingfacor} below, where $P_a^{\varepsilon}(t)$ is an extension of micro-localizing factor introduced in \cite{KS} for $N=2$ and localizes the state microlocally in the extended phase space $\R\times X\times X'$. Here $\R$ is the space for time parameter $t$, $X=\R^{\nu(N-1)}$ is the configuration space, and $X'=\R^{\nu (N-1)}$ denotes the conjugate momentum space. 
\begin{equation}
W_a^\pm=\mbox{s-}\lim_{t\to\pm\infty}e^{itH}P_a^{{\varepsilon}}(t)J_ae^{-itH_a}P^a.\label{wavemicrolocalizingfacor}
\end{equation}
With this extended microlocal factor, it is possible to apply the simple and beautiful Kato's celebrated `smooth operator' technique \cite{Kato-wave} to prove the existence of wave operators $W_a^\pm$ and the related limits for general $\delta>0$. This makes it possible to characterize the range of wave operators in terms of scattering spaces as in the following theorem, which will be proved in section \ref{scateringspaces} together with Theorem \ref{decompositionintoorthogonalsumofscatteringspaces}.
\begin{theorem}\label{Wapm=Sa0}
Let Assumptions \ref{potentialdecay} and \ref{eigenfunctiondecay} be satisfied.
Then wave operators \eq{waveoperators} exist, and satisfy
\begin{align}
\mathcal{R}(W_a^\pm)=S_a^0 \quad(2\le |a|\le N).\label{rangeofwaveequalsscatteringspace0}
\end{align}
\end{theorem}
Further we will prove the following.
\begin{theorem}\label{Wapm=Sa1}
Let Assumptions \ref{potentialdecay} and \ref{eigenfunctiondecay} be satisfied with $\delta\in (1/2,1)$ for all long-range pair potentials $V_{ij}^L$ and assume that short-range pair potentials vanish: $V_{ij}^S=0$ $(\forall(i,j))$.
 Suppose that the eigenspace of subsystem Hamiltonian $H^a$ $(2\le|a|\le N-1)$ is of finite dimension and that $V_{ij}(x)=V_{ij}^L(x)\ge0$ for all pairs $(i,j)$ and $x\in \R^\nu$.
Then wave operators \eq{waveoperators} exist, and satisfy
\begin{align}
\mathcal{R}(W_a^\pm)=S_a^0=S_a^1 \quad(2\le |a|\le N).\label{rangeofwaveequalsscatteringspace1}
\end{align}
\end{theorem}
Theorem \ref{decompositionintoorthogonalsumofscatteringspaces} therefore implies the asymptotic completeness for the long-range pair potentials specified in the theorem.
\begin{theorem}\label{asymptoticcompletenssshortlong} 
Let the assumptions of Theorem \ref{Wapm=Sa1} be satisfied. Then wave operators \eq{waveoperators}
\begin{align}
W_a^\pm=\text{s-}\lim_{t\to\pm\infty}e^{itH}J_ae^{-itH_a}P^a\label{waveoperatorscomplete}
\end{align}
exist, and are asymptotically complete.
\begin{align}
\HH_c(H)=\bigoplus_{2\le|a|\le N}\mathcal{R}(W_a^\pm).\label{usualasymptoticcompletenesscomplete}
\end{align}
\end{theorem}
A traditional proof of the asymptotic completeness for the short-range case will be given in section \ref{shortrangecompleteness} in order to contrast the new point of our method. 
At an early stage of the present investigation, we expected that for large part of long-range pair potentials, the asymptotic completeness would hold except for some special cases when Yafaev channel (\cite{[Yaf]}, \cite{[Yaf-2]}) occurs. However we noticed in the midst of the investigation that Yafaev channels seem to be rather dominant for the very long-range case when $\delta\in(0,1/2)$.
Yafaev channel is characterized by the condition that the initial condition $f\in\HH_c(H)$ satisfies $f\in S_a^1\setminus S_a^0$. Although we could not find any effective condition to control the occurrence of such channels, the traditional formulation of the asymptotic completeness survives as a new form as the equation \eq{decompsotionformula} in Theorem \ref{decompositionintoorthogonalsumofscatteringspaces}.

The proof for the short-range case given in section \ref{shortrangecompleteness} will use a traditional argument by mathematical induction originated in Sigal-Soffer \cite{[SS]}. In this proof there needs to prove the existence of the limit like
\begin{equation}
\begin{aligned}
\tilde \Omega_af=\lim_{t\to\infty}e^{itH_a}P_a^{\varepsilon}(t)J_a^*e^{-itH}f.
\end{aligned}\label{usualinversewave-2-introduction}
\end{equation}
In the short-range case the proof of the existence of this limit produces no problem, as the pair potentials are short-range and the integrability of the differentiation
\begin{equation}
\frac{d}{dt}(e^{itH_a}P_a^{\varepsilon}(t)J_a^*e^{-itH}f)\label{diff}
\end{equation}
is not hard to show by smooth operator technique due to the extended micro-localizing factor $P_a^{\varepsilon}(t)$. However in the long-range case, the proof of the existence of the limit itself is a problem, since the long-range part of intercluster interaction $I_a^L(x_a,x^a)$ (see \eq{potentials}) includes both of intercluster coordinate $x_a$ and internal coordinates $x^a$ when $N\ge 3$. This makes it hard to prove the integrability of \eq{diff} unless long-range pair potentials vanish even if we prepare a modifier $J_a$ that can handle the general long-range tail with $\delta>0$.
To overcome this difficulty, one inserts an intermediate interaction $I_a^L(x_a,0)$ as in Derezi\'nski \cite{[De]}, and tries to evaluate the difference 
\begin{equation}
I_a^L(x_a,x^a)-I_a^L(x_a,0)=x^a\cdot\int_0^1\nabla_{x^a}I_a^L(x_a,\theta x^a)d\theta.\label{Ia(x)-Ia(xa,xa)}
\end{equation}
To control the extra factor $x^a$ and get the integrability of \eq{diff}, one needs to analyze the internal motion and assume the condition $\delta>\sqrt{3}-1$. To avoid such problems of the traditional approach, we will firstly expand the scattering state in terms of subsystem scattering states,
and reduce the problem to the consideration of the following limit including eigenprojection $P^a$ of subsystem Hamiltonians under an additional Assumption \ref{eigenfunctiondecay} on subsystem eigenfunctions.
\begin{align}
\Omega_af=\lim_{t\to\infty}P^ae^{itH_a}J_a^*P_a^{\varepsilon}(t)e^{-itH}f.\label{eigenprojectioninversewave}
\end{align}
This step will give Theorems \ref{decompositionintoorthogonalsumofscatteringspaces} and \ref{Wapm=Sa0}. These two theorems imply that the asymptotic completeness is equivalent to the condition
\begin{equation}
S_a^0=S_a^1.
\end{equation}
When long-range pair potentials vanish, this condition readily follows from the existence of the limit \eq{usualinversewave-2-introduction} by the traditional method of induction with detouring the hard part of the discussion. When long-range pair potentials do not vanish, it is necessary to attack this problem directly, which constitutes a hard step of the problem as we will see in section \ref{scateringspaces}.

\section{Notation}\label{Notation}

We will give in this section basic notation which we will need in the following. In doing so we will avoid unnecessary abstraction, and will concentrate on the most important points. We will also assume that the reader is familiar with pseudodifferential operators, Fourier integral operators, and their calculus (see e.g., \cite{K2}-\cite{Kumano-go}), and will avoid the unnecessary complication of calculation with just stating the principal symbols and the support relations.

The coordinate space of $N$-particles is
$$
\R^{\nu N}=\{x|x=(r_1,\dots,r_N)\in \R^{\nu N}, r_j\in \R^\nu (j=1,\dots,N)\}.
$$
Let the center of mass of the $N$ particles be denoted by
$$
x_C=\frac{m_1r_1+\dots+m_Nr_N}{m_1+\dots+m_N}\in \R^\nu.
$$
We set with $n=N-1$
\begin{equation}
X_C=\{r|r\in \R^{\nu}, r=x_C\}=\R^\nu,\quad X=\{x|x\in \R^{\nu N}, x_C=0\}=\R^{\nu(N-1)}=\R^{\nu n}.\label{spaceX}
\end{equation}
Then the total space $\R^{\nu N}$ is decomposed as a direct product of $X_C$ and $X$.
\begin{equation}
\R^{\nu N}=\R^\nu\times \R^{\nu (N-1)}=X_C\times X.\label{XcxJS}
\end{equation}
As the coordinates of $X$, we adopt Jacobi coordinate system which describes the relative position of the $N$-particles.
\begin{equation}
x_i=r_{i+1}-\frac{m_1r_1+\dots+m_i r_i}{m_1+\dots+m_i}\in \R^\nu,\quad i=1,2,\dots,N-1.\label{1.5}
\end{equation}
The corresponding reduced mass $\mu_i$ is given by
\begin{equation}
\frac{1}{\mu_i}=\frac{1}{m_{i+1}}+\frac{1}{m_1+\dots+m_i}.\label{reducedmass}
\end{equation}
We equip $X$ with the inner product
\begin{equation}
\langle x,y\rangle=\sum_{i=1}^{N-1}\mu_ix_i\cdot y_i,\label{innerproduct}
\end{equation}
where $\cdot$ denote the Euclidean scalar product of $\R^\nu$. With respect to this inner product, the change of variables between Jacobi coordinates \eq{1.5} is realized by orthogonal transformations of the space $X$, while $\mu_i$ and $x_i$ depend on the order of the construction of the Jacobi coordinates \eq{1.5}.

To consider the behavior of the particles, we need to introduce the notion of clustered Jacobi coordinate.
 Let $a=\{C_1,\dots,C_k\}$ be a disjoint decomposition of the set $\{1,2,\dots,N\}$: 
$C_\ell \ne \emptyset$ $(\ell=1,2,\dots,k)$, $\cup_{\ell=1}^k C_\ell=\{1,2,\dots,N\}$, $C_\ell\cap C_m=\emptyset$ ($\ell\ne m$). We denote the number of elements of a set $S$ by $|S|$. Then $|a|=k$. We call $a$ a cluster decomposition with $|a|$ clusters $C_1,\dots,C_{|a|}$.
 A cluster decomposition $b$ is called a refinement of a cluster decomposition $a$, denoted by $b\le a$, iff every $C\in b$ is a subset of some $D\in a$, and $b\not\leq a$ is its negation: there exists a cluster $C\in b$ such that no $D\in a$ includes $C\in b$ as a subset. The notation $b<a$ means that $b\le a$ and $b\ne a$. For a pair ${(i,j)}$, ${(i,j)}\le a$ means that ${\{i,j\}}\subset D$ for some $D\in a$, and ${(i,j)}\not\leq a$ means that ${\{i,j\}}\not\subset D$ for all $D\in a$. 
 
A clustered Jacobi coordinate $x=(x_a,x^a)$ associated with a cluster decomposition $a=\{C_1,\dots,C_k\}$ is obtained by first choosing a Jacobi coordinate
$$
x^{C_\ell}=(x_1^{C_\ell},\dots,x_{|C_\ell|-1}^{C_\ell})\in X^{C_\ell}=\R^{\nu(|C_\ell|-1)}\quad (\ell=1,2,\dots,k)
$$
for the $|C_\ell|$ particles in the cluster $C_\ell$ and then by choosing an intercluster Jacobi coordinate
$$
x_a=(x_1,\dots,x_{k-1})\in X_a=\R^{\nu(k-1)}
$$
for the $k$ centers of mass of the clusters $C_1,\dots,C_k$. Then $x^a=(x^{C_1},\dots,x^{C_k})\in X^a=X^{C_1}\times \dots\times X^{C_k}=\R^{\nu(N-k)}$ and $x=(x_a,x^a)\in X_a\times X^a=\R^{\nu(k-1)}\times\R^{\nu(N-k)}=\R^{\nu(N-1)}=\R^{\nu n}=X$.
We denote the conjugate momentum space of $X$ by $X'=\{\xi|\xi\in \R^{\nu (N-1)}\}$.

Given Jacobi coordinate \eq{1.5}, the Schr\"odinger operator $H$ defined in the Hilbert space $\mathcal{H}=L^2(X)$ with the domain $\mathcal{D}(H)=H^2(X)$, and the corresponding classical Hamiltonian $H(x,\xi)$ with the center of mass motion separated and removed are expressed in the form
\begin{equation}
\begin{aligned}
&H=H(x,D_x)=H(x,D)=H_0+V,\\
&H_0=H_0(D_x)=H_0(D),\ V=V(x),\ \ (D_x=D=-i\partial_x),\\
&H(x,\xi)=H_0(\xi)+V(x),\ \ H_0(\xi)=\sum_{i=1}^n\frac{1}{2\mu_i}|\xi_i|^2, \ \ V(x)=\sum_{1\le i<j\le N}V_{ij}(x_{ij}).
\end{aligned}\label{Hamiltonian}
\end{equation}
Passing to the clustered Jacobi coordinates $x=(x_a,x^a)\in X$ and the corresponding conjugate momentum $\xi=(\xi_a,\xi^a)\in X'$ for a cluster decomposition $a=(C_1,\dots,C_k)$ with $|a|=k$, $2\le k\le N$, we see that the free part of the Hamiltonian is given by
\begin{equation}
\begin{aligned}
&H_0=H_0(D_x)=T_a+H_0^a,\\ 
&T_a=T_a(D_a),\ H_0^a=H_0^a(D^a),\ (D_a=D_{x_a}, D^a=D_{x^a}),\\
&H_0(\xi)=T_a(\xi_a)+H_0^a(\xi^a),\quad T_a(\xi_a)=\sum_{\ell=1}^{k-1}\frac{1}{2M_\ell}|\xi_\ell|^2,\\
&H_0^a(\xi^a)=\sum_{\ell=1}^kH_0^{C_\ell}(\xi^{C_\ell}),\quad H_0^{C_\ell}(\xi^{C_\ell})=\sum_{i=1}^{|C_\ell|-1}\frac{1}{2\mu_i^{C_\ell}}|\xi_i^{C_\ell}|^2,
\end{aligned}
\label{freeH}
\end{equation}
where $M_\ell$ and $\mu_i^{C_\ell}$ are reduced masses corresponding to the clustered Jacobi coordinates.
The potential part of the Hamiltonian is given by
\begin{equation}
\begin{aligned}
&V(x)=\sum_{1\le i<j\le N} V_{ij}(x_{ij})=I_a(x)+V^a(x^a).\quad
I_a(x_a,x^a)=I^S_a(x_a,x^a)+I^L_a(x_a,x^a),\\
&I^S_a(x_a,x^a)=\sum_{{(i,j)}\not\le a}V^S_{ij}(x_{ij}),\quad I^L_a(x_a,x^a)=\sum_{{(i,j)}\not\le a}V^L_{ij}(x_{ij}),\\ 
&V^a(x^a)=\sum_{C_\ell\in a}V^{C_\ell}(x^{C_\ell}),\quad V^{C_\ell}(x^{C_\ell})=\sum_{{(i,j)}\le C_\ell}V_{ij}(x_{ij}).
\end{aligned}
\label{potentials}
\end{equation}
The total Hamiltonian is given by
\begin{equation}
\begin{aligned}
&H=H_0+V=H_a+I_a=T_a+H^a+I_a,\ H=H(x,D_x)=H(x,D),\\
&H_0=H_0(D),\ H_a=H_a(x,D),\ H^a=H^a(x^a,D^a)=H_0^a(D^a)+V^a(x^a),\\ 
&H(x,\xi)=H_0(\xi)+V(x)=H_a(x,\xi)+I_a(x),\\
&H_a(x,\xi)=T_a(\xi_a)+H^a(x^a,\xi^a),\\
&H^a(x^a,\xi^a)=H_0^a(\xi^a)+V^a(x^a)=\sum_{C_\ell\in a}H^{C_\ell}(x^{C_\ell},\xi^{C_\ell}),\\
&H^{C_\ell}(x^{C_\ell},\xi^{C_\ell})=H_0^{C_\ell}(\xi^{C_\ell})+V^{C_\ell}(x^{C_\ell}).
\end{aligned}
\label{Hamiltonian-decomposed}
\end{equation}
Here $T_a$ acts in the Hilbert space $\HH_a=L^2(\R^{\nu(|a|-1)})$; $H^a$, $H_0^a$, and $V^a$ act in $\HH^a=L^2(\R^{\nu(N-|a|)})$; $H_a=T_a+H^a=T_a\otimes I+I\otimes H^a$ and $H=H_a+I_a$ act in the total Hilbert space $\HH=\HH_a\otimes \HH^a$, where $I$ in $T_a\otimes I$ and $I$ in $I\otimes H^a$ denote the identity operators in $\HH^a$ and $\HH_a$, respectively.
Making a change of variable $\xi_\ell=\sqrt{M_\ell}\xi'_\ell$ and $\xi_i^{C_\ell}=\sqrt{\mu_i^{C_\ell}}\xi'^{C_\ell}_i$, we obtain a more convenient form of the Hamiltonian
\begin{equation}
\begin{aligned}
&H(x,\xi)=H_0(\xi)+V(x),\quad H_0(\xi)=T_a(\xi_a)+H_0^a(\xi^a)=\frac{1}{2}|\xi|^2,\\
&T_a(\xi_a)=\frac{1}{2}\sum_{\ell=1}^k|\xi_\ell|^2,\quad H_0^a(\xi^a)=\frac{1}{2}\sum_{\ell=1}^k\sum_{i=1}^{|C_\ell|-1}|\xi_i^{C_\ell}|^2.
\end{aligned}
\label{conventionalHamiltonian}
\end{equation}
We note that in this setting the inner product defined by \eq{innerproduct} is just the Euclidean inner product: $\langle x,y\rangle=\sum_{i=1}^{N-1}x_i\cdot y_i$ of $X=\R^{\nu n}$.
As we write the configuration space $\R^{\nu n}$ by $X$ and the conjugate momentum space $\R^{\nu n}$ by $X'$, respectively, the phase space will be denoted by $X\times X'=\R^{\nu n}\times \R^{\nu n}=\{(x,\xi)|x\in X,\xi\in X'\}$. In the following, we consider the Hamiltonian $H=H(x,D_x)$ in \eq{conventionalHamiltonian} for an $N$-body quantum-mechanical system with $N\ge 2$ under Assumptions \ref{potentialdecay} and \ref{eigenfunctiondecay}.

\section{Continuous spectrum}\label{Contiuousspectrum}
We need to introduce some notation concerning the bound states of Hamiltonians.
Let $a$ be a cluster decomposition with $1\le |a|\le N$. We let ${P^a}=I$ for $|a|=N$, and let ${P^a}$ be the orthogonal projection onto the eigenspace $\HH_p(H^a)$ of $H^a$ defined on $\HH^a=L^2(X^a)=L^2(\R^{\nu(N-|a|)})$ for $1\le|a|\le N-1$. We write the extension $I\otimes {P^a}$ as ${P^a}$, where $I$ denotes the identity operator of $\HH_a=L^2(X_a)=L^2(\R^{\nu(|a|-1)})$ and for $|a|=1$ we write ${P^a}=P=P_H$. For an integer $M$ we let $P^a_M$ be an $M$-dimensional partial projection of the eigenprojection ${P^a}$ such that ${P^a}=\mbox{s-}\lim_{M\to\infty}P^a_M$. We use the notation $P^a_M=I\otimes P^a_M$ for the natural extension of $P^a_M$. Let a sequence of integers $M_1,M_2,\dots,M_N$ be fixed, and define ${\hat M}_a=(M_1,\dots,M_{|a|-1})$ and $M_a=({\hat M}_a,M_{|a|})=(M_1,\dots,M_{|a|-1},M_{|a|})$. Set
for $\hat M=(M_1,\dots,M_{\ell})$ $(1\le\ell\le N-1)$ and $P_{M_1}=P^{a_1}_{M_1}$ $(|a_1|=1)$
\begin{align}
{\hat P}^\ell_{\hat M}=\left(I-\sum_{|a_\ell|=\ell}P^{a_\ell}_{M_\ell}\right)\dots\left(I-\sum_{|a_2|=2}P^{a_2}_{M_2}\right)(I-P_{M_1})\label{1.19}
\end{align}
and
\begin{align}
\tilde P^a_{M_a}=P^a_{M_{|a|}}{\hat P}^{|a|-1}_{{\hat M}_a},\quad
2\le |a| \le N.\label{1.20}
\end{align}
Then we have
\begin{equation}
\sum_{2\le|a|\le N}\tilde{P}^a_{M_a}=I-P_{M_1}.\label{summationofprojections}
\end{equation}
We also set
\begin{align}
\mathcal{T}=\bigcup_{1\le|a|\le N}\sigma_p(H^a),\label{thresholds}
\end{align}
where $\sigma_p(H^a)$ denotes the point spectrum of the selfadjoint operator $H^a$ $(1\le |a|\le N)$ with $\sigma_p(H^a)=\{0\}$ for $|a|=N$. Namely $\mathcal{T}$ is the sum of the point spectrum of $H$ and the thresholds of $H$. We note that
\begin{equation}
\begin{aligned}
\HH_c(H)=\overline{\sum_{B\Subset \R\setminus\mathcal{T}}E_H(B)\HH}.\label{continuousspectralsubspace}
\end{aligned}
\end{equation}
We remark that the result of \cite{[FH]} implies that $\mathcal{T}\subset [b,0]$ for some constant $b\le0$. The first purpose of this section is to prove the following lemma which is an extension of the theorem known as RAGE theorem (\cite{RS}). 
\begin{lemma}\label{RAGE} Let $\{B(s)|s\in \R\}$ be a continuous family of uniformly bounded operators in $\HH$. Let $B\subset \R$ be a bounded open set satisfying $E_H(B)\HH\subset \HH_c(H)$ and let $2\le |b| \le N$. Then there is a constant $\epsilon_M>0$ that goes to $0$ when the components $M_j$ of the multi-index $M_b$ tend to $\infty$ such that as $T\to\infty$
\begin{equation}
\left\Vert\frac{1}{T}\int_0^T B(s) F(|x_{ij}|<R){\tilde P}^b_{M_b}e^{-isH} E_H(B)ds\right\Vert \sim_{\epsilon_M} 0
\end{equation}
for any pair $(i,j)\not\le b$. Here $\sim_{\epsilon_M}$ means that the norm of the difference of the both sides is asymptotically $\le\epsilon_M$ as $T\to\infty$, and
$F(S)$ denotes a smooth positive cut off function which is $1$ on the set $S\subset X=\R^{\nu (N-1)}$ and is $0$ outside some neighborhood of $S$. We write $\sim$ for $\sim_0$.
\end{lemma}
\begin{proof} We prove a general version of Lemma \ref{RAGE}:
\medskip

\noindent
{\it 
Under the assumption of the lemma, we have as $T\to\infty$
\begin{equation}
\left\Vert\frac{1}{T}\int_0^T B(s) F(|x_{ij}|<R)F(|x^b|<R){\hat P}^{|b|-1}_{{\hat M}_b} e^{-isH} E_H(B)ds\right\Vert \sim_{\epsilon_M}0\label{newl}
\end{equation}
for any $(i,j)\not\le b$.}

\medskip

\noindent
We prove \eq{newl} by induction on $k=|b|$.
\medskip

\noindent
I) 1st step: We prove that {\rm \eq{newl}} holds for $|b|=2$.
Since $\Vert F(|x|>S)F(|x_{ij}|<R)F(|x^b|<R)\Vert\to 0$ as $S\to\infty$ for $|b|=2$, $R<\infty$, $(i,j)\not\le b$, it suffices to show
\begin{equation}
\lim_{T\to\infty}\left\Vert\frac{1}{T}\int_0^T B(s)F(|x|<R)E_H(B)e^{-isH}ds\right\Vert=0.
\end{equation}
As the operator $F(|x|<R)E_H(B)$ is compact, the first step follows if we prove the lemma with $F(|x|<R)E_H(B)$ replaced by one dimensional operator $Kf=(f,\phi)\psi$, where $\phi\in\HH_c(H)$. We compute
\begin{equation}
\begin{aligned}
&\left\Vert\frac{1}{T}\int_0^T B(s)Ke^{-isH}ds\right\Vert^2=\left\Vert\frac{1}{T}\int_0^T e^{isH}K^*B(s)^*ds\right\Vert^2\\
&=\sup_{\Vert f\Vert=1}\left\Vert\frac{1}{T}\int_0^T e^{isH}K^*B(s)^*fds\right\Vert^2\\
&=\sup_{\Vert f\Vert=1}\frac{1}{T^2}\int_0^T\int_0^T (B(s)^* f,\psi)(\psi,B(t)^* f)(e^{-i(t-s)H}\phi,\phi)dt ds.
\end{aligned}
\end{equation}
The RHS is bounded by
\begin{equation}
C \frac{1}{T^2}\int_0^T\int_0^T |(e^{-i(t-s)H}\phi,\phi)|dt ds\\
\le C \frac{1}{T}\int_{-T}^T |(e^{-itH}\phi,\phi)|dt
\end{equation}
for $C=\Vert \psi\Vert^2\sup_{s\in \R}\Vert B(s)\Vert^2\ge0$.
By Schwarz inequality, the RHS is bounded by
\begin{equation}
\sqrt{2}C\left(\frac{1}{T}\int_{-T}^T |(e^{-itH}\phi,\phi)|^2 dt\right)^{\frac{1}{2}}.
\end{equation}
Noting that the function $\mu(\lambda)=(E_H(\lambda)\phi,\phi)$ is of bounded variation, we calculate the formula inside the parentheses.
\begin{equation}
\begin{aligned}\frac{1}{T}\int_{-T}^T\int_{\R}\int_{\R}e^{-i(\lambda-\lambda')t}d\mu(\lambda)d\mu(\lambda')dt=2\int_{\R}\int_{\R}\frac{\sin\{(\lambda-\lambda')T\}}{(\lambda-\lambda')T}d\mu(\lambda)d\mu(\lambda').\nonumber
\end{aligned}
\end{equation}
Dividing the integration region $\R^2$ into $|\lambda-\lambda'|\le \epsilon$ and the other, we obtain a bound.
\begin{equation}
2\int_{|\lambda-\lambda'|\le \epsilon}d\mu(\lambda)d\mu(\lambda')+\frac{2}{\epsilon T}.\label{sumofsmallterms}
\end{equation}
The first term is equal to
$$
\int_\R\int_{\lambda-\epsilon}^{\lambda+\epsilon}d\mu(\lambda')d\mu(\lambda)=\int_\R\Vert E_H((\lambda-\epsilon,\lambda+\epsilon])\phi\Vert^2d\mu(\lambda).
$$
Since $\Vert E_H((\lambda,\mu])\phi\Vert^2\le \Vert\phi\Vert^2$ is continuous with respect to $(\lambda,\mu)\in\R^2$ by $\phi\in \HH_c(H)$ and the measure $d\mu$ is finite on $\R$, the first term of \eq{sumofsmallterms} can be arbitrarily small if $\epsilon>0$ is taken small enough. Letting then $T\to\infty$ in \eq{sumofsmallterms} makes the second term go to $0$. 
\medskip

\noindent
II) 2nd step: Assuming \eq{newl} for $|b|<k$ $(3\le k\le N)$, we prove \eq{newl} for $|b|=k$. Let $b=\{C_1,\cdots,C_{|b|}\}$ with $|b|=k$ and assume $(i,j)$ connects the clusters $C_1$ and $C_2$ of $b$. We denote the new cluster decomposition by $d=\{C_1\cup C_2,C_3,\cdots,C_k\}$. Then $|d|=k-1$, and $K_1=F(|x_{ij}|<R)F(|x^b|<R)$ bounds\footnote{Here `$K_1$ bounds $x^d$' means that $\Vert F(|x^d|>L)K_1\Vert\to0$ as $L\to\infty$.} the variable $x^d$. We decompose ${\hat P}^{k-1}_{{\hat M}_b}$ (see \eq{1.19}) in \eq{newl} as
\begin{equation}
{\hat P}^{k-1}_{{\hat M}_b}=(I-P^d_{M_{k-1}}){\hat P}^{k-2}_{{\hat M}_d}
-\sum_{b_{k-1}\ne d} P^{b_{k-1}}_{M_{k-1}}{\hat P}^{k-2}_{{\hat M}_d},\label{decomp}
\end{equation}
where ${\hat M}_d=(M_1,\cdots,M_{k-2})$. Each $P^{b_{k-1}}_{M_{k-1}}$ in the second term bounds the variable $x^{b_{k-1}}$ with $|b_{k-1}|=k-1$. Since $b_{k-1}\ne d$, $|d|=k-1$ and $F(|x_{ij}|<R)F(|x^b|<R)$ bounds the variable $x^d$, $F(|x_{ij}|<R)F(|x^b|<R)P^{b_{k-1}}_{M_{k-1}}$ connects at least one pair of different two clusters in $b_{k-1}$. Thus the terms in the second summand on the RHS of \eq{decomp} are treated by the induction hypothesis. Thus we have to show when $T\to\infty$
\begin{equation}
\left\Vert\frac{1}{T}\int_0^T B(s) F(|x_{ij}|<R)F(|x^b|<R)
(I-P^d_{M_{k-1}}){\hat P}^{k-2}_{{\hat M}_d}e^{-isH}E_H(B)ds\right\Vert\sim_{\epsilon_M} 0.\label{b1}
\end{equation}
Let $S>0$ be arbitrary but fixed and let $t(s)=s-mS$ for $mS\le s <(m+1)S$. The norm of \eq{b1} is bounded by with $K_1=F(|x_{ij}|<R)F(|x^b|<R)$ and $K_2=K_1 (I-P^d_{M_{k-1}})$
\begin{equation}
\begin{aligned}
&\left\Vert\frac{1}{T}\int_0^T B(s) K_1
(I-P^d_{M_{k-1}})
e^{-it(s)H_d}e^{it(s)H}
{\hat P}^{k-2}_{{\hat M}_d}e^{-isH}E_H(B)ds\right\Vert\\
&+
\left\Vert\frac{1}{T}\int_0^T B(s) K_2
(I-e^{-it(s)H_d}e^{it(s)H})
{\hat P}^{k-2}_{{\hat M}_d}e^{-isH}E_H(B)ds\right\Vert.
\end{aligned}\label{twosum}
\end{equation}
Since $H-H_d=I_d$,
\begin{equation}
I-e^{-it(s)H_d}e^{it(s)H}=\int_0^{t(s)}e^{-i\tau H_d}i(H_d-H)e^{i\tau H}d\tau \quad (0\le t(s)<S) \label{int}
\end{equation}
is a sum of the terms, each of which bounds at least one variable $x_\beta$ with $\beta=(k,m)$ connecting two different clusters of $d$. Noting that
$F(|x^d|<CR)\ge K_1=F(|x_{ij}|<R)F(|x^b|<R)$ holds for some constant $C>0$,
we can treat the second term of \eq{twosum} by induction hypothesis.
The first term in \eq{twosum} is rewritten as
\begin{equation}
\left\Vert\frac{1}{T}\int_0^T B(s) K_1
(I-P^d_{M_{k-1}})
e^{-it(s)H_d}E_{H}(B)
{\hat P}^{k-2}_{{\hat M}_d}e^{-i(s-t(s))H}ds\right\Vert\label{main-contri}
\end{equation}
with some remainder terms. These remainder terms come from the commutator of $e^{it(s)H}$ and $E_H(B)$ with ${\hat P}^{k-2}_{{\hat M}_d}$, and can be treated by induction hypothesis. Since $s-t(s)=mS$, \eq{main-contri} is rewritten for $T=nS$
\begin{equation}
\left\Vert\frac{1}{nS}\sum_{m=0}^{n-1}\int_0^S B(s+mS) K_1
(I-P^d_{M_{k-1}})
e^{-isH_d}E_{H}(B)
{\hat P}^{k-2}_{{\hat M}_d}ds e^{-imSH}\right\Vert.\label{last}
\end{equation}
Since $K_1$ bounds $x^d$, the difference
\begin{equation}
K_1\{(I-P^d_{M_{k-1}})-(I-P^d)\}=K_1(P^d-P^d_{M_{k-1}})
\end{equation}
tends to $0$ in operator norm as $M_{k-1}\to\infty$. Thus we can replace $K_1(I-P^d_{M_{k-1}})$ in \eq{last} by $K_1(I-P^d)$ with an error $\epsilon_M$. This step yields the error $\epsilon_M$ in the lemma. To estimate \eq{last}, we first get an energy cut off for $H^d$ from $E_H(B)$. Then we can apply \eq{newl} for $|b|=2$ to \eq{last} with $H$ replaced by $H^d$. This proves \eq{newl} for $|b|=k$, and the proof of Lemma \ref{RAGE} is complete. 
\end{proof}
The following theorem is partially due to Volker Enss.
An outline of a proof was given at a seminar as a response to my question made at his office at California Institute of Technology when both of us visited there in 1984. Later he gave a partial proof in \cite{[En]} leaving the details to \cite{[En2]}, \cite{[En3]}. However both of \cite{[En2]} and \cite{[En3]} seem not have been published. So a complete proof is given here.

\begin{theorem}\label{quantumclassical}
Let $N=n+1\ge 2$ and let $H$ be the Hamiltonian $H$ in \eq{conventionalHamiltonian} for an $N$-body quantum-mechanical system. Let Assumptions \ref{potentialdecay} and \ref{eigenfunctiondecay} be satisfied.
Let $f\in \HH_c(H)$. Then there exist a sequence $t_m\to\pm\infty$ (as $m\to\pm\infty$) and a sequence $M_a^m$ of multi-indices whose components all tend to $\infty$ as $m\to\pm\infty$ such that for all cluster decompositions $a$ with $2\le |a|\le N$, for all $\varphi\in C_0^\infty(X_a)=C_0^\infty(\R^{\nu(|a|-1)})$, $R>0$, and $(i,j)\not\le a$
\begin{eqnarray}
&&\left\Vert\frac{|x^a|^2}{t_m^2}{\tilde P}^a_{M_a^m}e^{-it_mH}f\right\Vert \to 0\label{(2.6)}\\
&&\Vert \chi_{\{x| |x_{ij}|<R\}}{\tilde P}^a_{M_a^m}e^{-it_mH}f\Vert \to 0\label{(2.7)}\\
&&\Vert (\varphi(x_a/t_m)-\varphi(v_a)){\tilde P}^a_{M_a^m}e^{-it_mH}f\Vert \to 0\label{(2.8)}
\end{eqnarray}
as $m\to\pm\infty$. Here $\chi_S$ is the characteristic function of a set $S$.
\end{theorem}
\begin{proof}
By \eq{conventionalHamiltonian}
$$
H=H_0+V,\quad H_0=\frac{1}{2}D^2=-\frac{1}{2}\Delta, \quad V=\sum_{1\le i<j\le N}V_{ij}(x_{ij}).
$$
where
$$
D=\frac{1}{i}\frac{\partial}{\partial x},\quad x\in \R^{\nu n}.
$$
Let $f\in \HH$ satisfy $\langle x\rangle^2f\in \HH$ and $f=E_H(B)f\in\HH_c(H)$ for some bounded open set $B$ of $\R$. Note that such $f$'s are dense in $\HH_c(H)$. 
%
Writing $\tilde P^a_{M_a}=\tilde P^a$, 
we have $\sum_{2\le|a|\le N}\tilde P^a=I-P_{M_1}$ from \eq{summationofprojections}. Hence
\begin{equation}
\begin{aligned}
&\sum_{2\le|a|\le N}(f,e^{itH}\left(\frac{x}{t}-D\right)^2\tilde P^a e^{-itH}f)\\
&=\left\Vert \left(\frac{x}{t}-D\right)e^{-itH}{f}\right\Vert^2=\left(f,e^{itH}\left(\frac{x}{t}-D\right)^2e^{-itH}{f}\right)\\
&=\frac{1}{t^2}\left(f,(e^{itH}x^2e^{-itH}-x^2){f}\right)-\frac{2}{t}(f,e^{itH}Ae^{-itH}{f})\\
&\hskip160pt+(f,e^{itH}D^2e^{-itH}{f})+\frac{1}{t^2}(f,x^2{f}),
\end{aligned}\label{(2.9)}
\end{equation}
where $A=\frac{1}{2}(x\cdot D+D\cdot x)$. A direct computation gives
\begin{equation}
i[H_0,x^2]=2A.
\end{equation}
Therefore the first term on the RHS of \eq{(2.9)} is equal to
\begin{equation}
\begin{aligned}
\frac{1}{t^2}\int_0^t \frac{d}{ds}(f,e^{isH}x^2e^{-isH}{f}) ds
&=\frac{1}{t^2}\int_0^t (f,e^{isH}i[H_0,x^2]e^{-isH}{f}) ds\\
&=\frac{2}{t^2}\int_0^t(f,e^{isH}Ae^{-isH}f)ds.
\end{aligned}
\end{equation}
The sum of the first and second terms of the RHS of \eq{(2.9)} is thus equal to
\begin{equation}
\begin{aligned}
G(t)&=\frac{2}{t^2}\left(\int_0^t (f,e^{isH} Ae^{-isH}{f}) ds -t(f,e^{itH}Ae^{-itH}{f}\right)\\
&=\frac{1}{t^2}\int_0^t\frac{d(s^2 G)}{ds}(s)ds\\
&=-\frac{2}{t^2}\int_0^ts(f,e^{isH}i[H,A]e^{-isH}f)ds.
\end{aligned}\label{G(t)}
\end{equation}
Noting
$$
i[H,A]=i[T_a,A]+i[I_a,A]+i[H^a,A]=2T_a+i[I_a,A]+i[H^a,A^a],
$$
where
$A^a=\frac{1}{2}(x^a\cdot D^a+D^a\cdot x^a)$, we have
\begin{equation}
\begin{aligned}
&\sum_{2\le|a|\le N}(f,e^{itH}\left(\frac{x}{t}-D\right)^2\tilde P^a e^{-itH}f)\\
&=-\frac{2}{t^2}\int_0^ts(f,e^{isH}i[H,A]e^{-isH}f)ds+(f,e^{itH}D^2e^{-itH}f)+\frac{1}{t^2}(f,x^2f)\\
&=-\frac{2}{t^2}\sum_{2\le |a|\le N}\int_0^ts(f,e^{isH}(2T_a+i[I_a,A]+i[H^a,A^a])\tilde P^ae^{-isH}f)ds\\
&\hskip140pt+2\sum_{2\le|a|\le N}(f,e^{itH}H_0\tilde P^ae^{-itH}f)+\frac{1}{t^2}(f,x^2f)\\
&=\sum_{2\le|a|\le N}\left[-\frac{4}{t^2}\int_0^ts(f,e^{isH}T_a\tilde P^ae^{-isH}f)ds+2(f,e^{itH}T_a\tilde P^ae^{-itH}f)\right]\\
&\hskip10pt+2\sum_{2\le|a|\le N}(f,e^{itH}H_0^a\tilde P^ae^{-itH}f)\\
&\hskip10pt-\frac{2}{t^2}\sum_{2\le|a|\le N}\int_0^ts(f,e^{isH}i[I_a,A]\tilde P^ae^{-isH}f)ds\\
&\hskip10pt-\frac{2}{t^2}\sum_{2\le|a|\le N}\int_0^ts(f,e^{isH}i[H^a,A^a]\tilde P^ae^{-isH}f)ds+\frac{1}{t^2}(f,x^2f).
\end{aligned}\label{calculationsseparatingchannels}
\end{equation}
By Lemma \ref{RAGE}, the third term on the RHS goes to $0$ with some error $\epsilon_M$ as $t\to\infty$. The last term goes to $0$ as $t\to\infty$ by our assumption $\langle x\rangle^2 f\in \HH$. As Lemma \ref{RAGE} implies
$$
\sum_{2\le |a|,|b|\le N, a\ne b}\frac{2}{t^2}\int_0^ts(f,e^{isH}(\tilde P^b)^*i[H^a,A^a]\tilde P^ae^{-isH}f)ds\sim_{\epsilon_M}0
$$
as $t\to\infty$, we have asymptotically
\begin{equation}
\begin{aligned}
&\frac{2}{t^2}\sum_{2\le|a|\le N}\int_0^ts(f,e^{isH}i[H^a,A^a]\tilde P^ae^{-isH}f)ds\\
&\sim_{\epsilon_M} \frac{2}{t^2}\sum_{2\le|a|\le N}\int_0^ts(f,e^{isH}(\tilde P^a)^*i[H^a,A^a]\tilde P^ae^{-isH}f)ds.
\end{aligned}
\end{equation}
The RHS asymptotically equals
\begin{equation}
\lim_{t\to\infty}\frac{1}{t}\sum_{2\le|a|\le N}\int_0^t(f,e^{isH}(\tilde P^a)^*i[H^a,A^a]\tilde P^ae^{-isH}f)ds,\label{limitsimplifiedtimemean}
\end{equation}
if this limit exists.

We will prove that the limit \eq{limitsimplifiedtimemean} exists and equals $0$.
Letting $t(s)=s-mS$ for $mS\le s<(m+1)S$ for any fixed $S>0$, we have by Lemma \ref{RAGE} and some commutator arguments as $t\to\infty$
\begin{equation}
\begin{aligned}
&\sum_{2\le|a|\le N}\frac{1}{t}\int_0^t (f,e^{isH}({\tilde P}^a)^*i[H^a,A^a]{\tilde P}^a e^{-isH}f) ds\\
&\sim_{\epsilon_M}\sum_{2\le|a|\le N}
\frac{1}{t}\int_0^t (f,e^{i(s-t(s))H}({\tilde P}^a)^*e^{it(s)H_a}i[H^a,A^a]{\tilde P}^a e^{-isH}f) ds.
\end{aligned}\nonumber
\end{equation}
This can further be reduced and is asymptotically equal to as $t\to\infty$ with an error $\epsilon_M>0$
\begin{equation}
\sum_{2\le|a|\le N}\frac{1}{t}\int_0^t (f,e^{i(s-t(s))H}({\tilde P}^a)^*e^{it(s)H_a}i[H^a,A^a]e^{-it(s)H_a}{\tilde P}^a e^{-i(s-t(s))H}f) ds.\nonumber
\end{equation}
Noting $s-t(s)=mS$ for $mS\le s <(m+1)S$, we rewrite this for $t=nS$
\begin{equation}
\begin{aligned}
&\frac{1}{nS}\sum_{m=0}^{n-1}\int_0^S (f,e^{imSH}({\tilde P}^a)^* e^{isH_a}i[H^a,A^a]e^{-isH_a}{\tilde P}^a e^{-imSH}f) ds\\
&=\frac{1}{n}\sum_{m=0}^{n-1}\frac{1}{S}\int_0^S
\frac{d}{ds}(f,e^{imSH}({\tilde P}^a)^* e^{isH_a}A^a e^{-isH_a}{\tilde P}^a e^{-imSH}f) ds\\
&=\frac{1}{n}\sum_{m=0}^{n-1}\frac{1}{S}\bigl[(f,e^{imSH}({\tilde P}^a)^* e^{iSH_a}A^ae^{-iSH_a}{\tilde P}^a e^{-imSH}f)\\
&\hskip160pt-(f,e^{imSH}({\tilde P}^a)^* A^a{\tilde P}^a e^{-imSH}f)\bigr].
\end{aligned}\label{(3.14)}
\end{equation}
Writing ${\tilde P}^a=\sum_{j=1}^L P^{a,E_j}{\hat P}^{|a|-1}$ with $P^{a,E_j}$ being the one dimensional eigenprojection of $H^a$ with eigenvalue $E_j$, we see that the norm of the RHS is bounded by
\begin{equation}
\begin{aligned}
&\sum_{j=1}^L\frac{1}{n}\sum_{m=0}^{n-1}\frac{1}{S}\bigl|(f,e^{imSH}({\tilde P}^a)^* e^{iS(H^a-E_j)}A^a P^{a,E_j}{\hat P}^{|a|-1} e^{-imSH}f)\\
&\quad\quad\quad\quad\quad-(f,e^{imSH}({\tilde P}^a)^* A^a P^{a,E_j}{\hat P}^{|a|-1} e^{-imSH}f)\bigr|.
\end{aligned}\nonumber
\end{equation}
This is arbitrarily small when $S>0$ is fixed sufficiently large by our assumption $\Vert |x^a|P^{a,E_j}\Vert < \infty$. Thus we have proved that the limit \eq{limitsimplifiedtimemean} exists and is equal to $0$.
Hence by \eq{calculationsseparatingchannels} we have asymptotically
\begin{equation}
\begin{aligned}
&\sum_{2\le|a|\le N}(f,e^{itH}\left(\frac{x}{t}-D\right)^2\tilde P^a e^{-itH}f)\\
&\sim_{\epsilon_M}\sum_{2\le|a|\le N}\left[-\frac{4}{t^2}\int_0^ts(f,e^{isH}T_a\tilde P^ae^{-isH}f)ds+2(f,e^{itH}T_a\tilde P^ae^{-itH}f)\right]\\
&+2\sum_{2\le|a|\le N}(f,e^{itH}H_0^a\tilde P^ae^{-itH}f)
\end{aligned}\label{calculationsseparatingchannels2}
\end{equation}
 as $t\to\infty$.
A computation gives
$$
\left(\frac{x}{t}-D\right)^2=\left(\frac{x_a}{t}-D_a\right)^2+\left(\frac{(x^a)^2}{t^2}-\frac{2A^a}{t}+2H_0^a\right).
$$
This with the fact that $\Vert \chi_{\{|x^a|>R\}}\tilde P^a\Vert\to0$ as $R\to\infty$ yields that
\begin{equation}
\begin{aligned}
&\sum_{2\le |a|\le N}(f,e^{itH}\left(\frac{x}{t}-D\right)^2\tilde P^ae^{-itH}f)\\
&\approx \sum_{2\le |a|\le N}(f,e^{itH}\left(\frac{x}{t}-D\right)^2\chi_{\{|x^a|\le R\}}\tilde P^ae^{-itH}f)\quad(R\gg 1)\\
&=\sum_{2\le |a|\le N}(f,e^{itH}\left(\frac{x_a}{t}-D_a\right)^2\chi_{\{|x^a|\le R\}}\tilde P^ae^{-itH}f)\\
&+\sum_{2\le |a|\le N}(f,e^{itH}\left(\frac{(x^a)^2}{t^2}-\frac{2A^a}{t}+2H_0^a\right)\chi_{\{|x^a|\le R\}}\tilde P^ae^{-itH}f)\\
&\sim \sum_{2\le |a|\le N}(f,e^{itH}\left(\frac{x_a}{t}-D_a\right)^2\chi_{\{|x^a|\le R\}}\tilde P^ae^{-itH}f)\\
&+2\sum_{2\le |a|\le N}(f,e^{itH}H_0^a\chi_{\{|x^a|\le R\}}\tilde P^ae^{-itH}f)
\end{aligned}
\end{equation}
%
asymptotically as $t\to\infty$. Taking $R\gg1$ large and removing the factor $\chi_{\{|x^a|\le R\}}$, we get
\begin{equation}
\begin{aligned}
&\sum_{2\le |a|\le N}(f,e^{itH}\left(\frac{x}{t}-D\right)^2\tilde P^ae^{-itH}f)\\
&\sim \sum_{2\le |a|\le N}(f,e^{itH}\left(\frac{x_a}{t}-D_a\right)^2\tilde P^ae^{-itH}f)+2\sum_{2\le |a|\le N}(f,e^{itH}H_0^a\tilde P^ae^{-itH}f).
\end{aligned}
\end{equation}
Comparing this with \eq{calculationsseparatingchannels2} and using Lemma \ref{RAGE}, we get asymptotically as $t\to\infty$
\begin{equation}
\begin{aligned}
&\sum_{2\le |a|\le N}(f,e^{itH}\left(\frac{x_a}{t}-D_a\right)^2\tilde P^ae^{-itH}f)\\
&\sim_{\epsilon_M}
\sum_{1\le|a|\le N}\left[-\frac{4}{t^2}\int_0^ts(f,e^{isH}T_a\tilde P^ae^{-isH}f)ds+2(f,e^{itH}T_a\tilde P^ae^{-itH}f)\right]\\
&\sim_{\epsilon_M}
\sum_{2\le|a|\le N}\left[-\frac{4}{t^2}\int_0^ts(f,e^{isH}(\tilde P^a)^*T_a\tilde P^ae^{-isH}f)ds+2(f,e^{itH}(\tilde P^a)^*T_a\tilde P^ae^{-itH}f)\right]\\
&=t\frac{dH}{dt}(t),
\end{aligned}\label{nexttofinal}
\end{equation}
where 
\begin{equation}
H(t)=\frac{2}{t^2}\int_0^ts(f,e^{isH}(\tilde P^a)^*T_a\tilde P^ae^{-isH}f)ds.\label{H(t)}
\end{equation}
The following is Lemma 8.15 of \cite{[Enss4]}. For the completeness' sake we will reproduce the proof.
\begin{lemma}\label{tdHdt}
Let $H(t)\in C^1((0,\infty))$ be a real-valued bounded function with
\begin{equation}
\lim_{t\to\infty}|H'(t)|=0.\label{(1)}
\end{equation}
Then for any $0< A<\infty$ there is a sequence $T_k\to\infty$ ($k\to\infty$) such that
\begin{equation}
\lim_{k\to\infty}\frac{1}{A}\int_{T_k}^{T_k+A}tH'(t)dt=0.\label{(2)}
\end{equation}
\end{lemma}
\begin{proof}
Assume the contrary. Then there exist constants $A>0$, $\epsilon>0$ and  $J(\epsilon,A)>0$ such that for all $T\ge J(\epsilon,A)$
\begin{equation}
\frac{1}{A}\int_{T}^{T+A}tH'(t)dt > 2\epsilon \ \ (\mbox{or }<-2\epsilon),\label{absurdassumptio}
\end{equation}
since the LHS is continuous in $T$.
For any interval $[T,T+A]$ and $t\in[T,T+A]$, one can decompose
\begin{equation}
H'(t)=H_1(t;T)+H_2(t;T),\label{decompose}
\end{equation}
where
\begin{equation}
\begin{aligned}
&|H_1(t;T)|\le |H'(t)|,\\
&\int_T^{T+A}H_1(t;T)dt=0,\\
&
H_2(t;T)=0 \mbox{ if  } H(T+A)-H(T)=0,\\
&
\mbox{sign }\left\{(H(T+A)-H(T))\cdot H_2(t;T)\right\}\ge0 \mbox{ otherwise.}
\end{aligned}
\end{equation}
It follows that
\begin{equation}
\begin{aligned}
&H(T+A)-H(T)=\int_T^{T+A}H_2(t;T)dt,\\
&
\left|\frac{1}{A}\int_T^{T+A}tH_1(t;T)dt\right|=\left|\frac{1}{A}\int_T^{T+A}(t-T)H_1(t;T)dt\right|\\
&\le A\sup_{T\le t\le T+A}\left|H_1(t;T)\right|\le A\sup_{T\le t}|H'(t)|\to0
\end{aligned}
\end{equation}
as $T\to\infty$ by \eq{(1)}. Thus one has by \eq{absurdassumptio} and \eq{decompose} that for all sufficiently large $T$
\begin{equation}
\frac{1}{A}\int_T^{T+A}tH_2(t;T)dt>\epsilon,
\end{equation}
which implies in particular $H_2(t;T)\ge0$ and
\begin{equation}
\frac{A+T}{A}[H(T+A)-H(T)]= \frac{1}{A}\int_{T}^{T+A}(T+A)H_2(t;T)dt\ge\frac{1}{A}\int_T^{T+A}tH_2(t;T)dt>\epsilon.
\end{equation}
Thus for sufficiently large $T$,
\begin{equation}
\begin{aligned}
&H(T+A)-H(T)>\epsilon \frac{A}{T+A},\\
& H(nA)-H(mA)=\sum_{k=m+1}^n \{H(kA)-H((k-1)A)\}>\epsilon\sum_{k=m+1}^n\frac{1}{k}.
\end{aligned}
\end{equation}
For any $m$ this diverges as $n\to\infty$, in contradiction to the boundedness of $H$.
\end{proof}

\noindent
As $H(t)$ in \eq{H(t)} satisfies the conditions of the lemma, the relation \eq{nexttofinal}, Lemmas \ref{RAGE} and \ref{tdHdt} imply that there exist sequences $A_k\to\infty$ and $T_k\to\infty$ ($k\to\infty$) such that as $k\to\infty$
\begin{equation}
\begin{aligned}
\hskip-2pt&\frac{1}{A_k}\hskip-2pt\int_{T_k}^{T_k+A_k}\hskip-2pt\left[\sum_{2\le |a|\le N}\hskip-0pt\left\Vert \left(\frac{x_a}{t}-D_a\right)\tilde P^ae^{-itH}f\right\Vert^2\hskip-2pt+\hskip-2pt
\sum_{(ij)\not\le a}\hskip-0pt\Vert \chi_{\{|x_{ij}|<R\}}\tilde P^ae^{-itH}f\Vert^2\hskip-0pt\right]\hskip-2ptdt\hskip-0pt\\
\hskip-2pt&\sim_{\epsilon_M}0.
\end{aligned}\label{timemeanofsumof}
\end{equation}
The proof of Theorem \ref{quantumclassical} is complete.
\end{proof}
The following theorem will be used in section \ref{asymptoticdecompositionofcontinuousspectralsubspacebyscatteringspaces}.
\begin{theorem}\label{chidonPdMd} Let Assumptions \ref{potentialdecay} and \ref{eigenfunctiondecay} be satisfied. For a cluster decomposition $a$ with $2\le|a|\le N$, let $\tilde P^a_{M_a}$ be defined by \eq{1.20}. Let $f=E_H(B)f$ for $B\Subset\R\setminus \mathcal{T}$. Then there exist constants $0<d_1<d_2<\infty$ such that for the sequences $t_m\to\infty$ and $M_a^m$ in Theorem \ref{quantumclassical} depending on $f$
\begin{equation}
{\tilde P}^a_{M_a^m}e^{-it_mH}f
\sim\chi_a(D_a){\tilde P}^a_{M_a^m}e^{-it_mH}f\label{2.11-corollary}
\end{equation}
as $t=t_m\to\infty$. Here $\chi_a(\xi_a)\in C_0^\infty(X_a)$ such that
\begin{equation}
\begin{aligned}
&0\le \chi_a(\xi_a)\le 1,\\
&\chi_a(\xi_a)=
\left\{
\begin{array}{ll}
1,&\quad 0<d_{1}\le|\xi_a|\le d_{2},\\
0,&\quad |\xi_a|\le d_{1}/2\mbox{ or }|\xi_a|\ge 2d_{2}.
\end{array}
\right.
\end{aligned}
\label{chid}
\end{equation}
\end{theorem}
\begin{proof} 
Let $\phi\in C_0^\infty(\R\setminus \mathcal{T})$ satisfy $\phi(\lambda)=1$ for $\lambda\in B(\Subset\R\setminus\mathcal{T})$. Then by \eq{summationofprojections}
\begin{equation}
\begin{aligned}
\sum_{2\le|a|\le N}
{\tilde P}^a_{M_a}e^{-itH}f
=e^{-itH}f=\phi(H)e^{-itH}f=\sum_{2\le|a|\le N}\phi(H){\tilde P}^a_{M_a}e^{-itH}f.
\end{aligned}
\label{EHBbunkai}
\end{equation}
Theorem \ref{quantumclassical} gives that for any $R>0$ and $(i,j)\not\le a$
\begin{align}
\Vert\chi_{\{x||x_{ij}|<R\}}{\tilde P}^a_{M_a^m}e^{-it_mH}f\Vert\to0\nonumber
\end{align}
as $m\to\infty$. This and \eq{EHBbunkai} imply
\begin{align}
\sum_{2\le|a|\le N}(I-\phi(H_a)){\tilde P}^a_{M_a^m}e^{-it_mH}f\to0\label{asymp}
\end{align}
as $m\to\infty$. As
the summands on the LHS are mutually orthogonal for different $a$'s asymptotically as $m\to\infty$, we obtain
\begin{align}
(I-\phi(H_a)){\tilde P}^a_{M_a^m}e^{-it_mH}f\to0\quad(m\to\infty)\label{asymp-a}
\end{align}
for each $a$ with $2\le |a|\le N$.
Since $H_a=T_a+H^a$ and $H^aP_j^a=\lambda_j$ for the eigenprojection $P_j^a$ of $H^a$ corresponding to eigenvalue $\lambda_j\in \mathcal{T}$ of $H^a$, supp $\phi\Subset \R\setminus \mathcal{T}$ implies that there exist constants $d_2>d_1>0$ such that $\chi_a(D_a)\phi(H_a){\tilde P}^a_{M_a^m}=\phi(H_a){\tilde P}^a_{M_a^m}$ for the function $\chi_a$ in \eq{chid}. Thus multiplying both sides of \eq{asymp-a} by $I-\chi_a(D_a)$ we obtain
\begin{align}
(I-\chi_a(D_a)){\tilde P}^a_{M_a^m}e^{-it_mH}f\to0\label{asymp-b}
\end{align}
as $m\to\infty$.
\end{proof}

\section{A partition of unity}\label{partitionofunity}

To state an important proposition, we prepare some notation. 
Let $a$ be a cluster decomposition with $2\le|a|\le N$. For any two clusters $C_1$ and $C_2$ in $a$, we define a vector $z_{a1}\in \R^\nu$ that connects the two centers of mass of the clusters $C_1$ and $C_2$. The number of such vectors when we move over all pairs $C_i$, $C_j$ ($i\ne j$) of clusters in $a$ is $k_a=\left(\smallmatrix |a|\\ 2\endsmallmatrix\right)$ in total. We denote these vectors by $z_{a1},z_{a2},\dots,z_{ak_a}$.

 Let $z_{ak}$ $(1\le k\le k_a)$ connect two clusters $C_\ell$ and $C_m$ in $a$ $(\ell\ne m)$. Then for any pair ${(i,j)}$ with $i\in C_\ell$ and $j\in C_m$, the vector $x_{ij}=r_i-r_j\in\R^\nu$ is expressed as a sum of $z_{ak}$ and two position vectors $w_1\in\R^\nu$ and $w_2\in\R^\nu$ of the particles $i$ and $j$ in $C_\ell$ and $C_m$, respectively. In particular we have 
\begin{align}
|x_{ij}|\ge|z_{ak}|-(|w_1|+|w_2|)\ge |z_{ak}|-2|x^a|.\label{fundamentalinequality}
\end{align}

Next if $c<a$ and $|c|=|a|+1$, then just one cluster, say $C_\ell\in a$, is decomposed into two clusters $C'_\ell$ and $C''_\ell$ in $c$, and other clusters in $a$ remain the same in the finer cluster decomposition $c$. In this case, we can choose just one vector $z_{ck}$ $(1\le k\le k_c)$ that connects clusters $C'_\ell$ and $C''_\ell$ in $c$, and we can express $x^a=(z_{ck},x^c)$. The norm of this vector is written as 
\begin{equation}
|x^a|^2=|z_{ck}|^2+|x^c|^2.\label{3.1}
\end{equation}
Similarly the norm of $x=(x_a,x^a)$ is written as
\begin{equation}
|x|^2=|x_a|^2+|x^a|^2.\label{3.2}
\end{equation}
 We recall that norm is defined from the inner product defined by \eq{innerproduct} which changes in accordance with the cluster decomposition used in each context. For instance, in \eq{3.1}, the left-hand side (LHS) is defined by using \eq{innerproduct} for the cluster decomposition $a$, and the right-hand side (RHS) is by using \eq{innerproduct} for $c$.

Given these notation, we state the following lemma, which is partly a repetition of Lemma 2.1 in \cite{[K1]} or section 3 of \cite{[Kitada-S]}. We define subsets $T_a(\rho,\theta)$ and ${\tilde T}_a(\rho,\theta)$ of $X=\R^{\nu n}$ for cluster decompositions $a$ with $2\le |a|\le N$ and real numbers $\rho,\theta$ with $1>\rho,\theta>0$.
\begin{equation}
\begin{aligned}
&T_a(\rho,\theta)=\left(\bigcap_{k=1}^{k_a}\{x\ |\ |z_{ak}|^2>\rho|x|^2\}\right)\cap\{ x\ |\ |x_a|^2>(1-\theta)|x|^2\},\\
&{\tilde T}_a(\rho,\theta)=\left(\bigcap_{k=1}^{k_a}\{x\ |\ |z_{ak}|^2>\rho\}\right)\cap\{ x\ |\ |x_a|^2>1-\theta\}. 
\end{aligned}\label{3.4}
\end{equation}
Subsets $S$ and $S_{\theta}$ $(\theta>0)$ of $X$ are defined by
\begin{equation}
\begin{aligned}
&S=\{x\ |x\in X, \ |x|^2\ge 1\},\\
&S_\theta=\{x\ |x\in X, \ 1+\theta\ge |x|^2\ge 1\}.
\end{aligned}\label{kyuu}
\end{equation}

\begin{lemma}\label{lemma31}
Suppose that constants $1\ge\theta_1>\rho_j>\theta_j> \rho_N>\theta_N>0$ satisfy $\theta_{j-1}\ge \theta_j+\rho_j$ for $j=2,3,\dots,N$. 
Then the following hold.
\begin{namelist}{8888}
\item[\ \ {\rm 1)}]
\begin{equation}
S\subset \bigcup_{2\le |a|\le N}T_a(\rho_{|a|},\theta_{|a|}).\label{3.5}
\end{equation}
\item[\ \ {\rm 2)}]
Let $\gamma_j>1$ $(j=1,2)$ satisfy
\begin{equation}
9\gamma_1\gamma_2< r_0:=\min_{2\le j\le N}\{{\rho_j/\theta_j}\}.\label{3.6}
\end{equation}
If $a\not\le c$ with $|a|\ge |c|$, then
\begin{equation}
T_a(\gamma_1^{-1}\rho_{|a|},\gamma_2\theta_{|a|})
\cap T_c(\gamma_1^{-1}\rho_{|c|},\gamma_2\theta_{|c|})=\emptyset.\label{3.7}
\end{equation}
\item[\ \ {\rm 3)}]
For $\gamma>1$ and $2\le|a|\le N$
\begin{equation}
\begin{aligned}
T_a(\rho_{|a|},\theta_{|a|})\cap S_{\theta_N}&\subset
{\tilde T}_a(\rho_{|a|},\theta_{|a|})\cap S_{\theta_N}\\
&\Subset {\tilde T}_a(\gamma^{-1}\rho_{|a|},\gamma\theta_{|a|})\cap S_{\theta_N}\\
&\subset T_a({\gamma'_1}^{-1}\rho_{|a|},\gamma'_2\theta_{|a|})\cap S_{\theta_N},
\end{aligned}
\label{3.8}
\end{equation}
where
\begin{equation}
\gamma'_1=\gamma(1+\theta_N),\quad \gamma'_2=(1+\gamma)(1+\theta_N)^{-1}.\label{3.9}
\end{equation}
\item[\ \ {\rm 4)}]
Let $\gamma>1$ satisfy $1<\gamma(1+\theta_N)<2$. Then taking $\rho_j>\theta_j>0$ ($j=2,$ $\dots,$ $N$) suitably, we have for $2\le|a|\le N$
\begin{equation}
{T}_a({\gamma'_1}^{-1}\rho_{|a|},\gamma'_2\theta_{|a|})\subset \{ x\ |\ |x_{ij}|^2> \rho_{|a|}|x|^2/2 \ \mbox{for all}\ {(i,j)}\not\leq a\}.\label{3.10}
\end{equation}
\item[\ \ {\rm 5)}]
If $9\gamma(1+\gamma)<r_0$ and $a\not\le c$ with $|a|\ge |c|$, then
\begin{equation}
T_a({\gamma'_1}^{-1}\rho_{|a|},\gamma'_2\theta_{|a|})\cap
T_c({\gamma'_1}^{-1}\rho_{|c|},\gamma'_2\theta_{|c|})=\emptyset.\label{3.11}
\end{equation}
\end{namelist}
\end{lemma}

\begin{proof} To prove \eq{3.5}, suppose that $|x|^2\ge 1$ and $x$ does not belong to the set
$$
A=\bigcup_{2\le |a|\le N-1}\left[\left(\bigcap_{k=1}^{k_a}\{x\ |\ |z_{ak}|^2>\rho_{|a|}|x|^2\}\right)\cap\{ x\ |\ |x_a|^2>(1-\theta_{|a|})|x|^2\}\right].
$$
Under this assumption, we prove $|x_{ij}|^2>\rho_N|x|^2$ for all pairs ${(i,j)}$. (Note that $z_{ak}$ for $|a|=N$ equals some $x_{ij}$.) Let $|a|=2$ and write $x=(z_{a1},x^a)$. Then by \eq{3.1}, $1\le|x|^2=|z_{a1}|^2+|x^a|^2$. Since $x$ belongs to the complement $A^c$ of the set $A$, we have $|z_{a1}|^2\le \rho_{|a|}|x|^2$ or $|x_a|^2\le(1-\theta_{|a|})|x|^2$. If $|z_{a1}|^2\le \rho_{|a|}|x|^2$, then $|x^a|^2=|x|^2-|z_{a1}|^2\ge(1-\rho_{|a|})|x|^2\ge(\theta_1-\rho_{|a|})|x|^2\ge \theta_{|a|}|x|^2$ by $\theta_{j-1}\ge \theta_j+\rho_j$. Thus $|x_a|^2=|x|^2-|x^a|^2\le(1-\theta_{|a|})|x|^2$ for all $a$ with $|a|=2$.

Next let $|c|=3$ and assume $|x_c|^2>(1-\theta_{|c|})|x|^2$. Then by $x\in A^c$, we can choose $z_{ck}$ with $1\le k\le k_c$ such that $|z_{ck}|^2\le \rho_{|c|}|x|^2$. Let $C_\ell$ and $C_m$ be two clusters in $c$ connected by $z_{ck}$, and let $a$ be the cluster decomposition obtained by combining $C_\ell$ and $C_m$ into one cluster with retaining other clusters of $c$ in $a$. Then $|a|=2$, $x^a=(z_{ck},x^c)$, and $|x^a|^2=|z_{ck}|^2+|x^c|^2$. Thus $|x_a|^2=|x|^2-|x^a|^2=|x|^2-|z_{ck}|^2-|x^c|^2=|x_c|^2-|z_{ck}|^2>(1-\theta_{|c|}-\rho_{|c|})|x|^2\ge(1-\theta_{|a|})|x|^2$, which contradicts the result of the previous step. Thus $|x_c|^2\le(1- \theta_{|c|})|x|^2$ for all $c$ with $|c|=3$. 

Repeating this procedure, we finally arrive at $|x_d|^2\le (1-\theta_{|d|})|x|^2$, thus $|x^d|^2=|x|^2-|x_d|^2\ge \theta_{|d|}|x|^2> \rho_N|x|^2$ for all $d$ with $|d|=N-1$. Namely $|x_{ij}|^2>\rho_N|x|^2$ for all pairs ${(i,j)}$.
The proof of \eq{3.5} is complete.

We next prove \eq{3.7}. By $a\not\le c$, we can take a pair ${(i,j)}$ and clusters $C_\ell, C_m\in c$ such that ${(i,j)}\le a$, $i\in C_\ell$, $j\in C_m$, and $\ell\ne m$. Then $C_\ell$ and $C_m$ are connected by $z_{ck}$ for some $1\le k\le k_c$. Thus if there is an $x\in T_a(\gamma_1^{-1}\rho_{|a|},\gamma_2\theta_{|a|})\cap T_c(\gamma_1^{-1}\rho_{|c|},\gamma_2\theta_{|c|})$, then by \eq{fundamentalinequality}
\begin{equation}
\begin{aligned}
(\gamma_2\theta_{|a|})^{1/2}|x|&>|x^a|\ge|x_{ij}|\ge |z_{ck}|-2|x^c|\\
&>((\gamma_1^{-1}\rho_{|c|})^{1/2}-2(\gamma_2\theta_{|c|})^{1/2})|x|.
\end{aligned}\label{3.12}
\end{equation}
Thus $(\gamma_2\theta_{|a|})^{1/2}>(\gamma_1^{-1}\rho_{|c|})^{1/2}-2(\gamma_2\theta_{|c|})^{1/2}$. As $\theta_{|a|}\le\theta_{|c|}$ by $|a|\ge|c|$, we have $9\gamma_2\theta_{|c|}>\gamma_1^{-1}\rho_{|c|}$, contradicting \eq{3.6}. This completes the proof of \eq{3.7}.

\eq{3.8} follows by a simple calculation from the inequality $|x|^2(1+\theta_N)^{-1}\le 1$ that holds on $S_{\theta_N}$. \eq{3.10} follows from the inequality \eq{fundamentalinequality} as follows. Let $x\in T_a(\gamma_1'^{-1}\rho_{|a|},\gamma'_2\theta_{|a|})$ and ${(i,j)}\not\le a$. Then there are clusters $C_\ell, C_m\in a$ $(\ell\ne m)$ such that $i\in C_\ell$, $j\in C_m$ and $C_\ell,C_m$ are connected by $z_{ak}$ $(1\le k\le k_a)$. Then by \eq{fundamentalinequality}, $|x_{ij}|\ge|z_{ak}|-2|x^a|>\rho_{|a|}^{1/2}(\gamma_1'^{-1/2}-2\gamma'^{1/2}_2(\theta_{|a|}/\rho_{|a|})^{1/2})|x|$. From this and the definition of $r_0$, we have $|x_{ij}|/|x|>\rho_{|a|}^{1/2}\gamma_1'^{-1/2}(1-2(\gamma(1+\gamma))^{1/2}r_0^{-1/2})$. By taking $\rho_j>\theta_j>0$ suitably, we can choose constants $1<\alpha<2$ and $M>0$ such that $1<\gamma'_1=\gamma(1+\theta_N)<\alpha$, $1<\gamma(1+\gamma)<M^{-1}r_0$ and $(\alpha/2)^{1/2}+2M^{-1/2}<1$ hold. Then we have $|x_{ij}|/|x|>\rho_{|a|}^{1/2}\alpha^{-1/2}(1-2M^{-1/2})>(\rho_{|a|}/2)^{1/2}$. Namely $|x_{ij}|^2>\rho_{|a|}|x|^2/2$.
 \eq{3.11} follows from $\gamma'_1\gamma'_2=\gamma(1+\gamma)$ and 2).
\end{proof}

In the following we fix constants  $\gamma>1$ and $1/4\ge \theta_1>\rho_j>\theta_j>\rho_N> \theta_N>0$ such that
\begin{equation}
\begin{aligned}
&\theta_{j-1}\ge \theta_j+\rho_j\quad(j=2,3,\dots,N),\\
&9\gamma(1+\gamma)<r_0=\min_{2\le j\le N}\{\rho_j/\theta_j\},\quad 1<\gamma<2(1+\theta_N)^{-1}.
\end{aligned}\label{3.14}
\end{equation}
We recall \eq{3.9}.
\begin{equation}
\gamma'_1=\gamma(1+\theta_N),\quad \gamma'_2=(1+\gamma)(1+\theta_N)^{-1}.
\end{equation}

Let a function $\rho({t})\in C^\infty({{\mathbb R}})$ satisfy the following.
\begin{equation}
\begin{aligned}
&0 \le \rho ({t} ) \le 1,\\
&\rho ({t} ) =
\left\{
\begin{array}{ll}
 1&\quad ({t} \le -1)\\
 0&\quad({t} \ge 0)
\end{array}
\right.
\\
&\rho^{\prime} ({t} )=\frac{d\rho}{dt}(t) \le 0,\\
&\vert \rho^{\prime}({t} ) \vert ^{\frac{1}{2}} \in C^\infty ({{\mathbb R}}),
\end{aligned}\label{rhodef}
\end{equation}
and a functions $\psi_\sigma({t}>\tau)$ of ${t}\in \R$ be defined by
\begin{align}
\psi_\sigma({t}>\tau)=1-\rho(({t}-\tau)/\sigma)
\label{3.15}
\end{align}
for constants $\sigma\in(0,\rho_N/4)$ and $\tau\in \R$.
We note that $\psi_\sigma({t}>\tau)$ satisfies
\begin{equation}
\begin{aligned}
&\psi_\sigma({t}>\tau)=
\begin{cases}
0 \quad &({t}\le\tau-\sigma)\\
1 \quad& ({t}\ge\tau)
\end{cases}
\\
&\psi'_\sigma({t}>\tau)=\frac{d}{d{t}}\psi_\sigma({t}>\tau)\ge 0.
\end{aligned}\label{3.17}
\end{equation}
We define for a cluster decomposition $a$ with $2\le |a|\le N$ 
\begin{equation}
\varphi_a(x_a)=\prod_{k=1}^{k_a}\psi^2_\sigma(|z_{ak}|^2>\rho_{|a|})\psi^2_\sigma(|x_a|^2>1-\theta_{|a|}),\label{3.18}
\end{equation}
where $\sigma>0$ is fixed as
\begin{equation}
0<\sigma<\min_{2\le j\le N}\{(1-\gamma^{-1})\rho_j,(\gamma-1)\theta_j\}(<\theta_N).\label{3.19}
\end{equation}
Then $\varphi_a(x_a)$ satisfies for $x\in S_{\theta_N}$
\begin{equation}
\varphi_a(x_a)=\left\{
\begin{aligned}
&1\quad \text{for}\quad x\in {\tilde T}_a(\rho_{|a|},\theta_{|a|}),\\
&0\quad \text{for}\quad x\not\in {\tilde T}_a(\gamma^{-1}\rho_{|a|},\gamma\theta_{|a|}).
\end{aligned}\right.\label{3.20}
\end{equation}
We set for $x\in S_{\theta_N}$ and $|a|=k$ $(k=2,3,\dots,N)$
\begin{equation}
\Psi_a(x)=\varphi_a(x_a)
\left(1-\sum_{|a_{k-1}|=k-1}\varphi_{a_{k-1}}(x_{a_{k-1}})\right)
\dots
\left(1-\sum_{|a_{2}|=2}\varphi_{a_{2}}(x_{a_{2}})\right).\label{3.21}
\end{equation}
By 3) and 5) of Lemma \ref{lemma31} and \eq{3.20}, the sums on the RHS remain only in the case $a<a_j$ for $j=k-1,\dots,2$ and $x\in S_{\theta_N}$.
\begin{equation}
\Psi_a(x)=\varphi_a(x_a)
\left(1-\sum_{\substack{|a_{k-1}|=k-1\\a<a_{k-1}}}\varphi_{a_{k-1}}(x_{a_{k-1}})\right)
\dots
\left(1-\sum_{\substack{|a_{2}|=2\\{a<a_2}}}\varphi_{a_{2}}(x_{a_{2}})\right).\label{3.22}
\end{equation}
Thus $\Psi_a(x)$ is a function of the variable $x_a$ only.
\begin{equation}
\Psi_a(x)=\Psi_a(x_a)\quad\text{when}\quad x=(x_a,x^a)\in S_{\theta_N}.\label{3.23}
\end{equation}
We also note that the supports of $\varphi_{a_j}$ in each sum on the RHS of 
\eq{3.21} are disjoint mutually in $S_{\theta_N}$ by 3) and 5) of Lemma \ref{lemma31}.
By \eq{3.5} and \eq{3.8} of Lemma \ref{lemma31}, and the definition \eq{3.18}-\eq{3.21} of $\Psi_a(x_a)$, we therefore have
$$
\sum_{2\le|a|\le N}\Psi_a(x)=\sum_{2\le|a|\le N}\Psi_a(x_a)=1\quad\text{on}\quad S_{\theta_N}.
$$

We have constructed a partition of unity on $S_{\theta_N}$.

\begin{proposition}\label{partitionupropo} Let real numbers $1/4\ge\theta_1>\rho_j>\theta_j>\rho_N>\theta_N>4\sigma>0$ satisfy $\theta_{j-1}\ge \theta_j+\rho_j$ for $j=2,3,\dots,N$. 
Assume that \eq{3.14} hold and let $\Psi_a(x_a)$ be defined by \eq{3.18}-\eq{3.22}. Then we have
\begin{equation}
\sum_{2\le|a|\le N}\Psi_a(x)=\sum_{2\le|a|\le N}\Psi_a(x_a)=1\quad\text{on}\quad S_{\theta_N}.\label{3.24}
\end{equation}
$\Psi_a(x_a)$ is a $C^\infty$ function of $x_a$ and satisfies $0\le \Psi_a(x_a)\le 1$. Further for $1\le k\le k_a$ and $x\in\hskip3pt\mbox{\rm supp}\hskip2pt\Psi_a \cap S_{\theta_N}$
\begin{equation}
|z_{ak}|^2>\frac{3}{4}\rho_{|a|}|x|^2,\quad |x_a|^2\ge\omega_{|a|}|x^a|^2,\ \ (\omega_{|a|}=(\theta_{|a|}+\theta_N+\sigma)^{-1}(1-\theta_{|a|}-\theta_N-\sigma)),\label{3.25}
\end{equation}
and
\begin{equation}
\sup_{\substack{x\in X\\2\le|a|\le N}}|\nabla_{x_a} \Psi_a(x_a)|<\infty.\label{3.26}
\end{equation}
\end{proposition}

Let $\Psi_a(x)=\Psi_a(x_a)$ be extended to $X\setminus \{0\}$ as follows. For $(1+\theta_N)^\ell\le |x|^2< (1+\theta_N)^{\ell+1}$ $(\ell=0,\pm1,\pm2,\dots)$, define
\begin{equation}
\Psi_a(x)=\Psi_a((1+\theta_N)^{-\ell/2}x).\label{dyadical}
\end{equation}
This function satisfies
\begin{equation}
\begin{aligned}
&0\le \Psi_a(x)\le 1,\\
&\Psi_a(x)=\Psi_a(x_a) \ \ \mbox{for}\ \ x=(x_a,x^a)\in X\setminus\{0\}.
\end{aligned}\label{extendedPsi}
\end{equation}
Let $\lambda\in C_0^\infty(X)$ be a real-valued nonnegative function such that
\begin{equation}
\int_{X}\lambda(y)dy=1\label{lambda}
\end{equation}
and set
\begin{equation}
\Phi_a(x)=(\Psi_a*\lambda)(x)=\int_{X}\Psi_a(y)\lambda(x-y)dy.\label{Phia}
\end{equation}
Then
\begin{equation}
\begin{aligned}
&0\le \Phi_a(x)\le 1,\\
&\Phi_a(x)\in C^\infty(X),\\
&\sum_{2\le|a|\le N}\Phi_a(x)=1.
\end{aligned}
\end{equation}
Further by \eq{extendedPsi} we have for $x\in X$
\begin{equation}
\begin{aligned}
\Phi_a(x)
&=\int_{X}\Psi_a(y_a,y^a)\lambda(x_a-y_a,x^a-y^a)dy\\
&=\int_{X}\Psi_a(y_a)\lambda(x_a-y_a,x^a-y^a)dy\\
&=\int_{X_a}\Psi_a(y_a)\int_{X^a}\lambda(x_a-y_a,x^a-y^a)dy^ady_a\\
&=\int_{X_a}\Psi_a(y_a)\int_{X^a}\lambda(x_a-y_a,-y^a)dy^ady_a\\
&=\int_{X}\Psi_a(y_a)\lambda(x_a-y_a,-y^a)dy\\
&=\Phi_a(x_a,0).
\end{aligned}\label{simpleness}
\end{equation}
Thus $\Phi_a(x)$ depends only on $x_a$ and we can write $\Phi_a(x)=\Phi_a(x_a)$.
Summarizing we have proved the following theorem. We recall that $S=\{x|x\in X,|x|^2\ge1\}$.
\begin{theorem}\label{homogeneousextension}
Let $\Phi_a(x)\in C^\infty(X)$ for $2\le |a|\le N$ be the function defined by \eq{Phia}. Then $\Phi_a$ satisfies the following conditions.
\begin{equation}
\begin{aligned}
&0\le \Phi_a(x)\le 1,\\
&\Phi_a(x)\in C^\infty(X),\\
&\sum_{2\le|a|\le N}\Phi_a(x)=1\\
&\Phi_a(x)=\Phi_a(x_a)\quad(x=(x_a,x^a)\in X).
\end{aligned}\label{partition-of-unity-2}
\end{equation}
Further for any integer $L\ge1$ we can take the support of the function $\lambda\in C_0^\infty(X)$ so small that
for $1\le k\le k_a$ we have on {\rm supp}\hskip2pt$\Phi_a\cap(1+\theta_N)^{-L/2}S$
\begin{equation}
|z_{ak}|^2>\rho_{|a|}|x|^2/2,\quad |x_a|^2\ge\omega_{|a|}|x^a|^2/2.\label{3.25-2}
\end{equation}
In particular taking $\rho_j>\theta_j$ such that $\rho_j/\theta_j>2^{10}$, we have by \eq{fundamentalinequality}
\begin{align}
|x_{ij}|>\rho_{|a|}^{1/2}|x|/4\quad(\forall (i,j)\not\le a).\label{4.36}
\end{align}
\end{theorem}

\section{Classical scattering}\label{scatteringtrajectories}

We consider the behavior of the classical trajectory when the initial condition $(y,\eta)$ satisfies ${\lambda}>0$ and $H(y,\eta)=\lambda$.
In this case if $y$ can be taken such that $V(y)<\lambda$, then the corresponding $\eta$ satisfies $|\eta|^2>0$. In considering the $N$-particle case, it is further necessary to consider the condition $|\eta_a|>d>0$ for $d>0$ and a cluster decomposition $a$ with $2\le |a|\le N$. The investigation of classical trajectories developed in this section will be important when we consider the quantum mechanical case. As the case $t<0$ is treated similarly, we will consider the case $t>0$ only throughout the rest of the paper.

As we will use Hamilton's canonical equation of motion in this section, we will assume that $V^S_{ij}(x)=0$ $(x\in \R^\nu)$ for all pairs ${(i,j)}$ with $1\le i<j\le N$. So that our pair potential $V_{ij}(x)=V_{ij}^L(x)$ will be assumed to be continuously differentiable with respect to $x\in \R^\nu$ and satisfy the condition \eq{1.3}.

Let $\rho({t})\in C^\infty({{\mathbb R}})$ be a function defined by \eq{rhodef}.
We define for ${t} \in {{\mathbb R}}$ and ${\varepsilon} > 0$
\begin{equation}
\begin{aligned}
&\phi({t} < {\varepsilon}) = \rho (({t} - 2{\varepsilon} )/{\varepsilon} )=\left\{
\begin{array}{ll}
1,&\quad t\le {\varepsilon},\\
0,&\quad t\ge 2{\varepsilon},
\end{array}\right.\\
&\phi'({t} < {\varepsilon})=\frac{d}{d{t}}\phi({t} < {\varepsilon})
\left\{
\begin{array}{ll}
=0,&\quad t\in(-\infty,{\varepsilon}]\cup[2{\varepsilon},\infty),\\
\le 0,&\quad t\in[{\varepsilon},2{\varepsilon}].
\end{array}\right.
\end{aligned}\label{phi1def}
\end{equation}
Letting $0<d_1<d_2<\infty$, we define $\chi_a(\xi_a)\in C_0^\infty(X'_a)$ by \eq{chid}, i.e.
\begin{equation}
\begin{aligned}
&0\le \chi_a(\xi_a)\le 1,\\
&\chi_a(\xi_a)=
\left\{
\begin{array}{ll}
1,&\quad 0<d_1\le|\xi_a|\le d_2,\\
0,&\quad |\xi_a|\le d_1/2\mbox{ or }|\xi_a|\ge 2d_2.
\end{array}
\right.
\end{aligned}\label{chia}
\end{equation}
As in Theorem \ref{homogeneousextension}, let $\Phi_a(\xi)\in C^\infty(X')$ be the partition of unity defined from the function $\Psi_a(\xi)$ in \eq{3.22}.
\begin{equation}
\sum_{2\le|a|\le N}\Phi_a(\xi_a)=1 \mbox{  for  }\xi\in X'.\label{partition-of-unity}
\end{equation}
Let ${(i,j)}\not\le a$ such that $C_\ell, C_m\in a$ $(\ell\ne m)$, $i\in C_\ell,j\in C_m$, and let $z_{ak}$ $(1\le k\le k_a)$ connect two centers of mass of those clusters $C_\ell$ and $C_m$. 
By \eq{3.25-2}-\eq{4.36} there exists a suitable choice of constants
\begin{equation}
1/4\ge \theta_1>\rho_j>\theta_j>\rho_N>\theta_N>4\sigma>0, \quad \theta_{j-1}\ge \theta_j+\rho_j\ \ (j=2,3,\dots,N)\label{thetarho}
\end{equation}
such that for $\lambda(\xi) $ in \eq{lambda} with sufficiently small support one has on {\rm supp}\hskip2pt$(\chi_a^2\Phi_a)(\xi)$
\begin{equation}
|\zeta_{ak}|^2 >\rho_{|a|}|\xi|^2/2,\quad |\xi_{ij}|^2>\rho_{|a|}|\xi|^2/16.\label{pair}
\end{equation}
Here $\zeta_{ak}$ and $\xi_{ij}$ are the variables conjugate to $z_{ak}$ and $x_{ij}$, respectively.

We now introduce a localizing function $p_{a}^{\varepsilon}(t,x, \xi_a)$ in the extended phase space $(0,\infty)\times X\times (X'\setminus\{0\})$ for $|a|\ge2$,  ${\varepsilon}>0$, $(t,x,\xi)\in (0,\infty)\times X\times (X'\setminus\{0\})$
\begin{eqnarray}
p_{a}^{\varepsilon}(t,x, \xi_a)= \phi(|x/t - \xi_a|^2 < {\varepsilon})\chi_a(\xi_a)^2\Phi_a(\xi_a), \label{(4)}
\end{eqnarray}
where $\xi_a=(\xi_a,0)\in X\setminus\{0\}=\R^{\nu n}\setminus\{0\}$ and $|x/t - \xi_a|=|x/t - (\xi_a,0)|$.
We note that $p_{a}^{\varepsilon}(t,x, \xi_a)=p_a^{\varepsilon}(t,x,\xi)=\phi(|x/t - \xi_a|^2 < {\varepsilon})\chi_a(\xi_a)^2\Phi_a(\xi)$ for all $(x,\xi)\in X\times (X'\setminus\{0\})$ by the property $\Phi_a(\xi_a)=\Phi_a(\xi)$.

We will use an extended micro-localizing pseudodifferential operator $P_a^{\varepsilon}(t)=p_a^{\varepsilon}(t,x,D_a)$ with symbol $p_a^{\varepsilon}(t,x,\xi_a)=p_a^{\varepsilon}(t,x,\xi)$.
\begin{equation}
\begin{aligned}
P_a^{\varepsilon}(t)f(x)&=p_a^{\varepsilon}(t,x,D_a)f(x)\\
&=(2\pi)^{-\nu (|a|-1)/2}\int_{\R^{\nu (|a|-1)}}e^{ix_a\xi_a}p_a^{\varepsilon}(t,x,\xi_a)\hat f(\xi_a,x^a)d\xi_a\\
&=(2\pi)^{-\nu n/2}\int_{\R^{\nu n}}e^{ix\xi}p_a^{\varepsilon}(t,x,\xi_a)\hat f(\xi)d\xi,
\end{aligned}\label{Paepsilont}
\end{equation}
where $f\in \HH=L^2(X)$, and $\mathcal{F}_af(\xi_a,x^a)=\hat f(\xi_a,x^a)$ and $\mathcal{F}f(\xi)=\hat f(\xi)=\hat f(\xi_a,\xi^a)$ are partial and full Fourier transform, respectively.
\begin{equation}
\begin{aligned}
&\mathcal{F}_af(\xi_a,x^a)=\hat f(\xi_a,x^a)=(2\pi)^{-\nu (|a|-1)/2}\int_{\R^{\nu (|a|-1)}}e^{-ix_a\xi_a}f(x_a,x^a)dx_a,\\
&\mathcal{F}f(\xi)=\hat f(\xi_a,\xi^a)=(2\pi)^{-\nu n/2}\int_{\R^{\nu n}}e^{-ix\xi}f(x)dx.
\end{aligned}\label{partialfullFourier}
\end{equation}

Let $(x(t),\xi(t))$ be a solution of Hamilton equation with initial condition $(y,\eta)$. 
\begin{equation}
\begin{aligned}
&x(t,s)= y+\int_s^t \partial_\xi H(x(\tau,s),\xi(\tau,s))d\tau= y+\int_s^t \xi(\tau,s)d\tau,\\
&\xi(t,s)= \eta-\int_s^t \partial_x H(x(\tau,s),\xi(\tau,s))d\tau=\eta-\int_s^t \partial_x V(x(\tau,s))d\tau.\label{Hamilton2}
\end{aligned}
\end{equation}
We calculate the following quantity.
\begin{eqnarray*}
\frac{d}{dt}(p_a^{\varepsilon}(t,x(t),\xi(t)))
&\hskip-8pt=\hskip-8pt&(\partial_t p_a^{\varepsilon}+\{H, p_a^{\varepsilon}\})(t,x(t),\xi(t))\\
&\hskip-8pt=\hskip-8pt&\partial_tp_a^{\varepsilon}(t,x(t),\xi(t))+\partial_xp_a^{\varepsilon}(t,x(t),\xi(t))\cdot\xi(t)\\
&&\hskip0pt-\partial_\xi p_a^{\varepsilon}(t,x(t),\xi(t))\cdot\partial_x V(x(t)),
\end{eqnarray*}
where $\{F,G\}(x,\xi)=(\partial_\xi F\cdot\partial_x G-\partial_x F \cdot\partial_\xi G)(x,\xi)$ is a Poisson bracket. 
We set
\begin{equation}
\begin{aligned}
u_a^{\varepsilon}(t,x,\xi)&=t(\partial_tp_a^{\varepsilon}(t,x,\xi)+\partial_x p_a^{\varepsilon}(t,x,\xi)\cdot\xi)\\
&=-2\phi'(|x/t-\xi_a|^2<{\varepsilon})\{|x/t-\xi_a|^2-(x^a\cdot\xi^a)/t\} \chi_a(\xi_a)^2\Phi_a(\xi).
\end{aligned}\label{utxxi}
\end{equation}
Then calculating directly we have
\begin{equation*}
(\partial_tp_a^{\varepsilon}+\{H, p_a^{\varepsilon}\})(t,x(t),\xi(t))=\frac{1}{t}u_a^{\varepsilon}(t,x(t),\xi(t))-\partial_\xi p_a^{\varepsilon}(t,x(t),\xi(t))\cdot\partial_xV(x(t)).
\end{equation*}
Here we take $\theta_{|a|}>\theta_N>0$ and $\sigma>0$ in \eq{thetarho} and \eq{3.19} so small that
\begin{equation}
\begin{aligned}
\sqrt{{\varepsilon}}&> 2\sqrt{2}d_2\omega_{|a|}^{-\frac{1}{2}}\\
&=2\sqrt{2}d_2(\theta_{|a|}+\theta_N+\sigma)^{\frac{1}{2}}(1-\theta_{|a|}-\theta_N-\sigma)^{-{\frac{1}{2}}}.
\end{aligned}\label{condition-1}
\end{equation}
Then by \eq{3.25-2} we have $|\xi^a|\le \omega_{|a|}^{-\frac{1}{2}}|\xi_a|\le 2d_2\omega_{|a|}^{-\frac{1}{2}}$ for $\xi\in\mbox{supp }\chi_a(\xi_a)^2\Phi_a(\xi)$, and $|x^a|/t\le \sqrt{2{\varepsilon}}$ for $x\in$ supp $\phi'(|x/t-\xi_a|^2<{\varepsilon})$. 
 Hence
 $|x/t-\xi_a|^2-(x^a\cdot\xi^a)/t\ge{\varepsilon}-|(x^a\cdot\xi^a)/t|
 \ge{\varepsilon}- 2\sqrt{2{\varepsilon}}d_2 \omega_{|a|}^{-\frac{1}{2}}>0$ for $(x,\xi)\in$ supp $u_a^{\varepsilon}(t,x,\xi)$ so that we have
\begin{equation}
u_a^{\varepsilon}(t,x,\xi)\ge0\label{upositive}
\end{equation}
for all $t>0$ and $(x,\xi)\in X\times (X'\setminus\{0\})$.
We set
\begin{equation}
\begin{aligned}
&\hskip-5ptq_a^{\varepsilon}(t,x,\xi)\\
&=u_a^{\varepsilon}(t,x,\xi)^{\frac{1}{2}}\\
&=(-2\phi'(|x/t-\xi_a|^2<{\varepsilon}))^{1/2}(|x/t-\xi_a|^2-(x^a\cdot\xi^a)/t)^{1/2}\chi_a(\xi_a)\Phi_a(\xi_a)^{1/2}\\
&\in C^\infty(X\times (X'\setminus\{0\})).
\end{aligned}\label{qepsilon}
\end{equation}
Then
\begin{equation}
\begin{aligned}
\partial_tp_a^{\varepsilon}(t,x,\xi)+\partial_x p_a^{\varepsilon}(t,x,\xi)\cdot\xi=\frac{1}{t}q_a^{\varepsilon}(t,x,\xi)^2
\end{aligned}\label{squareqavarepsilontxxi}
\end{equation}
and
\begin{equation}
\begin{aligned}
&(\partial_t p_a^{\varepsilon}(t,x,D)+i[H_0,p_a^{\varepsilon}(t,x,D)])f(x)\\
&=(2\pi)^{-\nu n/2}\hskip-4pt\int\hskip-2pt e^{ix\xi}\{\partial_tp_a^{\varepsilon}(t,x,\xi)+\partial_xp_a^{\varepsilon}(t,x,\xi)\cdot\xi-\frac {i}{2} \Delta_x p_a^{\varepsilon}(t,x,\xi)\}{\hat f}(\xi)d\xi\\
&=(2\pi)^{-\nu n/2}\hskip-4pt\int\hskip-2pt e^{ix\xi}\biggl\{\frac{1}{t}q_a^{\varepsilon}(t,x,\xi)^2-\frac {i}{2} \Delta_x p_a^{\varepsilon}(t,x,\xi)\biggr\}{\hat f}(\xi)d\xi.
\end{aligned}\label{A.4-2}
\end{equation}
The second term in the integrand satisfies
\begin{equation}
|\partial_x^\alpha\partial_\xi^\beta\Delta_x p_a^{\varepsilon}(t,x,\xi)|\le C_{\alpha\beta}t^{-2-|\alpha|}\label{estimate}
\end{equation}
for all $t>1$, $(x,\xi)\in X\times X'$.
We define $Q_a^{\varepsilon}(t)$ as the pseudodifferential operator with symbol $q_a^{\varepsilon}(t,y,\xi)$.
\begin{equation}
Q_a^{\varepsilon}(t)f(x)=(2\pi)^{-\nu n}\iint e^{i(x-y)\xi}q_a^{\varepsilon}(t,y,\xi)f(y)dyd\xi.\label{Qvarepsilon}
\end{equation}
Then similarly to the proof of Lemma 4.2 in \cite{[K1]} or Lemma 4.1 in \cite{KS} with using Calder\'on-Vaillancourt theorem \cite{CV2}, one obtains a proposition which will play an important role in proving the existence of wave operators and the related limits.
\begin{proposition}\label{positiveHeizenbergDerivative}
There exist norm continuous bounded operators $Q_a^{\varepsilon}(t)$ and $R_a^{\varepsilon}(t)$ for $t>1$ and sufficiently small ${\varepsilon}>0$ such that the following holds for some constant $C_{\varepsilon}>0$.
\begin{equation}
\begin{aligned}
&\partial_t P_a^{\varepsilon}(t)+i[H_0,P_a^{\varepsilon}(t)]
=\frac{1}{t}Q_a^{\varepsilon}(t)^*Q_a^{\varepsilon}(t)+R_a^{\varepsilon}(t),\\
&Q_a^{\varepsilon}(t)^*Q_a^{\varepsilon}(t)\ge0,\quad \Vert R_a^{\varepsilon}(t)\Vert\le C_{\varepsilon} t^{-2}.
\end{aligned}\label{positiveHeisenbergDerivative}
\end{equation}
\end{proposition}
\begin{proof}
By a direct calculation we obtain from \eq{Qvarepsilon}
\begin{align}
Q_a^{\varepsilon}(t)^*Q_a^{\varepsilon}(t)f(x)=(2\pi)^{-\nu n}\iint e^{i(x-y)\xi}q_a^{\varepsilon}(t,x,\xi)q_a^{\varepsilon}(t,y,\xi)f(y)dyd\xi.\label{q*q}
\end{align}
We note
\begin{align}
q_a^{\varepsilon}(t,y,\xi)=q_a^{\varepsilon}(t,x,\xi)+(y-x)\cdot
\tilde\nabla_xq_a^{\varepsilon}(t,x,\xi,y),
\end{align}
where
\begin{align}
\tilde\nabla_xq_a^{\varepsilon}(t,x,\xi,y)=
\int_0^1 \nabla_xq_a^{\varepsilon}(t,x+\theta(y-x),\xi)d\theta.
\end{align}
By integration by parts
\begin{equation}
\begin{aligned}
Q_a^{\varepsilon}(t)^*Q_a^{\varepsilon}(t)f(x)&=(2\pi)^{-\nu n}\iint e^{i(x-y)\xi}q_a^{\varepsilon}(t,x,\xi)^2f(y)dyd\xi\\
&+(2\pi)^{-\nu n}\iint e^{i(x-y)\xi}D_\xi\tilde\nabla_xq_a^{\varepsilon}(t,x,\xi,y)f(y)dyd\xi.
\end{aligned}\label{Q^*Q}
\end{equation}
From \eq{rhodef}, \eq{phi1def} and the definition \eq{qepsilon} of $q_a^{\varepsilon}(t,x,\xi)$ we obtain
\begin{align}
|\partial_x^\alpha\partial_\xi^\beta\partial_y^\gamma D_\xi\tilde\nabla_xq_a^{\varepsilon}(t,x,\xi,y)|\le C_{\alpha\beta\gamma}t^{-1-|\alpha|-|\gamma|}
\end{align}
for all $t>1$, $(x,\xi,y)\in \R^{-3\nu n}$ and all multi-indices $\alpha,\beta,\gamma$. 
 From this, \eq{A.4-2}, \eq{estimate}, \eq{Q^*Q} and Calder\'on-Vaillancourt theorem follows \eq{positiveHeisenbergDerivative}.
\end{proof}

Returning to classical context, we prepare a lemma.
\begin{lemma}\label{lemmaonzakxa}
Let $(x,\xi)\in \mbox{\rm supp }\phi(|x/t-\xi_a|^2<{\varepsilon})\chi_a(\xi_a)^2\Phi_a(\xi)$.
Let ${(i,j)}\not\le a$ such that $C_\ell, C_m\in a$ $(\ell\ne m)$, $i\in C_\ell,j\in C_m$, and let $z_{ak}$ $(1\le k\le k_a)$ connect those two clusters $C_\ell$ and $C_m$. Taking ${\varepsilon}>0$ sufficiently small, we set
\begin{equation}
\kappa_{|a|}=2^{-4}(2^{-1}\rho_{|a|})^{1/2}d_1\ge\sqrt{2{\varepsilon}}>0.\label{condition-2}
\end{equation}
Then we have
\begin{equation}
|z_{ak}|/t \ge 7\kappa_{|a|},\quad |x_{ij}|/t \ge5\kappa_{|a|}.\label{xbetazakgekappaa}
\end{equation}
Further we have
\begin{equation}
|x^a|/t\le\kappa_{|a|}.\label{xakappaa}
\end{equation}
\end{lemma}
\begin{proof} It follows from $(x,\xi)\in \mbox{\rm supp }(\phi(|x/t-\xi_a|^2<{\varepsilon})\chi_a(\xi_a)^2\Phi_a(\xi))$ that $|x/t-\xi_{a}|=(|x_{a}/t-\xi_{a}|^2+(|x^a|/t)^2)^{1/2}\le\sqrt{2{\varepsilon}}$. Thus we have $|x^a|/t\le \sqrt{2{\varepsilon}}\le \kappa_{|a|}$, which is \eq{xakappaa}, and $|z_{ak}/t-\zeta_{ak}|\le|x_{a}/t-\xi_{a}|\le\sqrt{2{\varepsilon}}$. Using \eq{3.25-2}, we then have $|z_{ak}|/t\ge|\zeta_{ak}|-\sqrt{2{\varepsilon}}\ge (2^{-1}\rho_{|a|})^{1/2}|\xi_a|-\sqrt{2{\varepsilon}}$. As support condition implies $|\xi_a|\ge 2^{-1}d_1$, we obtain $|z_{ak}|/t\ge 2^{-1}(2^{-1}\rho_{|a|})^{1/2}d_1-\sqrt{2{\varepsilon}}\ge(2^3-1)\kappa_{|a|}$. From these and \eq{fundamentalinequality} follows $|x_{ij}|/t\ge (2^3-3)\kappa_{|a|}$, which gives \eq{xbetazakgekappaa}.
\end{proof}
We note that both conditions \eq{condition-1} and \eq{condition-2} are simultaneously satisfied if we take $\rho_{|a|}>\theta_{|a|}>\sigma>0$ and ${\varepsilon}>0$ so small that
\begin{equation}
\rho_{|a|}\ge 2^{10}d_1^{-2}{\varepsilon} > 2^{14}d_2^2d_1^{-2}\omega_{|a|}^{-1}.\label{compatibilitycondition}
\end{equation}
In the rest of the paper we always assume that the constants $\rho_j>\theta_j>\sigma>0$ and ${\varepsilon}>0$ are taken to satisfy this condition \eq{compatibilitycondition} for given $d_2>d_1>0$ in each context.

For $t\ge s>1$
\begin{equation}
\begin{aligned}
&p_a^{\varepsilon}(t,x(t),\xi(t))-p_a^{\varepsilon}(s,y,\eta)\label{p(t)-p(s)}\\
&=\int_s^t\frac{d}{d\tau}p_a^{\varepsilon}(\tau,x(\tau),\xi(\tau))d\tau
\\
&=\int_s^t\left|\frac{1}{\sqrt{\tau}}q_a^{\varepsilon}(\tau,x(\tau),\xi(\tau))\right|^2d\tau-\int_s^t\partial_\xi p_a^{\varepsilon}(\tau,x(\tau),\xi(\tau))\cdot\partial_xV(x(\tau))d\tau.
\end{aligned}
\end{equation}
Here
\begin{equation}
\begin{aligned}
&\partial_\xi p_a^{\varepsilon}(t,x,\xi)\cdot\partial_xV(x)\\
&=-2\phi'(|x/t-\xi_a|^2<{\varepsilon})\chi_a(\xi_a)^2\Phi_a(\xi_a)(x_a/t-\xi_a)\cdot\partial_{x_a}V(x)\\
&\hskip20pt+2\phi(|x/t-\xi_a|^2<{\varepsilon})\partial_{\xi_a}(\chi_a(\xi_a)^2\Phi_a(\xi_a))\cdot\partial_{x_a}V(x)\\
&=
-2\phi'(|x/t-\xi_a|^2<{\varepsilon})\chi_a(\xi_a)^2\Phi_a(\xi_a)(x_a/t-\xi_a)\cdot\partial_{x_a}I_a(x_a,x^a)\\
&\hskip20pt+2\phi(|x/t-\xi_a|^2<{\varepsilon})\partial_{\xi_a}(\chi_a(\xi_a)^2\Phi_a(\xi_a))\cdot\partial_{x_a}I_a(x_a,x^a).
\end{aligned}\label{drhodV}
\end{equation}
As $I_a(x_a,x^a)=\sum_{{(i,j)}\not\le a}V_{ij}(x_{ij})$, we have from \eq{drhodV} and Lemma \ref{lemmaonzakxa} that for all $t>1$, $(x,\xi)\in X\times (X'\setminus\{0\})$ and a sufficiently small ${\varepsilon}>0$
\begin{eqnarray}
|\partial_\xi p_a^{\varepsilon}(t,x,\xi)\cdot\partial_xV(x)|\le C t^{-1-\delta},\label{decay-potential-commutator}
\end{eqnarray}
where $C>0$ is a constant.
By 
\eq{p(t)-p(s)} and
\eq{decay-potential-commutator}, we have
\begin{equation}
\begin{aligned}
0\le\int_s^t\left|\frac{1}{\sqrt{\tau}}q_a^{\varepsilon}(\tau,x(\tau),\xi(\tau))\right|^2d\tau
\le 2+C\delta^{-1}s^{-\delta}\le 2+C\delta^{-1}<\infty\label{u-integrable}
\end{aligned}
\end{equation}
for all $t>s>1$. This and \eq{p(t)-p(s)} imply that
\begin{equation}
1\ge p_a^{\varepsilon}(t,x(t),\xi(t))\ge p_a^{\varepsilon}(s,x(s),\xi(s))-C\delta^{-1}s^{-\delta}
\end{equation} 
for $t>s>1$. Therefore if $(y,\eta)\in X\times (X'\setminus\{0\})$ and $s>1$ satisfy $2M=p_a^{\varepsilon}(s,y,\eta)>2C\delta^{-1}s^{-\delta}>0$, then for all $t>s>1$
\begin{equation}
p_a^{\varepsilon}(t,x(t,s,y,\eta),\xi(t,s,y,\eta))\ge M>0.
\end{equation}
Summarizing we have proved the following.
\begin{theorem} Let $|a|\ge2$, $1/4>\rho_{|a|}>\theta_{|a|}>0$, $\sigma>0$, ${\varepsilon}>0$, and $0<d_1<d_2<\infty$ satisfy \eq{compatibilitycondition}.
Assume that $(y,\eta)\in X\times (X'\setminus\{0\})$ satisfies $|y/s-\eta_a|^2<{\varepsilon}$, $d_1\le|\eta_a|\le d_2$ and $\eta\in \mbox{{\rm supp} }\Phi_a$.  Let $t>s>1$ satisfy $C\delta^{-1}s^{-\delta}<1/2$ for the constant $C$ in \eq{decay-potential-commutator}. Then
\begin{equation}
1\ge p_a^{\varepsilon}(t,x(t,s,y,\eta),\xi(t,s,y,\eta))\ge 1/2>0.
\end{equation}
In particular $(x(t),\xi(t))\in \mbox{\rm supp }p_a^{\varepsilon}(t,\cdot,\cdot)$. Namely $|x(t)/t-\xi_a(t)|^2<2{\varepsilon}$, $d_1/2\le |\xi_a(t)|\le 2d_2$ and $\xi(t)\in \mbox{\rm supp }\Phi_a$. Hence letting $(i,j)\not\le a$ and $z_{ak}$ $(1\le k\le k_a)$ such that $C_\ell,C_m\in a$ $(\ell\ne m)$, $i\in C_\ell$, $j\in C_m$, and $z_{ak}$ connects $C_\ell$ and $C_m$, we have by Lemma \ref{lemmaonzakxa}, $|x_{ij}(t)|\ge5\kappa_{|a|}t$, $|z_{ak}(t)|\ge7\kappa_{|a|}t$, and $|x^a(t)|\le\kappa_{|a|}t$.
\end{theorem}
The theorem implies that if the initial condition $(y,\eta)$ at $t=s$ of the Hamilton's equation satisfies $y/s\approx \eta_a$, $|\eta_a|>0$ and $|\eta_{ij}|/|\eta^a|$ is sufficiently large for any pair $(i,j)\not\le a$, then henceforth the trajectory $(x(t),\xi(t))$ satisfies the same property forever $t\ge s$.

\section{Some asymptotics}\label{asymptoticdecompositionofcontinuousspectralsubspacebyscatteringspaces}

We use the notation introduced at the beginning of section \ref{Contiuousspectrum} and recall Theorem \ref{quantumclassical} as the following lemma.
\begin{lemma}\label{weakconvergence}
Let Assumptions \ref{potentialdecay} and \ref{eigenfunctiondecay} hold. Let $f\in \HH_c(H)$. Then there exist a sequence $t_m\to\infty$ (as $m\to\infty$) and a sequence $M_d^m$ of multi-indices whose components all tend to $\infty$ as $m\to\infty$ such that for all cluster decompositions $a$ with $2\le |a|\le N$, for all $\varphi\in C_0^\infty(X_a)$, $R>0$, and ${(i,j)}\not\leq a$
\begin{align}
&\left\Vert\frac{|x^a|^2}{t_m^2}{\tilde P}^a_{M_a^m}e^{-it_mH}f\right\Vert \to 0\label{(2.6)-Lemma}\\
&\Vert \chi_{\{x||x_{ij}|<R\}}{\tilde P}^a_{M_a^m}e^{-it_mH}f\Vert \to 0\label{1.24}\\
&\Vert (\varphi(x_a/t_m)-\varphi(D_{x_a})){\tilde P}^a_{M_a^m}e^{-it_mH}f\Vert \to 0\label{1.25}
\end{align}
as $m\to\infty$. Here $\chi_S$ is the characteristic function of a set $S$.
\end{lemma}
We can decompose the asymptotic behavior of $e^{-itH}f$ when $f\in \HH_c(H)$ as follows.
\begin{theorem}\label{eigenprojectionisemergent} Let Assumptions \ref{potentialdecay} and \ref{eigenfunctiondecay} be satisfied. Let $f=E_H(B)\in \HH_c(H)$ with $B\Subset\R\setminus\mathcal{T}$ and let the sequence $t_m$ be defined in Theorem \ref{quantumclassical} for this $f$. Then for $\Phi_a$ in Theorem \ref{homogeneousextension} with sufficiently small constants $\rho_j>\theta_j>\sigma>0$, we have
\begin{align}
\bigl\Vert e^{-it_mH}f-\sum_{2\le|a|\le N}\Phi_a(x/t_m)\sum_{d\le a}\tilde P^d_{M_d^m}e^{-it_mH}f\bigr\Vert\to0\quad(m\to\infty).\label{oneofdesiredones0}
\end{align}
Further there exist constants $0<d_1<d_2<\infty$ such that for
the function $\chi_a$ defined by \eq{chid} and for a sufficiently small ${\varepsilon}_a>0$ satisfying \eq{compatibilitycondition}:
 \begin{equation}
\rho_{|a|}\ge 2^{10}{d}_1^{-2}{\varepsilon}_a \ge 2^{14}{d}_2^2{d}_1^{-2}\omega_{|a|}^{-1},\label{compatibilitycondition-a}
\end{equation}
we have for $P^{{\varepsilon}_a}_a(t)$ defined by \eq{(4)}-\eq{Paepsilont} with the above $\Phi_a$ and $\chi_a$
\begin{equation}
\bigl\Vert e^{-it_mH}f-\sum_{2\le|a|\le N}P_a^{{\varepsilon}_a}(t_m)e^{-it_mH}f\bigr\Vert\to0\quad(m\to\infty).
\label{asymptotics}
\end{equation}
\end{theorem}
\begin{proof}
By \eq{summationofprojections} and Theorem \ref{chidonPdMd}, there
 exist constants $0<{\tilde d}_1<{\tilde d}_2<\infty$ such that
\begin{align}
e^{-it_mH}f\sim\sum_{2\le|d|\le N}\chi_d(D_d)\tilde P^d_{M_d^m}e^{-it_mH}f\label{asymptoticrelationaldecomposition-0}
\end{align}
as $t_m\to\infty$, where $\chi_d$ is the function defined by \eq{chid} with $d_j={\tilde d}_j$ $(j=1,2)$.
Applying partition $\{\Phi_a\}_{2\le|a|\le N}$ of unity in Theorem \ref{homogeneousextension} to \eq{asymptoticrelationaldecomposition-0}, we get
\begin{align}
e^{-it_mH}f\sim\sum_{2\le |a|\le N}\Phi_a(x/t_m)\sum_{2\le|d|\le N}\chi_d(D_d)\tilde P^d_{M_d^m}e^{-it_mH}f\label{asymptoticrelationaldecomposition}
\end{align}
as $t_m\to\infty$. Theorem \ref{quantumclassical} tells that we can freely exchange any factor $\varphi(x_d/t_m)$ for $\varphi(D_d)$ on the left of $\tilde P^d_{M_d^m}e^{-it_mH}f$ in \eq{asymptoticrelationaldecomposition} and we have
$2{\tilde d}_2\ge |x_d/t_m|\ge {\tilde d}_1/2>0$ and $|x/t_m|\sim |x_d/t_m|$ on each state $\chi_d(D_d)\tilde P^d_{M_d^m}e^{-it_mH}f$ asymptotically as $m\to\infty$. 
Thus 
$|x|\ge 2^{-1}{\tilde d}_1t_m$ asymptotically on $\chi_d(D_d)\tilde P^d_{M_d^m}e^{-it_mH}f$. Therefore taking the constants $\rho_j>\theta_j>\sigma>0$ in $\Phi_a$ such that $\rho_j/\theta_j>2^{10}$ as in Theorem \ref{homogeneousextension} so that $|x_{ij}|>\rho_{|a|}^{1/2}|x|/4$ holds on supp $\Phi_a(x/t_m)\chi_d(x_d/t_m)$ for any pair $(i,j)\not\le a$,
we have on supp $\Phi_a(x/t_m)$ that $|x_{ij}|\ge ct_m$ for any $(i,j)\not\le a$ with some constant $c>0$. Hence by Theorem \ref{quantumclassical} only the terms with $d\le a$ remain in the second sum on the RHS of \eq{asymptoticrelationaldecomposition} and we arrive at a formula equivalent with \eq{asymptoticrelationaldecomposition}.
\begin{equation}
\begin{aligned}
e^{-it_mH}f\sim\sum_{2\le|a|\le N}\sum_{d\le a}\Phi_a(x/t_m)\chi_d(D_d)\tilde P^d_{M_d^m}e^{-it_mH}f.
\end{aligned}\label{asymptoticrelationaldecomposition-1}
\end{equation}
We note that this tells we can freely add or remove terms with $d\not\le a$ in \eq{asymptoticrelationaldecomposition-1}. Again using Theorem \ref{chidonPdMd} we obtain from \eq{asymptoticrelationaldecomposition-1}
\begin{align}
e^{-it_mH}f\sim\sum_{2\le|a|\le N}\Phi_a(x/t_m)\sum_{d\le a}\tilde P^d_{M_d^m}e^{-it_mH}f,\label{oneofdesiredones}
\end{align}
which is \eq{oneofdesiredones0}.

We will use the conventional notation $x^a=(x_d^a,x^d)$ $(d\le a)$, $D_{x_d}=D_d$, $D_{x_d^a}=D_d^a$. We set $d_1={\tilde d}_1/4$ and $d_2=2{\tilde d}_2$.
Noting that the factor $\chi_d(D_d)$ gives $2d_1\le{\tilde d}_1/2\le|D_d|\le 2{\tilde d}_2=d_2$ and using \eq{3.25-2}, \eq{1.25} and the factor $\Phi_a(x/t)$ in \eq{asymptoticrelationaldecomposition-1}, we have asymptotically $|D_d^a|^2\le |x^d/t|^2+|x_d^a/t|^2=|x^a/t|^2\le \omega_{|a|}^{-1}|x_a/t|^2\sim \omega_{|a|}^{-1}|D_a|^2\le \omega_{|a|}^{-1}|D_d|^2\le d_2^2\omega_{|a|}^{-1}$. Hence $d_2^2\ge |D_d|^2\ge |D_a|^2=|D_d|^2-|D_d^a|^2\ge 4d_1^2-d_2^2\omega_{|a|}^{-1}$. Taking $\rho_{|a|}>\theta_{|a|}>\sigma>0$ so small that $\omega_{|a|}^{-1}<3d_1^2d_2^{-2}$, we thus have $d_2^2\ge |D_a|^2\ge d_1^2$ and $|x_d^a/t|^2\sim |D_d^a|^2\le|x^a/t|^2\le d_2^2\omega_{|a|}^{-1}\le 2^{-4}{\varepsilon}_a<{\varepsilon}_a$ by \eq{compatibilitycondition-a} asymptotically on each state $\tilde P^d_{M_d}e^{-itH}f$ in 
\eq{asymptoticrelationaldecomposition-1} as $t=t_m\to\infty$ with $M_d=M_d^m$. Let $\chi_a$ be defined by \eq{chid} for the constants $d_2>d_1>0$. We then get from \eq{asymptoticrelationaldecomposition-1}
\begin{equation}
\Vert e^{-itH}f
-\hskip-4pt\sum_{2\le|a|\le N}\phi(|x/t-D_{a}|^2<{\varepsilon}_a)\chi_a(D_a)^2\Phi_a(D_{a})\sum_{d\le a}\tilde P^d_{M_d}e^{-itH}f\Vert\to 0
\label{final-2}
\end{equation}
as $t=t_m\to\infty$ with $M_d=M_d^m$, 
where $\phi(t<{\varepsilon})$ is defined by \eq{phi1def}.
Namely we have for the function $p_a^{{\varepsilon}_a}(t,x,\xi_a)=\phi(|x/t-\xi_a|^2<{\varepsilon}_a)\chi_a(\xi_a)^2\Phi_a(\xi_a)$ in \eq{(4)}
\begin{equation}
\Vert e^{-itH}f-
\sum_{2\le|a|\le N}p_a^{{\varepsilon}_a}(t,x,D_{a})\sum_{d\le a}\tilde P^d_{M_d}e^{-itH}f\Vert\to 0\quad(t=t_m\to\infty).
\label{final-1}
\end{equation}
By what stated after \eq{asymptoticrelationaldecomposition-1}, we can recover the terms with $d\not\le a$. Using \eq{summationofprojections}, we finally arrive at for ${\varepsilon}_a>0$ satisfying \eq{compatibilitycondition} or \eq{compatibilitycondition-a}
\begin{equation}
\Vert e^{-it_mH}f-\sum_{2\le|a|\le N}P_a^{{\varepsilon}_a}(t_m)e^{-it_mH}f\Vert\to0\quad(m\to\infty).
\label{asymptotics-proven}
\end{equation}
This completes the proof of the theorem.
\end{proof}

\section{Time-independent modifier}

In this section, we assume Assumption \ref{potentialdecay} on long-range pair potentials. Under this assumption, we will construct a time-independent modifier $J_a$ which treats the intercluster long-range pair potentials $I^L_a(x)=\sum_{(ij)\not\le a}V^L_{ij}(x_{ij})$, which is an extension of $J$ introduced in \cite{IK-1}, \cite{[IK]} for two-body quantum scattering.

Let a $C^\infty$ function $\chi_0(x)$ of $x\in R^\nu$ satisfy $0\le \chi_0(x)\le 1$ and
\begin{align}
\chi_0(x)=
\left\{
\begin{array}{ll}
1 & (|x|\ge 2),\\
0 & (|x|\le 1).
\end{array} \label{5.4}
\right.
\end{align}
To define $J_{a}$ $(2\le|a|\le N)$, we will introduce a time-dependent potential $I^L_{a\rho}(t,x)$ for $\rho\in(0,1)$.
\begin{align}
I^L_{a\rho}(t,x)=I^L_a(x_a,x^a)\prod_{(ij)\not\le a}\chi_0(\rho x_{ij})\chi_0(\langle \log\langle t\rangle\rangle x_{ij}/\langle t\rangle).\label{5.5}
\end{align}
Then $I^L_{a\rho}(t,x)$ satisfies
\begin{align}
|\partial_{x}^\alpha I^L_{a\rho}(t,x)|\le C_\alpha \rho^{\delta_0}\langle t\rangle^{-\ell}\max_{(ij)\not\le a}(\langle x_{ij}\rangle^{-m})\label{5.6}
\end{align}
for any multiindex $\alpha$, and real numbers $\ell,m\ge 0$, $0<\delta_0<\delta$ with $\delta_0+\ell+m<|\alpha|+\delta$, where $C_\alpha>0$ is a constant independent of $t,x$ and $\rho$.

Then we can apply the arguments in section 2 of \cite{[IK]} to get a solution $\varphi_{a}(x,\xi)$ of the eikonal equation:
\begin{align}
\frac{1}{2}|\nabla_{x}\varphi_{a}(x,\xi)|^2+I^L_{a}(x_a,x^a)=\frac{1}{2}|\xi|^2
\label{5.7}
\end{align}
in some conic region in phase space. 

Namely let $H_{a\rho}(t,x,\xi)$ be the Hamiltonian corresponding to the time-dependent potential $I^L_{a\rho}(t,x)$.
\begin{align}
H_{a\rho}(t,x,\xi)=\frac{1}{2}|\xi|^2+I^L_{a\rho}(t,x).\label{timedependentHamiltonian}
\end{align}
The corresponding classical orbit $(q,p)(t,s,y,\xi)=(q(t,s,y,\xi),p(t,s,y,\xi))\in X\times X'=\R^{\nu(N-1)}\times \R^{\nu(N-1)}$ is determined by the equation
\begin{equation}
\left\{
\begin{aligned}
q(t,s)&=y+\int_s^t \nabla_\xi H_{a\rho}(\tau,q(\tau,s),p(\tau,s))d\tau=y+\int_s^t p(\tau,s)d\tau,\\
p(t,s)&=\xi- \int_s^t \nabla_x H_{a\rho}(\tau,q(\tau,s),p(\tau,s))d\tau=\xi- \int_s^t \nabla_xI^L_{a\rho}(\tau,q(\tau,s))d\tau.
\end{aligned}
\right.
\label{classical-orbits}
\end{equation}
We remark that when all long-range pair potentials vanish: $V^L_{ij}=0$, one has the unique solution $q(t,s,y,\xi)=y+(t-s)\xi$, $p(t,s,y,\xi)=\xi$.
Letting $\delta_0,\delta_1>0$ be fixed as $0<\delta_0+\delta_1<\delta$, we have the following estimates for $(q,p)(t,s,y,\xi)$, which are proved by solving the equation \eq{classical-orbits} by iteration as in Proposition 2.1 of \cite{[Ki-ScTime-Dep-Pot]}.
\begin{proposition}\label{Estimates-for-qp}
There are constants $C_\ell>0$ $(\ell=0,1,2,\cdots)$ such that for all $(y,\xi)\in X\times X'$ and $\pm t\ge \pm s \ge 0$, the solutions $q, p$ of \eq{classical-orbits} exist and satisfy for all multi-index ${\alpha}$:
\begin{eqnarray}
&&|p(s,t,y,\xi)-\xi|+|p(t,s,y,\xi)-\xi|\le C_0\rho^{\delta_0}\langle s\rangle^{-{\delta_1}},\label{pst}\\
&&|{\partial}_y^{\alpha}[\nabla_y q(s,t,y,\xi)-I]|\le C_{|{\alpha}|}\rho^{\delta_0}\langle s\rangle^{-{\delta_1}},\label{(5.13)}\\
&&|{\partial}_y^{\alpha}[\nabla_y p(s,t,y,\xi)]|\le C_{|{\alpha}|}\rho^{\delta_0}\langle s\rangle^{-1-{\delta_1}},\\
&&|\nabla_\xi q(t,s,y,\xi)-(t-s)I|\le C_0\rho^{\delta_0}\langle s\rangle^{-{\delta_1}}|t-s|,\label{q-2}\\
&&|\nabla_\xi p(t,s,y,\xi)-I|\le C_0\rho^{\delta_0}\langle s\rangle^{-{\delta_1}},\\
&&|\nabla_y q(t,s,y,\xi)-I|\le C_0\rho^{\delta_0}\langle s\rangle^{-1-{\delta_1}}|t-s|,\\
&&|\nabla_y p(t,s,y,\xi)|\le C_0\rho^{\delta_0}\langle s\rangle^{-1-{\delta_1}},\\
&&|{\partial}_\xi^{\alpha}[q(t,s,y,\xi)-y-(t-s)p(t,s,y,\xi)]|\nonumber\\
&&\le C_{|{\alpha}|}\rho^{\delta_0}\min(\langle t\rangle^{1-{\delta_1}},|t-s|\langle s\rangle^{-{\delta_1}}).\label{qts}
\end{eqnarray}
Further for any ${\alpha},{\beta}$ satisfying $|{\alpha}+{\beta}|\ge 2$, there is a constant $C_{{\alpha}{\beta}}>0$ such that
\begin{eqnarray}
&&|{\partial}_y^{\alpha}{\partial}_\xi^{\beta} q(t,s,y,\xi)|\le C_{{\alpha}{\beta}}\rho^{\delta_0}|t-s|\langle s\rangle^{-{\delta_1}},\label{q-3}\\
&&|{\partial}_y^{\alpha}{\partial}_\xi^{\beta} p(t,s,y,\xi)|\le C_{{\alpha}{\beta}}\rho^{\delta_0}\langle s\rangle^{-{\delta_1}}.\label{p-3}
\end{eqnarray}
\end{proposition}
{}For the constant $C_0>0$ in this proposition, we take $\rho>0$ so small that $C_0\rho^{\delta_0}<1/2$ holds. Then the mapping $T_x(y)=x+y-q(s,t,y,\xi):X\longrightarrow X$ becomes a contraction. Therefore there is a unique fixed point $y\in X=\R^{\nu (N-1)}$ for every $x\in X$ such that $T_x(y)=y$, hence $x=q(s,t,y,\xi)$. Thus we obtain the following. (See Proposition 2.2 of \cite{[Ki-ScTime-Dep-Pot]}.)

\begin{proposition}\label{Estimates-for-etay}
Take $\rho>0$ so small that $C_0\rho^{\delta_0}<1/2$ for the constant $C_0>0$ in Proposition \ref{Estimates-for-qp}. Then for $\pm t\ge \pm s\ge 0$ one can construct a diffeomorphism of $X=\R^{\nu (N-1)}$ for $\xi\in X'=\R^{\nu (N-1)}$
\begin{eqnarray}
x\mapsto y(s,t,x,\xi)\nonumber
\end{eqnarray}
such that
\begin{eqnarray}
q(s,t,y(s,t,x,\xi),\xi)=x.
\label{qp-yeta}
\end{eqnarray}
The mapping $y(s,t,x,\xi)$ is $C^\infty$ in $(x,\xi)\in X\times X'$ and its derivatives ${\partial}_x^{\alpha}{\partial}_\xi^{\beta} y$ are $C^1$ in $(t,s,x,\xi)$. Using this diffeomorphism we define for $\xi\in X'$
\begin{eqnarray}
\eta(t,s,x,\xi)=p(s,t,y(s,t,x,\xi),\xi).\label{etap}
\end{eqnarray}
Then $\eta(t,s,x,\xi)$ is a $C^\infty$ mapping from $X\times X'$ into $X'$, and satisfies
\begin{eqnarray}
p(t,s,x,\eta(t,s,x,\xi))=\xi.\label{peta}
\end{eqnarray}
They satisfy the relation
\begin{eqnarray}
y(s,t,x,\xi)=q(t,s,x,\eta(t,s,x,\xi))
\label{y-q-eta-p}
\end{eqnarray}
and the estimates
for any ${\alpha},{\beta}$
\begin{eqnarray}
&&|{\partial}_x^{\alpha}{\partial}_\xi^{\beta}[\nabla_x y(s,t,x,\xi)-I]|\le C_{{\alpha}{\beta}}\rho^{\delta_0}\langle s\rangle^{-{\delta_1}},\label{y-1}\\
&&|{\partial}_x^{\alpha}{\partial}_\xi^{\beta}[\nabla_x\eta(t,s,x,\xi)]|\le C_{{\alpha}{\beta}}\rho^{\delta_0}\langle s\rangle^{-1-{\delta_1}},\label{eta-1}\\
&&|{\partial}_\xi^{\alpha}[\eta(t,s,x,\xi)-\xi]|\le C_{{\alpha}}\rho^{\delta_0}\langle s\rangle^{-{\delta_1}},\label{eta-2}\\
&&|{\partial}_\xi^{\alpha}[y(s,t,x,\xi)-x-(t-s)\xi]|\label{y-2}\le C_{{\alpha}}\rho^{\delta_0}\min(\langle t\rangle^{1-{\delta_1}},|t-s|\langle s\rangle^{-{\delta_1}}).
\end{eqnarray}
Further for any $|{\alpha}+{\beta}|\ge 2$
\begin{eqnarray}
&&|{\partial}_x^{\alpha}{\partial}_\xi^{\beta} \eta(t,s,x,\xi)|\le C_{{\alpha}{\beta}}\rho^{\delta_0}\langle s\rangle^{-{\delta_1}},\label{eta-3}\\
&&|{\partial}_x^{\alpha}{\partial}_\xi^{\beta} y(s,t,x,\xi)|\le C_{{\alpha}{\beta}}\rho^{\delta_0}\langle t-s\rangle\langle s\rangle^{-{\delta_1}}.\label{y-3}
\end{eqnarray}
Here the constants $C_{\alpha}, C_{{\alpha}{\beta}}>0$ are independent of $t,s,x,\xi$
\end{proposition}

The following illustration would be helpful to understand the meaning of the diffeomorphisms $y(s,t,x,\xi)$ and $\eta(t,s,x,\xi)$. Let $U_{a\rho}(t,s)$ be the map that assigns the point $(q,p)(t,s,x,\eta)$ to the initial data $(x,\eta)$. Then
\begin{eqnarray}
\begin{array}{ccc}
\mbox{time}\ s& \ &\mbox{time}\ t\\
\left(
\begin{array}{c}
x\\
\ \\
\eta(t,s,x,\xi)
\end{array}
\right)
&
\begin{array}{c}
U_{a\rho}(t,s) \\
\longmapsto\\
\ 
\end{array}
&
\left(
\begin{array}{c}
y(s,t,x,\xi)\\
\ \\
\xi
\end{array}
\right).
\end{array}\nonumber
\end{eqnarray}

We now define $\phi_a(t,x,\xi)$ for $(x,\xi)\in X\times (X'\setminus\{0\})$ by
\begin{eqnarray}
\phi_a(t,x,\xi)=u_a(t,x,\eta(t,0,x,\xi)),\nonumber
\end{eqnarray}
where
\begin{eqnarray}
u_a(t,x,\eta)=x\cdot \eta+\int_0^t(H_{a\rho}-x\cdot \nabla_{x}H_{a\rho})(\tau,q(\tau,0,x,\eta),p(\tau,0,x,\eta))d\tau.\nonumber
\end{eqnarray}
Then it is shown by a standard calculation that $\phi_a(t,x,\xi)$ satisfies the Hamilton-Jacobi equation
\begin{equation}
\left\{
\begin{aligned}
&{\partial}_t\phi_a(t,x,\xi)=\frac{1}{2}|\xi|^2+I^L_{a\rho}(t,\nabla_{\xi}\phi_a(t,x,\xi)),\\
&\phi_a(0,x,\xi)=x\cdot \xi,
\end{aligned}
\right.
\label{Hamilton-Jacobi}
\end{equation}
and the relation
\begin{equation}
\left\{
\begin{aligned}
&\nabla_x\phi_a(t,x,\xi)=\eta(t,0,x,\xi),\\
&\nabla_\xi\phi_a(t,x,\xi)=y(0,t,x,\xi).
\end{aligned}
\right.
\label{eta-phi-y-phi}
\end{equation}
Then we can show the existence of the limit for $(x,\xi)\in X\times (X'\setminus \{0\})$
\begin{eqnarray}
\phi_{a}^{\pm}(x,\xi)=\lim_{t\to\pm\infty}(\phi_a(t,x,\xi)-\phi_a(t,0,\xi)),\label{limit-phi}
\end{eqnarray}
and show that the limit $\phi_{a}^{\pm}(x,\xi)$ defines a $C^\infty$-function of $(x,\xi)\in X\times (X'\setminus \{0\})$.
We remark that when all long-range pair potentials vanish: $V^L_{ij}=0$, one has the solution $\phi_a(t,x,\xi)=x\cdot\xi+t|\xi|^2$ of \eq{Hamilton-Jacobi}. Hence $\phi^\pm_a(x,\xi)=x\cdot\xi$.

Let for $1\le i<j\le N$
\begin{align*}
\cos(x_{ij},\xi_{ij}):=\frac{x_{ij}\cdot {\xi_{ij}}}{|x_{ij}||\xi_{ij}|},
\end{align*}
where $\xi_{ij}$ is the momentum variable conjugate to $x_{ij}$.
We set for $R_0>0$, $d>0$, $\theta\in(0,1)$, and $2\le|a|\le N$
\begin{align}
\Gamma_a^{\pm}(R_0,d,\theta)=\{(x,\xi)\ |\ |x_{ij}|\ge R_0,|\xi_{ij}|\ge d, \pm\cos(x_{ij},\xi_{ij})\ge \theta,((ij)\not\le a)\}.\label{Gammaryoiki}
\end{align}
We can now prove the following theorem in the same way as in \cite{[IK]}.
\begin{theorem}\label{Theorem 5.1} Let Assumption \ref{potentialdecay} be satisfied and let $2\le|a|\le N$. Then there exists a $C^\infty$ function $\phi_{a}^\pm(x,\xi)$ that satisfies the following properties: For any $0<\theta<1$, $d>0$ there exists a constant $R_0> 1$ such that for any $(x,\xi)\in\Gamma_a^{\pm}(R_0,d,\theta)$
\begin{align}
\frac{1}{2}|\nabla_{x}\phi_{a}^\pm(x,\xi)|^2+I^L_{a}(x_a,x^a)=\frac{1}{2}|\xi|^2
\label{5.8}
\end{align}
and
\begin{align}
|\partial_{x}^\alpha\partial_{\xi}^\beta(\phi_{a}^\pm(x,\xi)- x\cdot\xi)|\le
 C_{\alpha\beta}\max_{(ij)\not\le a}(\langle x_{ij}\rangle^{1-\delta-|\alpha|}),
\label{5.9}
\end{align}
where $C_{\alpha\beta}>0$ is a constant independent of $(x,\xi)\in\Gamma_a^\pm(R_0,d,\theta)$.
\end{theorem}
\begin{proof} We consider $\phi^+$ only. $\phi^-$ can be treated similarly. 
We first prove the existence of the limit \eq{limit-phi} for $t\to+\infty$ and $(x,\xi)\in X\times (X'\setminus\{0\})$. To do so, 
setting
\begin{eqnarray}
R(t,x,\xi)=\phi(t,x,\xi)-\phi(t,0,\xi),\nonumber
\end{eqnarray}
we show the existence of the limits 
\begin{eqnarray}
\lim_{t\to\infty}{\partial}_x^{\alpha}{\partial}_\xi^{\beta} R(t,x,\xi)=\lim_{t\to\infty}\int_0^t
{\partial}_x^{\alpha}{\partial}_\xi^{\beta} {\partial}_t R(\tau,x,\xi)d\tau+{\partial}_x^{\alpha}{\partial}_\xi^{\beta}(x\cdot\xi).\nonumber
\end{eqnarray}
By Hamilton-Jacobi equation \eq{Hamilton-Jacobi},
\begin{equation}
\begin{aligned}
{\partial}_t R(t,x,\xi)&={\partial}_t\phi(t,x,\xi)-{\partial}_t\phi(t,0,\xi)\\
&=I^L_{a\rho}(t,\nabla_\xi\phi(t,x,\xi))-I^L_{a\rho}(t,\nabla_\xi\phi(t,0,\xi))\\
&=(\nabla_\xi\phi(t,x,\xi)-\nabla_\xi\phi(t,0,\xi))\cdot
a(t,x,\xi)
\\
&=(y(0,t,x,\xi)-y(0,t,0,\xi))\cdot a(t,x,\xi)\\
&=\nabla_\xi R(t,x,\xi)\cdot a(t,x,\xi),
\end{aligned}\label{dtR}
\end{equation}
where
\begin{eqnarray}
&&a(t,x,\xi)=\int_0^1(\nabla_x I^L_{a\rho})
(t,\nabla_\xi\phi(t,0,\xi)+\sigma\nabla_\xi R(t,x,\xi))d\sigma,\label{def-a}\\
&&
\nabla_\xi R(t,x,\xi)=x\cdot\int_0^1(\nabla_xy)(0,t,\sigma x,\xi)d\sigma.
\end{eqnarray}
By \eq{y-1}, we have for any ${\alpha},{\beta}$
\begin{eqnarray}
|{\partial}_x^{\alpha}{\partial}_\xi^{\beta}\nabla_{\xi_{ij}} R(t,x,\xi)|\le C_{{\alpha}{\beta}}\langle x_{ij}\rangle.\label{R-1}
\end{eqnarray}
By \eq{y-2} and \eq{eta-phi-y-phi}, for $|{\beta}|\ge 1$ and $\xi\in  X'\setminus\{0\}$
\begin{eqnarray}
|{\partial}_\xi^{\beta} \nabla_\xi \phi(t,0,\xi)|\le C_{\beta} |t|.\label{nablaxi-phi}
\end{eqnarray}
From this, \eq{def-a}, and \eq{R-1}, we have for $\xi\in X'\setminus\{0\}$
\begin{eqnarray}
&&|{\partial}_x^{\alpha}{\partial}_\xi^{\beta} a(t,x,\xi)|\le C_{{\alpha}{\beta}}\langle t\rangle^{-1-\delta/2}\langle x\rangle^{|{\alpha}|+|{\beta}|}
\label{a-1}
\end{eqnarray}
Thus by \eq{dtR}, \eq{R-1} and \eq{a-1}, there exists the limit for any ${\alpha},{\beta}$ and $(x,\xi)\in X\times (X'\setminus\{0\})$
\begin{eqnarray}
\lim_{t\to\infty}{\partial}_x^{\alpha}{\partial}_\xi^{\beta} R(t,x,\xi)=\int_0^\infty {\partial}_x^{\alpha}{\partial}_\xi^{\beta}\left(\nabla_\xi R(t,x,\xi)\cdot a(t,x,\xi)\right)dt+{\partial}_x^{\alpha}{\partial}_\xi^{\beta}(x\cdot \xi).\nonumber
\end{eqnarray}
In particular, $\phi^+(x,\xi)=\lim_{t\to\infty}R(t,x,\xi)$ and $\eta(\infty,0,x,\xi)=\lim_{t\to\infty}\nabla_x\phi(t,x,\xi)$ exist and are $C^\infty$ in $(x,\xi)\in X\times (X'\setminus\{0\})$.

Next we show \eq{5.8}. By the arguments above, the following limit exist:
\begin{align*}
\nabla_x\phi^+(x,\xi)&=\lim_{t\to\infty}\nabla_x\phi(t,x,\xi)=\lim_{t\to\infty}\eta(t,0,x,\xi)\\
&=\lim_{t\to\infty} p(0,t,y(0,t,x,\xi),\xi).
\end{align*}
Thus for a sufficiently large $|x_{ij}|$ (i.e. for $|\rho x_{ij}|\ge 2$ $((ij)\not\le a)$) we have
\begin{eqnarray}
{2}^{-1}|\nabla_x\phi^+(x,\xi)|^{2}+I^L_{a}(x)={2}^{-1}\lim_{t\to\infty}|p(0,t,y(0,t,x,\xi),\xi)|^{2}+I^L_{a\rho}(0,x).\label{phi-p-limit}
\end{eqnarray}
Set for $0\le s\le t<\infty$
\begin{eqnarray}
f_t(s,y,\xi)={2}^{-1}|p(s,t,y,\xi)|^{2}+I^L_{a\rho}(s,q(s,t,y,\xi)).\nonumber
\end{eqnarray}
Then by \eq{classical-orbits} we have
\begin{align*}
\frac{{\partial} f_t}{{\partial} s}(s,y,\xi)&= p(s,t,y,\xi)\cdot {\partial}_s p(s,t,y,\xi)+(\nabla_x I^L_{a\rho})(s,q(s,t,y,\xi))\cdot{\partial}_s q(s,t,y,\xi)\\
&+\frac{{\partial} I^L_{a\rho}}{{\partial} t}(s,q(s,t,y,\xi))\\
&=\frac{{\partial} I^L_{a\rho}}{{\partial} t}(s,q(s,t,y,\xi)).
\end{align*}
On the other hand we have from \eq{qp-yeta}, \eq{etap}, \eq{peta}  and \eq{y-q-eta-p}
\begin{align*}
q(s,t,y(0,t,x,\xi),\xi)&=q(s,t,q(t,0,x,\eta(t,0,x,\xi)),\xi)\\
&=q(s,0,x,\eta(t,0,x,\xi)),\\
p(s,t,y(0,t,x,\xi),\xi)&=p(s,t,q(t,0,x,\eta(t,0,x,\xi)),\xi)\\
&= p(s,0,x,\eta(t,0,x,\xi)).
\end{align*}
Now using Proposition \ref{Estimates-for-qp}, we have for $\cos(x_{ij},\xi_{ij})\ge \theta$ $((ij)\not\le a)$
\begin{align*}
&|q_{ij}(s,t,y(0,t,x,\xi),\xi)|=|q_{ij}(s,0,x,\eta(t,0,x,\xi))|\\
&\ge |x_{ij}+sp_{ij}(s,0,x,\eta(t,0,x,\xi))|-C_0\rho^{\delta_0}\langle s\rangle^{1-{\delta_1}}\\
&=|x_{ij}+sp_{ij}(s,t,y(0,t,x,\xi),\xi)|-C_0\rho^{\delta_0}\langle s\rangle^{1-{\delta_1}}\\
&\ge c(|x_{ij}|+s|\xi_{ij}|)-C_0\rho^{\delta_0}\langle s\rangle^{1-{\delta_1}}-C_0\rho^{\delta_0}\langle s\rangle^{1-{\delta_1}},
\end{align*}
where $c>0$ is a constant independent of $s,t,x,\xi$. By $(x,\xi)\in \Gamma_a^+(R,d,\theta)$, we have $|\xi_{ij}|\ge d$, and from the definition \eq{5.5} of $I^L_{a\rho}(t,x)$
\begin{eqnarray}
\mbox{supp}\ \frac{{\partial} I^L_{a\rho}}{{\partial} t}(s,x)\subset
\{x| 1\le \langle \log\langle s\rangle\rangle|x_{ij}|/\langle s\rangle\le 2\ ((ij)\not\le a)\}.\nonumber
\end{eqnarray}
Thus there is a constant $S=S_{d,\theta}>1$ independent of $t$ such that for any $s\in[S,t]$ and $(x,\xi)\in \Gamma_a^+(R,d,\theta)$
\begin{eqnarray}
\frac{{\partial} f_t}{{\partial} s}(s,y(0,t,x,\xi),\xi)=0.\nonumber
\end{eqnarray}
For $s\in[0,S]$, taking $R=R_S>1$ large enough, we have for $|x_{ij}|\ge R$ and $\cos(x_{ij},\xi_{ij})\ge \theta$
\begin{eqnarray}
\frac{{\partial} f_t}{{\partial} s}(s,y(0,t,x,\xi),\xi)=0.\nonumber
\end{eqnarray}
Therefore we have shown that for $(x,\xi)\in \Gamma_a^+(R,d,\theta)$
\begin{eqnarray}
f_t(s,y(0,t,x,\xi),\xi)=\mbox{constant for}\ 0\le s\le t<\infty.\nonumber
\end{eqnarray}
In particular we have
\begin{eqnarray}
f_t(0,y(0,t,x,\xi),\xi)=f_t(t,y(0,t,x,\xi),\xi),\nonumber
\end{eqnarray}
which means
\begin{eqnarray}
{2}^{-1}|p(0,t,y(0,t,x,\xi),\xi)|^{2}+I^L_{a\rho}(0,x)={2}^{-1}|\xi|^{2}+I^L_{a\rho}(t,y(0,t,x,\xi)).\nonumber
\end{eqnarray}
Since $I^L_{a\rho}(t,y)\to0$ uniformly in $y\in \R^n$ when $t\to\infty$ by \eq {5.6}, we have from this and \eq{phi-p-limit}
\begin{eqnarray}
{2}^{-1}|\nabla_x\phi^+(x,\xi)|^{2}+I^L(x)={2}^{-1}|\xi|^{2}\quad \mbox{for}\quad (x,\xi)\in \Gamma_a^+(R,d,\theta),\nonumber
\end{eqnarray}
if $R>1$ is sufficiently large.

We finally prove the estimates \eq{5.9}. We first consider the derivatives with respect to $\xi$:
\begin{eqnarray}
{\partial}_\xi^{\beta}(\phi^+(x,\xi)-x\cdot \xi)=\int_0^\infty{\partial}_\xi^{\beta} {\partial}_t R(t,x,\xi)dt,\label{partialxiphi}
\end{eqnarray}
where $R(t,x,\xi)=\phi(t,x,\xi)-\phi(t,0,\xi)$ as above. Set
\begin{eqnarray}
\gamma(t,x,\xi)=y(0,t,x,\xi)-(x+t\xi)\nonumber
\end{eqnarray}
for $(x,\xi)\in\Gamma_a^+(R,d,\theta)$. Then
by \eq{y-2} we have for $\sigma\in[0,1]$
\begin{equation}
\begin{aligned}
&|\nabla_{\xi_{ij}}\phi(t,0,\xi)+\sigma\nabla_{\xi_{ij}} R(t,x,\xi)|\\
&=|y_{ij}(0,t,0,\xi)+\sigma(y_{ij}(0,t,x,\xi)-y_{ij}(0,t,0,\xi))|\\
&=|t\xi_{ij}+\gamma_{ij}(t,0,\xi)+\sigma(x_{ij}+\gamma_{ij}(t,x,\xi)-\gamma_{ij}(t,0,\xi))|\\
&=|\sigma x_{ij}+t\xi_{ij}+(1-\sigma)\gamma_{ij}(t,0,\xi)+\sigma\gamma_{ij}(t,x,\xi)|\\
&\ge c_0(\sigma|x_{ij}|+t|\xi_{ij}|{})-c_1\rho^{\delta_0}\min(\langle t\rangle^{1-{\delta_1}},|t|)
\end{aligned}\label{nabphiR}
\end{equation}
for some constants $c_0,c_1>0$ independent of $x,\xi,\sigma$ and $t\ge0$. Thus there are constants $\rho\in(0,1)$ and $T=T_{d,\theta}>0$ such that for all $t\ge T$ and $(x,\xi)\in \Gamma_a^+(R,d,\theta)$
\begin{eqnarray}
\langle \nabla_{\xi_{ij}}\phi(t,0,\xi)+\sigma\nabla_{\xi_{ij}} R(t,x,\xi)\rangle^{-1}\le C\langle \sigma|x_{ij}|+t|\xi_{ij}|{}\rangle^{-1}.\nonumber
\end{eqnarray}
Therefore $a(t,x,\xi)$ defined by \eq{def-a} satisfies by \eq{R-1} and \eq{nablaxi-phi}
\begin{eqnarray}
&&\hskip-48pt|{\partial}_\xi^{\beta} a(t,x,\xi)|\le C_{\beta} \sum_{(ij)\not\le a}\int_0^1\langle \sigma|x_{ij}|+t|\xi_{ij}|{}\rangle^{-1-\delta}d\sigma.\label{est-a}
\end{eqnarray}
Using \eq{nabphiR}, we see that \eq{est-a} holds also for $t\in[0,T]$ if we take $\rho>0$ small enough. Therefore for all $(x,\xi)\in\Gamma_a^+(R,d,\theta)$ we have from \eq{dtR}, \eq{R-1} and \eq{partialxiphi}
\begin{align*}
|{\partial}_\xi^{\beta}(\phi^+(x,\xi)-x\cdot\xi)|
&\le C_{T,{\beta},d} \sum_{(ij)\not\le a}\langle x_{ij}\rangle
\int_0^\infty\int_0^1\langle \sigma|x_{ij}|+t|\xi_{ij}|{}\rangle^{-1-\delta}d\sigma dt\\
&\le C_{T,{\beta},d}\sum_{(ij)\not\le a}\langle x_{ij}\rangle|\xi_{ij}|^{-1}\int_0^1\int_0^\infty\langle\sigma|x_{ij}|+\tau\rangle^{-1-\delta} d\tau d\sigma\\
&\le C_{T,{\beta},d}\sum_{(ij)\not\le a}\langle x_{ij}\rangle^{1-\delta}|\xi_{ij}|^{-1}.
\end{align*}

We next consider
\begin{equation}
\begin{aligned}
\nabla_x\phi^+(x,\xi)-\xi&=\lim_{t\to\infty}(\nabla_x\phi(t,x,\xi)-\xi)\\
&=\lim_{t\to\infty}(p(0,t,y(0,t,x,\xi),\xi)-\xi)\\
&=\lim_{t\to\infty}\int_0^t(\nabla_xI^L_{a\rho})\left(\tau,q(\tau,t,y(0,t,x,\xi),\xi)\right)d\tau\\
&=\lim_{t\to\infty}\int_0^t(\nabla_x I^L_{a\rho})\left(\tau,q(\tau,0,x,\eta(t,0,x,\xi))\right)d\tau\\
&=\int_0^\infty (\nabla_x I^L_{a\rho})\left(\tau,q(\tau,0,x,\eta(\infty,0,x,\xi))\right)d\tau.
\end{aligned}\label{nablaxphi}
\end{equation}
By \eq{pst} and \eq{qts} of Proposition \ref{Estimates-for-qp}
\begin{align*}
|q_{ij}(\tau,0,x,\eta(\infty,0,x,\xi))|&\ge
|x_{ij}+\tau p_{ij}(\tau,0,x,\eta(\infty,0,x,\xi))|-C_0\rho^{\delta_0}\langle \tau\rangle^{1-{\delta_1}}\\
&\ge
|x_{ij}+\tau p_{ij}(\tau,\infty,y(0,\infty,x,\xi),\xi)|-C_0\rho^{\delta_0}\langle \tau\rangle^{1-{\delta_1}}\\
&\ge c_1(|x_{ij}|+\tau|\xi_{ij}|{})-C_0\rho^{\delta_0}\langle \tau\rangle^{1-{\delta_1}}-C_0\rho^{\delta_0}\langle \tau\rangle^{1-{\delta_1}}
\end{align*}
for some constant $c_1>0$ and for all $(x,\xi)\in\Gamma_a^+(R,d,\theta)$.
Thus taking $\rho>0$ sufficiently small and $R=R_{d,\theta,\rho}>1$ sufficiently large, we have for $(x,\xi)\in\Gamma_a^+(R,d,\theta)$
\begin{eqnarray}
|q_{ij}(\tau,0,x,\eta(\infty,0,x,\xi))|\ge c_0(|x_{ij}|+\tau|\xi_{ij}|{})\nonumber
\end{eqnarray}
for some constant $c_0>0$. Therefore we obtain
\begin{eqnarray}
|\nabla_x\phi^+(x,\xi)-\xi|\le C\sum_{(ij)\not\le a}\int_0^\infty\langle|x_{ij}|+\tau|\xi_{ij}|{}\rangle^{-1-\delta}d\tau
\le C\max_{(ij)\not\le a}\langle x_{ij}\rangle^{-\delta}.\nonumber
\end{eqnarray}

For higher derivatives, the proof is similar. For example let us consider
\begin{align*}
{\partial}_\xi{\partial}_x\phi^+(x,\xi)-I&=\int_0^\infty  {\partial}_\xi\{(\nabla_x I^L_{a\rho})\left(\tau,q(\tau,0,x,\eta(\infty,0,x,\xi))\right)\}d\tau\\
&=\int_0^\infty  (\nabla_x\nabla_x I^L_{a\rho})\left(\tau,q(\tau,0,x,\eta(\infty,0,x,\xi))\right) \nabla_\xi q\cdot \nabla_\xi\eta d\tau,
\end{align*}
where we abbreviated $q=q(\tau,0,x,\eta(\infty,0,x,\xi))$ and $\eta=\eta(\infty,0,x,\xi)$.
The RHS is bounded by a constant times
\begin{eqnarray}
 \sum_{(ij)\not\le a}\int_0^\infty\langle |x_{ij}|+\tau|\xi_{ij}|{}\rangle^{-2-\delta}\langle\tau\rangle d\tau\le c_{d} \max_{(ij)\not\le a}\langle x_{ij}\rangle^{-\delta}\nonumber
\end{eqnarray}
for $(x,\xi)\in\Gamma_a^+(R,d,\theta)$ by \eq{q-2} and \eq{eta-2} of Propositions \ref{Estimates-for-qp} and \ref{Estimates-for-etay}.
Other estimates are proved similarly by using \eq{(5.13)}, \eq{y-1}, \eq{q-2}, \eq{q-3}, \eq{eta-2}, \eq{eta-3}, \eq{nablaxphi} and the following relations.
\begin{align*}
&\nabla_x\phi^+(x,\xi)-\xi=\lim_{t\to\infty}\int_0^t(\nabla_xI^L_{a\rho})\left(\tau,q(\tau,t,y(0,t,x,\xi),\xi)\right)d\tau,\\
&q(\tau,t,y(0,t,x,\xi),\xi)=q(\tau,0,x,\eta(t,0,x,\xi)).
\end{align*}
\end{proof}

From this we can derive the following theorem in the same way as in Theorem 2.5 of \cite{[IK]}. Let $0<\theta<1$ and let $\psi_\pm(\tau)\in C^\infty([-1,1])$ satisfy
\begin{align}
\begin{array}{ll}
&0\le \psi_\pm(\tau)\le 1,\\
&\psi_+(\tau)=\left\{
\begin{array}{ll} 1 &\text{for}\ \theta\le \tau\le 1,\\
0&\text{for}\ -1\le\tau\le\theta/2,
\end{array}\right.\\
&\psi_-(\tau)=\left\{
\begin{array}{ll} 0 &\text{for}\ -\theta/2\le\tau\le 1,\\
1&\text{for}\ -1\le \tau\le-\theta.
\end{array}\right.
\end{array}
\end{align}
We set for $(x,\xi)\in X\times (X'\setminus\{0\})$
\begin{align*}
\chi_a^\pm(x,\xi)=\prod_{(ij)\not\le a}\psi_\pm(\cos(x_{ij},\xi_{ij}))
\end{align*}
and define $\varphi_{a}(x,\xi)=\varphi_{{a},\theta,d,R_0}(x,\xi)$ by
\begin{equation}
\begin{aligned}
&\varphi_{a}(x,\xi)\\
&=\{(\phi_{a}^+(x,\xi)-x\cdot\xi)\chi_a^+(x,\xi)+
(\phi_{a}^-(x,\xi)-x\cdot\xi)\chi_a^-(x,\xi)\}\\
&\times\prod_{(ij)\not\le a}\chi_0(2\xi_{ij}/d)\chi_0(2x_{ij}/R_0)+x\cdot\xi
\end{aligned}\label{(5.50)}
\end{equation}
for
$d,R_0>0$. Note that $\varphi_{{a},\theta,d,R_0}(x,\xi_a)=\varphi_{{a},\theta,d',R_0'}(x,\xi_a)$ when $|x_{ij}|\ge\max(R_0,R_0')$, $|\xi_{ij}|\ge\max(d,d')$ for all $(ij)\not\le a$. We remark that when all long-range pair potentials vanish: $V^L_{ij}=0$, one has $\varphi_a(x,\xi)=x\cdot\xi$ since $\phi_a^\pm(x,\xi)=x\cdot\xi$ as mentioned. 
We now have

\begin{theorem}\label{Theorem 5.2} Let Assumption \ref{potentialdecay} be satisfied. Let $0<\theta<1$ and $d>0$. Then there exists a constant $R_0>1$ such that the $C^\infty$ function $\varphi_{a}(x,\xi)$ defined above satisfies the following properties.
\begin{namelist}{8888}
\item[  {\rm i)}] For $(x,\xi)\in \Gamma_a^+(R_0,d,\theta)\cup\Gamma_a^-(R_0,d,\theta)$, $\varphi_{a}$ is a solution of
\begin{align}
\frac{1}{2}|\nabla_{x}\varphi_{a}(x,\xi)|^2+I^L_{a}(x)=\frac{1}{2}|\xi|^2.
\label{5.10}
\end{align}
\item[  {\rm ii)}] For any $(x,\xi)\in X\times X'$ and multi-indices $\alpha,\beta$, $\varphi_{a}$ satisfies
\begin{align}
|\partial_{x}^\alpha\partial_{\xi}^\beta(\varphi_{a}(x,\xi)-x\cdot\xi)|
\le
 C_{\alpha\beta}\max_{(ij)\not\le a}(\langle x_{ij}\rangle^{1-\delta-|\alpha|}).\label{5.11}
\end{align}
In particular, if $\alpha\ne 0$,
\begin{align}
|\partial_{x}^\alpha\partial_{\xi}^\beta(\varphi_{a}(x,\xi)-x\cdot\xi)|\le C_{\alpha\beta}R_0^{-\delta_0}
\max_{(ij)\not\le a}(\langle x_{ij}\rangle^{1-\delta_1-|\alpha|})\label{5.12}
\end{align}
for any $\delta_0,\delta_1\ge 0$ with $\delta_0+\delta_1=\delta$. Further
\begin{align}
\varphi_{a}(x,\xi)=x\cdot\xi\label{5.13}
\end{align}
when $|x_{ij}|\le R_0/2$ or $|\xi_{ij}|\le d/2$ for some $(ij)\not\le a$.
\item[  {\rm iii)}] Let
\begin{align}
A_{a}(x,\xi)=e^{-i\varphi_{a}(x,\xi)}\left(H_0+I^L_{a}(x)-\frac{1}{2}|\xi|^2\right) e^{i\varphi_{a}(x,\xi)}.\label{5.14}
\end{align}
Then
\begin{equation}
\begin{aligned}
A_{a}(x,\xi)=\frac{1}{2}|\nabla_{x}\varphi_{a}(x,\xi)|^2+I^L_{a}(x_a,x^a)-\frac{1}{2}|\xi|^2-\frac{1}{2}i\Delta_{x}\varphi_{a}(x,\xi)
\end{aligned}\label{5.15}
\end{equation}
and
\begin{equation}
\begin{aligned}
&|\partial_{x}^\alpha\partial_{\xi}^\beta A_{a}(x,\xi)|\\
&\le\left\{
\begin{array}{ll}
C_{\alpha\beta}\max_{(ij)\not\le a}(\langle x_{ij}\rangle^{-1-\delta-|\alpha|}),& (x,\xi)\in\Gamma_a^+(R_0,d,\theta)\cup\Gamma_a^-(R_0,d,\theta),\\
C_{\alpha\beta}\max_{(ij)\not\le a}(\langle x_{ij}\rangle^{-\delta-|\alpha|}),&
\mbox{otherwise}.
\end{array}\right.
\end{aligned}\label{estimateofAa}
\end{equation}
\end{namelist}
\end{theorem}

In the same way, we obtain the following theorem with setting $x^a=0$, $\xi^a=0$ in the above. We note that by this setting $x^a=\xi^a=0$ we have $x_{ij}=z_{ak}$, $\xi_{ij}=\zeta_{ak}$ for some $k$ with $1\le k\le k_a$ for $2\le|a|\le N$, hence 
\begin{align}
\Gamma_a^{\pm}(R_0,d,\theta)=\{(x,\xi)\ |\ |z_{ak}|\ge R_0,|\zeta_{ak}|\ge d, \pm\cos(z_{ak},\zeta_{ak})\ge \theta,(1\le k\le k_a)\}.\label{Gammaryoiki-zak}
\end{align}
We let $X_a=\R^{\nu(|a|-1)}$ and $X_a'$ be the conjugate momentum space $\R^{\nu(|a|-1)}$.

\begin{theorem}\label{Theorem 5.2-2} Let Assumption \ref{potentialdecay} be satisfied. Let $0<\theta<1$ and $d>0$. Then there exists a constant $R_0>1$ such that the $C^\infty$ function $\varphi_{a}(x_a,\xi_a):=\varphi_a(x_a,0,\xi_a,0)$ satisfies the following properties.
\begin{namelist}{8888}
\item[  {\rm i)}] For $(x_a,\xi_a)\in \Gamma_a^+(R_0,d,\theta)\cup\Gamma_a^-(R_0,d,\theta)$, $\varphi_{a}$ is a solution of
\begin{align}
\frac{1}{2}|\nabla_{x}\varphi_{a}(x_a,\xi_a)|^2+I^L_{a}(x_a,0)=\frac{1}{2}|\xi_a|^2.
\label{5.10-zak}
\end{align}
\item[  {\rm ii)}] For any $(x_a,\xi_a)\in X_a\times X_a'$ and multi-indices $\alpha,\beta$, $\varphi_{a}$ satisfies
\begin{align}
|\partial_{x_a}^\alpha\partial_{\xi_a}^\beta(\varphi_{a}(x_a,\xi_a)-x_a\cdot\xi_a)|
\le
 C_{\alpha\beta}\max_{1\le k\le k_a}(\langle z_{ak}\rangle^{1-\delta-|\alpha|}).\label{5.11-zak}
\end{align}
In particular, if $\alpha\ne 0$,
\begin{align}
|\partial_{x_a}^\alpha\partial_{\xi_a}^\beta(\varphi_{a}(x_a,\xi_a)-x_a\cdot\xi_a)|\le C_{\alpha\beta}R_0^{-\delta_0}
\max_{1\le k\le k_a}(\langle z_{ak}\rangle^{1-\delta_1-|\alpha|})\label{5.12-zak}
\end{align}
for any $\delta_0,\delta_1\ge 0$ with $\delta_0+\delta_1=\delta$. Further
\begin{align}
\varphi_{a}(x_a,\xi_a)=x_a\cdot\xi_a\label{5.13-zak}
\end{align}
when $|z_{ak}|\le R_0/2$ or $|\zeta_{ak}|\le d/2$ for some $1\le k\le k_a$.
\item[  {\rm iii)}] Let
\begin{align}
A_{a}(x_a,\xi_a)=e^{-i\varphi_{a}(x_a,\xi_a)}\left(T_a+I^L_{a}(x_a,0)-\frac{1}{2}|\xi_a|^2\right) e^{i\varphi_{a}(x_a,\xi_a)}.\label{5.14-zak}
\end{align}
Then
\begin{equation}
\begin{aligned}
A_{a}(x_a,\xi_a)=\frac{1}{2}|\nabla_{x_a}\varphi_{a}(x_a,\xi_a)|^2+I^L_{a}(x_a,0)-\frac{1}{2}|\xi_a|^2-\frac{1}{2}i\Delta_{x_a}\varphi_{a}(x_a,\xi_a)
\end{aligned}\label{5.15-zak}
\end{equation}
and
\begin{equation}
\begin{aligned}
&|\partial_{x_a}^\alpha\partial_{\xi_a}^\beta A_{a}(x_a,\xi_a)|\\
&\le\left\{
\begin{array}{ll}
C_{\alpha\beta}\max_{1\le k\le k_a}(\langle z_{ak}\rangle^{-1-\delta-|\alpha|}),& (x_a,\xi_a)\in\Gamma_a^+(R_0,d,\theta)\cup\Gamma_a^-(R_0,d,\theta),\\
C_{\alpha\beta}\max_{1\le k\le k_a}(\langle z_{ak}\rangle^{-\delta-|\alpha|}),&
\mbox{otherwise}.
\end{array}\right.
\end{aligned}\label{estimateofAa-zak}
\end{equation}
\end{namelist}
\end{theorem}

We now define a Fourier integral operator $J_{a}=J_{{a},\theta,d,R_0}$ by
\begin{equation}
\begin{aligned}
J_{a}f(x)&=(2\pi)^{-\nu(|a|-1)/2}\int_{\R^{\nu(|a|-1)}}e^{i\varphi_a(x_a,\xi_a)}\hat f(\xi_a,x^a)d\xi_a
\end{aligned}\label{5.17}
\end{equation}
for $f\in \HH=\HH_{a}\otimes \HH^{a}=L^2(\R^{\nu(|{a}|-1)})\otimes L^2(\R^{\nu(N-|a|)})$,
where $\hat f(\xi_a,x^a)$ is the partial Fourier transform defined by \eq{partialfullFourier}. We note that this definition yields that $H^a$ and $J_a$ commute: $[H^a,J_a]=0$.
We remark that when all long-range pair potentials vanish: $V^L_{ij}=0$, one has
$\varphi_a(x_a,\xi_a)=x_a\cdot\xi_a$, hence $J_a=I$.

\section{Existence of modified wave operators}\label{existenceofwaveandinversewaveoperators}

In this section we will prove the following theorem.

\begin{theorem}\label{adjointmodifiedwaveoperator} Let Assumptions \ref{potentialdecay} and \ref{eigenfunctiondecay} hold. 
Let $a$ be a cluster decomposition with $2\le |a|\le N$, and $0<d_1<d_2<\infty$ be the constants in \eq{chid} in the definition of $\chi_a(\xi_a)$. Let $J_a$ be the Fourier integral operator defined by \eq{5.17}. Let ${\varepsilon}_a>0$ and constants $\rho_j>\theta_j>0$ satisfy \eq{compatibilitycondition}:
\begin{equation}
\rho_{|a|}\ge 2^{10}d_1^{-2}{\varepsilon}_a > 2^{14}d_2^2d_1^{-2}\omega_{|a|}^{-1},\label{compatibilitycondition-copy}
\end{equation}
 and let $P_{a}^{\varepsilon}(t)=P_{a,d_1,d_2}^{\varepsilon}(t)$ be the pseudodifferential operator defined by \eq{(4)}-\eq{Paepsilont}. Let $M\ge 0$ be an integer. Then for any $f\in \HH$ the following limits exist.
\begin{eqnarray}
\hskip-42pt&&\begin{aligned}
\Omega_af&=\lim_{t\to\infty}P^a_Me^{itH_a}J_a^*P_a^{{\varepsilon}_a}(t)e^{-itH}f=\lim_{t\to\infty}P^a_Me^{itH_a}P_a^{{\varepsilon}_a}(t)J_a^*e^{-itH}f,
\end{aligned}\label{inversewave}\\
\hskip-42pt&&\begin{aligned}
W_af&=\lim_{t\to\infty}e^{itH}P_a^{{\varepsilon}_a}(t)J_ae^{-itH_a}P^a_Mf=\lim_{t\to\infty}e^{itH}J_aP_a^{{\varepsilon}_a}(t)e^{-itH_a}P^a_Mf\\
&=\lim_{t\to\infty}e^{itH}J_ae^{-itH_a}P^a_M\chi_a(D_a)^2f,
\end{aligned}\label{wave}\\
\hskip-42pt&&\begin{aligned}
G_{a} f=\lim_{t\to\infty}e^{itH}P_{a}^{{\varepsilon}_a}(t)e^{-itH}f,
\end{aligned}\label{perturbedlimit}\\
\hskip-42pt&&\begin{aligned}
K_{a} f=\lim_{t\to\infty}e^{itH_a}P_{a}^{{\varepsilon}_a}(t)^*e^{-itH_a}f=\lim_{t\to\infty}e^{itH_a}P_{a}^{{\varepsilon}_a}(t)e^{-itH_a}f.
\end{aligned}\label{nonperturbedlimit}
\end{eqnarray}
In particular when the long-range part vanishes, i.e. when $V_{ij}^L\equiv 0$ for all pairs $(i,j)$, the limits \eq{inversewave}-\eq{wave} exist with $J_a$ replaced by the identity operator $I$. Further we have the existence of the following limits.
\begin{eqnarray}
&&\begin{aligned}
\tilde\Omega_af=\lim_{t\to\infty}e^{itH_a}P_a^{{\varepsilon}_a}(t)e^{-itH}f,
\end{aligned}\label{usualinversewave}\\
&&\begin{aligned}
\tilde W_af=\lim_{t\to\infty}e^{itH}P_{a}^{{\varepsilon}_a}(t)e^{-itH_a}P^a_Mf.
\end{aligned}\label{usualwave}
\end{eqnarray}
The wave operators $W_a$ $(2\le|a|\le N)$ satisfy the following properties.
\begin{namelist}{8888}
\item[  {\rm 1)}] Let $a$ and $a'$ be different cluster decompositions: $a\ne a'$. Then $\mathcal{R}(W_{a})$ and $\mathcal{R}(W_{a'})$ are orthogonal.
\begin{align}
\mathcal{R}(W_{a})\perp \mathcal{R}(W_{a'}).\label{orthogonal}
\end{align}
Similarly
\begin{equation}
\mathcal{R}(G_a)\perp\mathcal{R}(G_{a'}).\label{Gorthogonal}
\end{equation}
\item[  {\rm 2)}]
$W_{a}$ is an isometry on $P^a_ME_{T_a}([d_1^2/2,d_2^2/2])\HH$ and satisfies the intertwining property.
\begin{align}
E_H(B)W_{a}=W_{a}E_{H_a}(B),\quad B\subset \R \mbox{ {\rm(Borel set)}}.\label{intertwiningproperty}
\end{align}
\end{namelist}
\end{theorem}
\begin{proof} Proof is done by smooth operator technique as follows.
Let $\tilde H_a=H^a+T_a+I_a^L(x_a,0)$. Then we have by Theorem \ref{Theorem 5.2-2} and $[V^a,P_a^{{\varepsilon}_a}(t)]=0$
\begin{equation}
\begin{aligned}
&(\tilde H_aJ_a-J_aH_a)f(x)=(2\pi)^{-\nu(|a|-1)/2}\int_{\R^{\nu(|a|-1)}}e^{i\varphi_a(x_a,\xi_a)}A_a(x_a,\xi_a)\mathcal{F}_af(\xi_a,x^a)d\xi_a,\\
&[\tilde H_a,P_a^{{\varepsilon}_a}(t)]=[H_0,P_a^{{\varepsilon}_a}(t)]+[I_a^L(x_a,0),P_a^{{\varepsilon}_a}(t)].
\end{aligned}\label{HaJa-JaHa=Aa}
\end{equation}
This and Proposition \ref{positiveHeizenbergDerivative} yield that
\begin{equation}
\begin{aligned}
&\frac{d}{dt}(e^{itH_a}J_a^*P_a^{{\varepsilon}_a}(t)J_a e^{-itH_a}f)\\
&=e^{itH_a}\{i(H_aJ_a^*-J_a^*\tilde H_a)P_a^{{\varepsilon}_a}(t)J_a+J_a^*i[\tilde H_a,P_a^{{\varepsilon}_a}(t)]J_a\\
&+J_a^*P_a^{{\varepsilon}_a}(t)(\tilde H_aJ_a-J_aH_a)+J_a^*\partial_tP_a^{{\varepsilon}_a}(t) J_a\}e^{-itH_a}f\\
&=e^{itH_a}J_a^*\frac{1}{t}Q_a^{{\varepsilon}_a}(t)^*Q_a^{{\varepsilon}_a}(t)J_ae^{-itH_a}f+O(t^{-1-\delta})f,
\end{aligned}
\end{equation}
where $O(t^{-1-\delta})$ denotes an operator such that $\Vert O(t^{-1-\delta})\Vert\le Ct^{-1-\delta}$.
Thus for $t>s>1$
\begin{equation}
\begin{aligned}
&|(e^{itH_a}J_a^*P_a^{{\varepsilon}_a}(t)J_a e^{-itH_a}f,f)
-(e^{isH_a}J_a^*P_a^{{\varepsilon}_a}(s)J_a e^{-isH_a}f,f)|\\
&=\left|\int_s^t \frac{d}{d\tau}(e^{i\tau H_a}J_a^*P_a^{{\varepsilon}_a}(\tau)J_a e^{-i\tau H_a}f,f)dt\right|\\
&\ge \int_s^t\left\Vert\frac{1}{\sqrt{\tau}}Q_a^{{\varepsilon}_a}(\tau)J_ae^{-i\tau H_a}f\right\Vert^2d\tau-Cs^{-\delta}\Vert f\Vert^2.
\end{aligned}
\end{equation}
This gives
\begin{equation}
\int_s^t\left\Vert\frac{1}{\sqrt{\tau}}Q_a^{{\varepsilon}_a}(\tau)J_ae^{-i\tau H_a}f\right\Vert^2d\tau\le M_1^2\Vert f\Vert^2.\label{JeitHa}
\end{equation}
for some constant $M_1>0$ independent of $t>s>1$. Similarly
\begin{equation}
\begin{aligned}
\frac{d}{dt}(e^{itH}P_a^{{\varepsilon}_a}(t) e^{-itH}f)
&=e^{itH}\{i[H,P_a^{{\varepsilon}_a}(t)]+\partial_tP_a^{{\varepsilon}_a}(t) \}e^{-itH}f\\
&=e^{itH}\frac{1}{t}Q_a^{{\varepsilon}_a}(t)^*Q_a^{{\varepsilon}_a}(t)e^{-itH}f+O(t^{-1-\delta})f
\end{aligned}
\end{equation}
gives
\begin{equation}
\int_s^t\left\Vert\frac{1}{\sqrt{\tau}}Q_a^{{\varepsilon}_a}(\tau)e^{-i\tau H}f\right\Vert^2d\tau\le M_2^2\Vert f\Vert^2.\label{QeitH}
\end{equation}
for some constant $M_2>0$ independent of $t>s>1$.
On the other hand, 
since
\begin{equation}
\begin{aligned}
(HJ_a-J_a H_a)f(x)
=(\tilde H_aJ_a-J_aH_a)+(I_a^L(x_a,x^a)-I_a^L(x_a,0)+I_a^S(x_a,x^a))J_a
\end{aligned}\label{HJa-JaH0-00}
\end{equation}
and $P^a_M$ satisfies $\Vert P^a_M|x^a|\Vert<\infty$ by Assumption \ref{eigenfunctiondecay}, one has with using \eq{Ia(x)-Ia(xa,xa)}
\begin{equation}
\begin{aligned}
&\frac{d}{dt}(P^a_Me^{itH_a}J_a^*P_a^{{\varepsilon}_a}(t)e^{-itH}f,g)\\
&=(P^a_Me^{itH_a}\{i(H_aJ_a^*-J_a^*H)P_a^{{\varepsilon}_a}(t)+J_a^*i[H,P_a^{{\varepsilon}_a}(t)]+J_a^*\partial_t P_a^{{\varepsilon}_a}(t)\}e^{-itH}f,g)\\
&=(P^a_Me^{itH_a}\bigl\{-i(\tilde H_aJ_a-J_aH_a)^*P_a^{{\varepsilon}_a}(t)\\
&-iJ_a^*(I_a^L(x_a,x^a)-I_a^L(x_a,0)+I_a^S(x_a,x^a))P_a^{{\varepsilon}_a}(t)\\
&+J_a^*(i[H_0,P_a^{{\varepsilon}_a}(t)]+\partial_t P_a^{{\varepsilon}_a}(t))+J_a^*i[I_a,P_a^{{\varepsilon}_a}(t)]\bigr\}e^{-itH}f,g)\\
&=(P^a_Me^{itH_a}\{J_a^*t^{-1}Q_a^{{\varepsilon}_a}(t)^*Q_a^{{\varepsilon}_a}(t)+O(t^{-1-{\delta}})\}e^{-itH}f,g).
\end{aligned}
\end{equation}
Integrating this on the interval $[s,t]$ and applying \eq{JeitHa} and \eq{QeitH} yield that for $t>s>1$
\begin{equation}
\begin{aligned}
&|(P^a_Me^{itH_a}J_a^*P_a^{{\varepsilon}_a}(t)e^{-itH}f,g)-
(P^a_Me^{isH_a}J_a^*P_a^{{\varepsilon}_a}(s)e^{-isH}f,g)|\\
&\le \left(\int_s^t\left\Vert\frac{1}{\sqrt{\tau}}Q_a^{{\varepsilon}_a}(\tau)e^{-i\tau H}f\right\Vert^2d\tau\right)^{1/2}\left(\int_s^t\left\Vert\frac{1}{\sqrt{\tau}}Q_a^{{\varepsilon}_a}(\tau)J_ae^{-i\tau H_a}P^a_Mg\right\Vert^2d\tau\right)^{1/2}\\
&+O(s^{-{\delta}})\Vert f\Vert \Vert g\Vert\\
&\le \left(\int_s^t\left\Vert\frac{1}{\sqrt{\tau}}Q_a^{{\varepsilon}_a}(\tau)e^{-i\tau H}f\right\Vert^2d\tau\right)^{1/2}M_1\Vert g\Vert+Cs^{-{\delta}}\Vert f\Vert \Vert g\Vert.
\end{aligned}
\end{equation}
Thus for $t>s\to\infty$
\begin{equation}
\begin{aligned}
&\Vert P^a_Me^{itH_a}J_a^*P_a^{{\varepsilon}_a}(t)e^{-itH}f-
P^a_Me^{isH_a}J_a^*P_a^{{\varepsilon}_a}(s)e^{-isH}f\Vert\\
&\le\hskip-2pt M_1\hskip-2pt\left(\int_s^t\left\Vert\frac{1}{\sqrt{\tau}}Q_a^{{\varepsilon}_a}(\tau)e^{-i\tau H}f\right\Vert^2d\tau\hskip-2pt\right)^{1/2}\hskip-14pt+Cs^{-{\delta}}\Vert f\Vert\to 0,
\end{aligned}
\end{equation}
In the same way we have as $t>s\to\infty$
\begin{equation}
\begin{aligned}
&\Vert e^{itH}P_a^{{\varepsilon}_a}(t)J_ae^{-itH_a}P^a_Mf
-
e^{isH}P_a^{{\varepsilon}_a}(s)J_ae^{-isH_a}P^a_Mf
\Vert\\
&\le\hskip-2pt M_2\hskip-2pt\left(\int_s^t\left\Vert\frac{1}{\sqrt{\tau}}Q_a^{{\varepsilon}_a}(\tau)J_ae^{-i\tau H_a}P^a_Mf\right\Vert^2d\tau\hskip-2pt\right)^{1/2}\hskip-14pt+Cs^{-{\delta}}\Vert f\Vert\to 0.
\end{aligned}
\end{equation}
We have proved the existence of
\begin{eqnarray}
&&\begin{aligned}
\Omega_af=\lim_{t\to\infty}P^a_Me^{itH_a}J_a^*P_a^{{\varepsilon}_a}(t)e^{-itH}f,
\end{aligned}\\
&&\begin{aligned}
W_af=\lim_{t\to\infty}e^{itH}P_a^{{\varepsilon}_a}(t)J_ae^{-itH_a}P^a_Mf.
\end{aligned}
\end{eqnarray}
Lemma \ref{Fourierintegralpseudodifferential} in the Appendix implies
$$
[J_a,P_a^{{\varepsilon}_a}(t)]f(x)=O(\langle t\rangle^{-\delta})\Vert f\Vert.
$$
This yields that for $f\in\HH$
\begin{equation}
\begin{aligned}
W_af&=\lim_{t\to\infty}e^{itH}P_a^{{\varepsilon}_a}(t)J_ae^{-itH_a}P^a_Mf\\
&=\lim_{t\to\infty}e^{itH}J_aP_a^{{\varepsilon}_a}(t)e^{-itH_a}P^a_Mf.
\end{aligned}
\label{waveoperator-aP}
\end{equation}
Using the identity $(x_a-tD_a)e^{-itT_a}=e^{-itT_a}x_a$, we can prove that for the constant ${\rho_2}(>\theta_2>\dots>\theta_N>0)$ in \eq{compatibilitycondition-copy} there exists $f_{\rho_2}\in\HH$ such that
$$
(\chi_a(D_a)^2-P_a^{{\varepsilon}_a}(t))e^{-itH_a}P^a_Mf_{\rho_2}\to0\ (t\to\infty),\quad \lim_{{\rho_2}\downarrow0}\Vert f_{\rho_2}-f\Vert=0
$$
so that we have in the limit $\rho_2\downarrow0$
\begin{equation}
\begin{aligned}
W_a&=\lim_{t\to\infty}e^{itH}J_ae^{-itH_a}P^a_M\chi_a(D_a)^2f.
\end{aligned}
\label{waveoperator-aP-2}
\end{equation}
The existence of $G_a$ and $K_a$ is proved similarly.
Lemma \ref{lemmaonzakxa} proves 1). That $W_a$ is an isometry on $P^a_ME_{T_a}([d_1^2/2,d_2^2/2])\HH$ follows from \eq{waveoperator-aP}-\eq{waveoperator-aP-2} and Lemma \ref{FourierconjugateFourierproduct}. Let $s\in \R$ be fixed. Then for $f\in \HH$
\begin{align}
e^{isH}W_{a}f&=\lim_{t\to\infty}e^{i(s+t)H}J_ae^{-i(s+t)H_a}P^a_M\chi_a(D_a)^2e^{isH_a}f\nonumber\\
&=W_{a}e^{isH_a}f.\nonumber
\end{align}
This gives 2). 
\end{proof}

\section{Short-range case}\label{shortrangecompleteness}

In this section we assume that all long-range pair potentials vanish: $V_{ij}^L(x)\equiv0$. Under this assumption we will prove the asymptotic completeness of short-range wave operators.
\begin{theorem}\label{shortrangecompletenesstheorem} Let 
Assumptions \ref{potentialdecay} and \ref{eigenfunctiondecay} be satisfied with
 $V_{ij}^L=0$ for any pair $(i,j)$, $1\le i<j\le N$.
Then the wave operators
\begin{equation}
W_af=\lim_{t\to\infty}e^{itH}e^{-itH_a}P^af\label{shortrangewaveop-1}
\end{equation}
exists for $f\in\HH$ and $2\le|a|\le N$, and the asymptotic completeness of wave operators holds.
\begin{equation}
\bigoplus_{2\le|a|\le N}\mathcal{R}(W_a)=\HH_c(H).\label{shortrangeasymptoticcompleteness-a}
\end{equation}
\end{theorem}
\begin{proof} By Theorem \ref{adjointmodifiedwaveoperator}, for arbitrary constants $0<d_1<d_2<\infty$ and the functions $\chi_a(\xi_a)$ $(2\le|a|\le N)$ in \eq{chid}, wave operators
\begin{equation}
W_af=\lim_{t\to\infty}e^{itH}e^{-itH_a}P^a\chi_a(D_a)^2f\label{shortrangewaveop}
\end{equation}
exist for $f\in\HH$ and $2\le|a|\le N$.
As $\chi_a(\xi_a)=1$ for $|\xi_a|\in[d_1,d_2]$ we can extend the isometry $W_a$ on $E_{T_a}([d_1^2/2,d_2^2/2])P^a\HH$ to the whole of $P^a\HH$ with retaining the isometry, which gives the definition of wave operator $W_a$ in the theorem. The intertwining property \eq{intertwiningproperty} implies 
$\mathcal{R}(W_a)\subset\HH_c(H)$.
We have to show the reverse relation
\begin{align}
\HH_c(H)\subset\bigoplus_{2\le|a|\le N}\mathcal{R}(W_a).\label{shortrangeasymptoticcompleteness-a-R}
\end{align}
We prove this by mathematical induction on $|a|$. Our induction hypothesis is that \eq{shortrangeasymptoticcompleteness-a-R} holds with $H$ replaced by any $H^a$ with $2\le |a|\le N$. When $|a|=N$, $P^a=I$ so there occurs no scattering: $\HH_c(H^a)=\{0\}$, and the result is trivial. Assuming \eq{shortrangeasymptoticcompleteness-a-R} for $2\le |a|\le N$, we will prove \eq{shortrangeasymptoticcompleteness-a-R} for $|a|=1$, i.e. for $H$.
We recall relation \eq{continuousspectralsubspace}.
\begin{equation}
\begin{aligned}
\HH_c(H)=\overline{\sum_{B\Subset \R\setminus\mathcal{T}}E_H(B)\HH}.\label{continuousspectralsubspace-1}
\end{aligned}
\end{equation}
It thus suffices to show 
for given $B\Subset\R\setminus\mathcal{T}$ and $f=E_H(B)\in\HH_c(H)$
\begin{equation}
f\in\bigoplus_{2\le|a|\le N}\mathcal{R}(W_a).\label{shortrangeasymptoticcompleteness}
\end{equation}
By Theorem \ref{eigenprojectionisemergent}, there exist constants $0<d_1<d_2<\infty$ such that for the sequence $t_m$ in Theorem \ref{quantumclassical} and sufficiently small ${\varepsilon}_a>0$
\begin{equation}
e^{-itH}f\sim\sum_{2\le|a|\le N}P_a^{{\varepsilon}_a}(t)e^{-itH}f
\label{asymptoticbehavior}
\end{equation}
asymptotically as $t=t_m\to\infty$. On the other hand,
Theorem \ref{adjointmodifiedwaveoperator} implies that the following limit exists.
\begin{equation}
\begin{aligned}
\Omega_af=\lim_{t\to\infty}e^{itH_a}P_a^{{\varepsilon}_a}(t)e^{-itH}f \quad(f\in\HH).
\end{aligned}\label{usualinversewave-2}
\end{equation}
Therefore using \eq{shortrangewaveop-1} we have as $t\to\infty$
\begin{equation}
\begin{aligned}
P_a^{{\varepsilon}_a}(t)e^{-itH}f &\sim e^{-itH_a}\Omega_af= e^{-itH_a}P^a\Omega_af\oplus e^{-itH_a}(I-P^a)\Omega_af\\
&\sim e^{-itH}W_a\Omega_af+e^{-itH_a}(I-P^a)\Omega_af
\end{aligned}
\label{reductiontofewerbodies}
\end{equation}
Combining this with \eq{asymptoticbehavior} gives as $t=t_m\to\infty$
\begin{align}
e^{-itH}f\sim\sum_{2\le|a|\le N}\left\{e^{-itH}W_a\Omega_af+ (e^{-itT_a}\otimes e^{-itH^a}(I-P^a))\Omega_af\right\}.\label{1st}
\end{align}
Our induction hypothesis \eq{shortrangeasymptoticcompleteness-a-R} for $2\le|a|\le N$ implies for any $h\in L^2(X^a)$ and $a_1<a$ there exists some $w_{a_1}h\in L^2(X^a)$ such that
\begin{align}
e^{-itH^a}(I-P^a)h\sim \sum_{a_1<a}e^{-itH_{a_1}^a}P^{a_1}w_{a_1}h
=\sum_{a_1<a}(e^{-itT_{a_1}^a}\otimes e^{-itH^{a_1}}P^{a_1})w_{a_1}h
\label{IH-consequence}
\end{align}
as $t\to\infty$, where $H_{a_1}^a=T_{a_1}^a+H^{a_1}$, $T_{a_1}^a=T_{a_1}-T_a$. Applying this to \eq{1st}, we obtain
\begin{equation}
\begin{aligned}
e^{-itH}f&\sim\sum_{2\le|a|\le N}
\biggl\{
e^{-itH}W_a \Omega_af +
 \bigl(e^{-itT_a}\otimes 
\sum_{a_1<a}
(e^{-itT_{a_1}^a}\otimes e^{-itH^{a_1}}P^{a_1})w_{a_1}\bigr)\Omega_af\biggr\}\\
&=\sum_{2\le|a|\le N}
\biggl\{
e^{-itH}W_a \Omega_af+
\sum_{a_1<a} (e^{-itT_{a_1}}\otimes 
e^{-itH^{a_1}}P^{a_1})w_{a_1}\Omega_af\biggr\}\\
&=\sum_{2\le|a|\le N}
\biggl\{
e^{-itH}W_a \Omega_af+
\sum_{a_1<a} e^{-itH_{a_1}}P^{a_1}w_{a_1}\Omega_af\biggr\}.
\end{aligned}\label{1st-2}
\end{equation}
Thus
\begin{align}
f=\sum_{2\le|a|\le N}\left\{W_a\Omega_af+\sum_{a_1<a}W_{a_1}w_{a_1}\Omega_af\right\}\in\bigoplus_{2\le|a|\le N}\mathcal{R}(W_a).
\end{align}
This completes the proof.
\end{proof}
The proof above uses the existence of the limit $\Omega_a$, which makes it possible to decompose the second term on the RHS of \eq{reductiontofewerbodies}. The proof of the existence of $\Omega_a$ $=\lim_{t\to\infty}$ $e^{itH_a}J_a^*P_a^{{\varepsilon}_a}(t)e^{-itH}f$ in the long-range case requires strong condition on the decay rate as one needs to prove the existence of the limit without the eigenprojection $P^a$ in $\Omega_a$. This makes it necessary to analyze the internal motion of each cluster and assume the decay rate $\delta>\sqrt{3}-1$ of long-range part of the pair potentials as in Derezi\'nski \cite{[De]}. It is however possible to eliminate the second term on the RHS of \eq{reductiontofewerbodies} by appealing to the concept of scattering spaces before\footnote{This is the original idea developed in the preprint \cite{K3} in 1984, in which it is attempted to prove the asymptotic completeness without using the existence of particular limits except for wave operators and with utilizing the asymptotic orthogonality of different scattering spaces.} taking any limit as $t\to\infty$ in \eq{asymptoticbehavior}.
 This will make it possible to deal with the long-range tail without worrying about the internal motion and prove the asymptotic completeness for some long-range pair potentials with $\delta>1/2$ as will be done in the next section.

\section{Scattering spaces -- Long-range case}\label{scateringspaces}

We define scattering spaces following \cite{[Kitada-S]}.
In the following we consider the case $t\to\infty$ only. The case $t\to-\infty$ is treated similarly. As before we use the notation $f(t)\sim g(t)$ as $t\to\infty$ to mean that $\Vert f(t)-g(t)\Vert\to 0$ as $t\to\infty$ for $\HH$-valued functions $f(t)$ and $g(t)$ of $t>1$.
\begin{definition}\label{scatteringspaces} Let real numbers $r, \sigma, {\mu}$ and a cluster decomposition $a$ satisfy $0\le r\le 1$, $\sigma, {\mu} >0$ and $2\le |a|\le N$. Let $B\Subset \R\setminus\mathcal{T}$ be a closed set.
\begin{namelist}{888}
\item[ i)] We define $S_a^{r\sigma{\mu}}(B)$ for $0<r\le 1$ by
\begin{equation}
\begin{aligned}
S_a^{r\sigma{\mu}}(B)=\{&f\in E_H(B)\HH \ |\ \\ 
 &e^{-itH}f\sim \prod_{(i,j)\not\leq a}F(|x_{ij}|\ge \sigma t)F(|x^a|\le{\mu} t^r)e^{-itH}f\ \mbox{as}\ t\to\infty\}.
\end{aligned}\label{2.1}
\end{equation}
For $r=0$ we define $S_a^{0\sigma}(B)$ by
\begin{equation}
\begin{aligned}
S_a^{0\sigma}(&B)=\{f\in E_H(B)\HH \ |\ \\
&\lim_{R\to\infty}\limsup_{t\to\infty}\Bigl\Vert e^{-itH}f - \prod_{(i,j)\not\leq a}F(|x_{ij}|\ge \sigma t)F(|x^a|\le R)e^{-itH}f\Bigr\Vert=0 \}.
\end{aligned}\label{2.2}
\end{equation}
\item[ ii)]
We define the localized scattering space $S_a^r(B)$ of order $r\in(0,1]$ for $H$ as the closure of
\begin{equation}
\begin{aligned}
\sum_{\sigma>0}\bigcap_{{\mu}>0}S_a^{r\sigma{\mu}}&(B)
=\{f\in  E_H(B)\HH \ |\ \exists \sigma>0, \forall {\mu}>0:\\
&\ e^{-itH}f\sim \prod_{(i,j)\not\leq a}F(|x_{ij}|\ge \sigma t)F(|x^a|\le{\mu} t^r)e^{-itH}f\ \text{as}\ t\to\infty\}.
\end{aligned}\label{2.3}
\end{equation}
$S_a^0(B)$ is defined as the closure of
\begin{equation}
\begin{aligned}
\sum_{\sigma>0}&S_a^{0\sigma}(B)=\{f\in E_H(B)\HH\ |\ \exists \sigma>0:\\
&\lim_{R\to\infty}\limsup_{t\to\infty}\Bigl\Vert e^{-itH}f - \prod_{(i,j)\not\leq a}F(|x_{ij}|\ge \sigma t)F(|x^a|\le R)e^{-itH}f\Bigr\Vert=0 \}.
\end{aligned}\label{2.4}
\end{equation}
\item[  iii)] We define the scattering space $S_a^r$ of order $r\in [0,1]$ for $H$ as the closure of
\begin{equation}
\sum_{B\Subset \R\setminus\mathcal{T}}S_a^r(B).\label{2.5}
\end{equation}
\end{namelist}
\end{definition}
We note that $S_a^{r\sigma{\mu}}(B)$, $S_a^{0\sigma}(B)$, $S_a^{r}(B)\subset E_H(B)\HH$ and $S_a^{r}\subset \HH_c(H)$ define closed subspaces.
\begin{proposition}\label{Proposition 2.2} Let Assumptions \ref{potentialdecay} and \ref{eigenfunctiondecay} be satisfied. Let $B\Subset \R\setminus\mathcal{T}$ and $f\in S_a^{r\sigma{\mu}}(B)$ for $0< r\le 1$ or $f\in S_a^{0\sigma}(B)$ for $r=0$ with $\sigma,{\mu}>0$ and $2\le |a|\le N$. Then the following limit relations hold:
\begin{namelist}{888}
\item[  {\rm i)}] Let $(i,j)\not\leq a$. Then for $0<r\le 1$ we have when $t\to\infty$
\begin{equation}
F(|x_{ij}|<\sigma t)F(|x^a|\le{\mu} t^r)e^{-itH}f\to 0.\label{2.6}
\end{equation}
For $r=0$ we have
\begin{equation}
\lim_{R\to\infty}\limsup_{t\to\infty}\left\Vert F(|x_{ij}|< \sigma t)F(|x^a|\le R)e^{-itH}f\right\Vert=0.\label{2.7}
\end{equation}
\item[ {\rm ii)}] For $0<r\le 1$ we have when $t\to\infty$
\begin{equation}
F(|x^a|> {\mu} t^r)e^{-itH}f\to 0.\label{2.8}
\end{equation}
For $r=0$
\begin{equation}
\lim_{R\to\infty}\limsup_{t\to\infty}\left\Vert F(|x^a|> R)e^{-itH}f\right\Vert=0.\label{2.9}
\end{equation}
\item[ {\rm iii)}] There exists a sequence $t_m\to\infty$ as $m\to\infty$ depending on $f\in S_a^{r\sigma{\mu}}(B)$ or $f\in S_a^{0\sigma}(B)$ such that
\begin{equation}
\left\Vert \left(\varphi\left({x_a}/{t_m}\right)-\varphi(D_a)\right)e^{-it_mH}f\right\Vert\to 0\quad \text{as}\quad m\to\infty \label{2.10}
\end{equation}
for any function $\varphi\in C_0^\infty(\R^{\nu(|a|-1)})$.
\end{namelist}
\end{proposition}
\begin{proof}
i) and ii) are clear from the definition of $S_a^{r\sigma{\mu}}(B)$ or $S_a^{0\sigma}(B)$. We prove iii). 
Since $f\in E_H(B)\HH\subset H_c(H)$, we have by \eq{summationofprojections} and $f\in S_a^{r\sigma{\mu}}(B)$ (or $f\in S_a^{0\sigma}(B)$)
\begin{equation}
e^{-itH}f\sim \sum_{d\le a}{\tilde P}^{\ d}_{M_d}e^{-itH}f\label{2.11}
\end{equation}
as $t\to\infty$.
Theorem \ref{quantumclassical} and the restriction $d\le a$ in the sum of the RHS of \eq{2.11} imply \eq{2.10}.
\end{proof}
The following propositions are obvious by definition.
\begin{proposition}\label{Proposition 2.3} Let $2\le |a|\le N$. If $1\ge r'\ge r> 0$, $\sigma\ge \sigma'>0$ and ${\mu}'\ge {\mu}>0$ and $B\Subset \R\setminus\mathcal{T}$, then $S_a^{0\sigma}(B)\subset S_a^{0\sigma'}(B)$, $S_a^{0\sigma}(B)\subset S_a^{r\sigma{\mu}}(B)\subset S_a^{r'\sigma'{\mu}'}(B)$, $S_a^0(B)\subset S_a^r(B)\subset S_a^{r'}(B)$, $S_a^0(B)\subset S_a^r(B)\subset S_a^{r}$, and $S_a^0\subset S_a^r\subset S_a^{r'}$.
\end{proposition}
\begin{proposition}\label{Proposition 2.4} Let $a$ and $a'$ be different cluster decompositions: $a\ne a'$. Then for any $0\le r,r' \le 1$, $S_a^r$ and $S_{a'}^{r'}$ are orthogonal: $S_a^r\perp S_{a'}^{r'}$.
\end{proposition}
Now we will prove the asymptotic completeness for the long-range case in a series of Theorems. In the following we always assume that the constants $\rho_j>\theta_j>\sigma>0$ and ${\varepsilon}_a>0$ satisfy the condition \eq{compatibilitycondition} for the constants $d_2>d_1>0$ determined by a given set $B\Subset\R\setminus\mathcal{T}$.
\begin{theorem}\label{sumG_a^+=H_c(H)}
Let Assumptions \ref{potentialdecay} and \ref{eigenfunctiondecay} be satisfied. 
Let $f=E_H(B)f\in \HH_c(H)$ with $B\Subset\R\setminus\mathcal{T}$, and let $d_2>d_1>0$ be the constants specified in Theorem \ref{eigenprojectionisemergent} corresponding to the set $B$. Let $\chi_a$ be defined by \eq{chid} for the constants $d_2>d_1>0$. Let $P_a^{{\varepsilon}_a}(t)$ $(2\le|a|\le N)$ be defined by \eq{(4)}-\eq{Paepsilont} with this $\chi_a$, and let $G_a$ be defined by \eq{perturbedlimit} with this $P_a^{{\varepsilon}_a}(t)$. Then we have
\begin{align}
f=\bigoplus_{2\le |a|\le N}G_af.\label{Gdecomposition}
\end{align}
Further we have
\begin{align}
\HH_c(H)=\bigoplus_{2\le|a|\le N}S_a^1.\label{decompositionofH_cbyS_a^1-0}
\end{align}
\end{theorem}
\begin{proof}
The relation \eq{Gdecomposition} follows from Theorem \ref{eigenprojectionisemergent} and Theorem \ref{adjointmodifiedwaveoperator}. From the definition \eq{perturbedlimit} of $G_a$, we see that $G_af$ approximates an element of $S_a^1(B)$ by taking ${\varepsilon}_a>0$ in \eq{perturbedlimit} small, and the approximation is better when ${\varepsilon}_a>$ is smaller.
Hence from \eq{Gdecomposition} and the reverse inclusion $S_a^1(B)\subset E_H(B)\HH$, we have
\begin{align}
E_H(B)\HH=\bigoplus_{2\le|a|\le N}S_a^1(B).\label{decompositionofH_cbyS_a^1}
\end{align}
This holds for arbitrary $B\Subset\R\setminus\mathcal{T}$. Hence we have \eq{decompositionofH_cbyS_a^1-0} by
\begin{equation}
\begin{aligned}
\HH_c(H)=\overline{\sum_{B\Subset \R\setminus\mathcal{T}}E_H(B)\HH}\label{continuousspectralsubspace-2}
\end{aligned}
\end{equation}
and
\begin{align}
S_a^1=\overline{\sum_{B\Subset \R\setminus\mathcal{T}}S_a^1(B)}.\label{2.5-2}
\end{align}
\end{proof} 
\begin{theorem}\label{S_^1(B)asymptoticbehavior}
Let Assumptions \ref{potentialdecay} and \ref{eigenfunctiondecay} be satisfied. Let $B\Subset\R\setminus\mathcal{T}$ and $f=E_H(B)f\in S_a^1(B)$. 
Let $d_2>d_1>0$ be the constants specified in Theorem \ref{eigenprojectionisemergent} corresponding to the set $B$. Let $\chi_a$ be defined by \eq{chid} for the constants $d_2>d_1>0$ and let the constants in $\Phi_a$ be chosen suitably as in the proof of Theorem \ref{eigenprojectionisemergent}. Let $P_a^{{\varepsilon}_a}(t)$ $(2\le|a|\le N)$ be defined by \eq{(4)}-\eq{Paepsilont} with this $\chi_a$ and $\Phi_a$.
Then for the sequence $t_m\to\infty$ in Theorem \ref{quantumclassical} we have as $m\to\infty$
\begin{align}
&\bigl\Vert e^{-it_mH}f-\Phi_a(x/t_m)e^{-it_mH}f\bigr\Vert\to0,
\label{oneofdesiredones0copy}\\
&\bigl\Vert e^{-it_mH}f-P_a^{{\varepsilon}_a}(t_m)e^{-it_mH}f\bigr\Vert\to0.\label{asymptotics0copy}
\end{align}
In particular the following limits exist for $f\in S_a^1(B)$ and we have
\begin{align}
&f=\lim_{t_m\to\infty}e^{it_mH}\Phi_a(x/t_m)e^{-it_mH}f\in S_a^1(B)\label{f=limPhi_a},\\
&f=\lim_{t_m\to\infty}e^{it_mH}P_a^{{\varepsilon}_a}(t_m)e^{-it_mH}f\in S_a^1(B).\label{f=limP_a}
\end{align}
\end{theorem}
\begin{proof}
By
Theorem \ref{eigenprojectionisemergent}-\eq{oneofdesiredones0} and Theorem \ref{chidonPdMd},
 we have for $f=E_H(B)f\in S_a^1(B)$ as $m\to\infty$
\begin{align}
\bigl\Vert e^{-it_mH}f-\sum_{2\le|a'|\le N}\Phi_{a'}(x/t_m)\sum_{d\le a'}\chi_d(D_d)\tilde P^d_{M_d^m}e^{-it_mH}f\bigr\Vert\to0.\label{oneofdesiredones0proof}
\end{align}
If $a\not\le a'$, there is a pair $(i,j)\not\le a'$ with $(i,j)\le a$. By the definition of $S_a^1$ we have $|x_{ij}|\le |x^a|\le{\mu} t$ for any ${\mu}>0$ while arguing similarly to the proof of Theorem \ref{eigenprojectionisemergent}, by the factor $\Phi_{a'}(x/t_m)$ we obtain $|x_{ij}|\ge ct_m$ for any $(i,j)\not\le a'$ with some constant $c>0$. Thus the terms $\Phi_{a'}$ with $a\not\le a'$ in the above sum must vanish asymptotically as  $t=t_m\to\infty$. On the other hand if $a'\not\le a$, there is a pair $(i,j)\not\le a$ such that $|x_{ij}|\le |x^{a'}|$.
The condition $f\in S_a^1$ gives $|x_{ij}|\ge \sigma t$ for some $\sigma>0$.
The factor $\Phi_{a'}$ yields that $|x_{ij}|\le |x^{a'}|\le \omega_{|a'|}^{-1}|x_{a'}|\le \omega_{|a'|}^{-1}|x_d|$ by $d\le a'$. Thus the factor $\chi_d(D_d)$ gives $|x_{ij}|\le 2\omega_{|a'|}^{-1}d_2t_m$ for the constant $d_2>0$ in Theorem \ref{chidonPdMd}, and the terms with $a'\not\le a$ in the sum of \eq{oneofdesiredones0proof} vanish asymptotically as $t=t_m\to\infty$ if we take the constants $\theta_{|a'|}>0$ and $\sigma>0$ in the definition of $\Phi_{a'}$ small enough. These with Theorem \ref{chidonPdMd} yield 
\begin{align}
\bigl\Vert e^{-it_mH}f-\Phi_a(x/t_m)\sum_{d\le a}\tilde P^d_{M_d^m}e^{-it_mH}f\bigr\Vert\to0.
\label{oneofdesiredones0copy2}
\end{align}
Recovering the terms with $d\not\le a$ similarly to the proof of Theorem \ref{eigenprojectionisemergent} proves
\eq{oneofdesiredones0copy}. From \eq{oneofdesiredones0copy2} follows \eq{asymptotics0copy} in the same way as in the proof of Theorem \ref{eigenprojectionisemergent}.
\end{proof}
\begin{theorem}\label{rangeofW=Sa0}
Let Assumptions \ref{potentialdecay} and \ref{eigenfunctiondecay} be satisfied. Let $B\Subset\R\setminus\mathcal{T}$ and $f=E_H(B)f\in S_a^0(B)$. Let $t_m$ and $M^m$ be th sequences specified in Theorem \ref{quantumclassical} for $f\in \HH_c(H)$.
 Then
\begin{equation}
\begin{aligned}
f=
\lim_{m\to\infty}
e^{it_mH}P^a_{M^m}P_a^{{\varepsilon}_a}(t_m)e^{-it_mH}f.
\end{aligned}\label{asymptoticeigenprojectedevolutedstateSa0}
\end{equation}
In particular we have
\begin{equation}
\mathcal{R}(W_a)=S_a^0.\label{RW=Sa0}
\end{equation}
\end{theorem}
\begin{proof}
By \eq{summationofprojections} and $f\in S_a^0(B)\subset \HH_c(H)$ we have 
\begin{equation}
e^{-itH}f\sim \sum_{d\le a}{\tilde P}^{\ d}_{M_d}e^{-itH}f\label{2.11-2}
\end{equation}
as $t\to\infty$. This, Theorem \ref{quantumclassical}-\eq{(2.7)} and $f\in \HH_c(H)$ imply
\begin{equation}
e^{-it_mH}f\sim {P}^{\ a}_{M^m}e^{-it_mH}f\label{2.11-4}
\end{equation}
as $m\to\infty$. Using Theorem \ref{S_^1(B)asymptoticbehavior}-\eq{asymptotics0copy}, we obtain from this that
\begin{equation}
\begin{aligned}
f=
\lim_{m\to\infty}
e^{it_mH}P^a_{M^m}P_a^{{\varepsilon}_a}(t_m)e^{-it_mH}f.
\end{aligned}\label{asymptoticeigenprojectedevolutedstateSa0Pa}
\end{equation}

By the definition \eq{wave} of $W_a$ it is clear that $\mathcal{R}(W_a)\subset S_a^0$. Thus we have only to show the reverse inclusion $S_a^0\subset \mathcal{R}(W_a)$. Let $B\Subset \R\setminus\mathcal{T}$ and let $f=E_H(B)f\in S_a^0(B)$. By Theorem \ref{S_^1(B)asymptoticbehavior}-\eq{asymptotics0copy}, \eq{asymptoticeigenprojectedevolutedstateSa0}, $\Vert[P_a^{\varepsilon}(t_m),P^a_{M^m}]\Vert\to0$ $(t_m\to\infty)$, Lemma \ref{FourierconjugateFourierproduct} and Theorem \ref{adjointmodifiedwaveoperator}-\eq{inversewave} and \eq{wave} we have
\begin{equation}
\begin{aligned}
f&=\lim_{m\to\infty} e^{it_mH}P_a^{{\varepsilon}}(t_m)J_ae^{-it_mH_a}P^a_{M^m}e^{it_mH_a} J_a^*P_a^{{\varepsilon}}(t_m)e^{-it_mH}f\\
&= W_a\Omega_af\in \mathcal{R}(W_a).
\end{aligned}
\end{equation}
The proof is complete.
\end{proof}

Let $\mathcal{T}_a$ $(1\le|a|\le N)$ be defined by
\begin{align}
\mathcal{T}_a=\bigcup_{b\le a}\sigma_p(H^b).\label{subthresholdsandeigenvalues}
\end{align}
We note that $0\in \mathcal{T}_a$ $(1\le |a|\le N)$, and $\mathcal{T}_{a}=\mathcal{T}$ when $|a|=1$. 
Let $\psi_a(\lambda)\in C_0^\infty(\R)$ satisfy 
\begin{equation}
\begin{aligned}
&1\ge \psi_a(\lambda)\ge 0,\\
&\psi_a(\lambda)=1\mbox{ in a neighborhood of each }\lambda\in \mathcal{T}_a.
\end{aligned}\label{psi_a}
\end{equation}
For any real numbers ${\rho}>0$ and $y\in \R$, set $B_{\rho}(y)=\{x|x\in \R, |x-y|<{\rho}\}$.
We write supp $\psi_a\downarrow \mathcal{T}_a$ if and only if for any ${\rho}>0$ there is $\psi_a$ satisfying \eq{psi_a} and a finite covering $\{B_{{\rho}_j}(y_j)\}_{j=1}^L$ of supp $\psi_a$ such that for $1\le j\le L$, ${\rho}>{\rho}_j>0$, $y_j\in\mathcal{T}_a$,  $\mathcal{T}_a\subset$ supp $\psi_a\subset \bigcup_{j=1}^L B_{{\rho}_j}(y_j)$, and $\sum_{j=1}^L{\rho}_j<{\rho}$. Since $\mathcal{T}_a$ is a closed countable subset of $[b,0]$ for some $b\le0$ by \cite{[FH]} and \cite{PSS}, and non-threshold eigenvalues can accumulate only at thresholds by Theorem 1.5 of \cite{PSS}, it is possible to take such covering of supp $\psi_a$ for any ${\rho}>0$.
 We notice that 
\begin{align}
\lim_{\mbox{\scriptsize supp }\psi_a\downarrow \mathcal{T}_a}\psi_a(H^a)h=P^ah=\lim_{M\to\infty}P^a_Mh\label{slimrelation}
\end{align}
 for $h\in \HH^a$. In fact since $\psi_a(H^a)-P^a=\psi_a(H^a)(I-P^a)$ and $H^a$ has no singular continuous spectrum by Theorem 1.5 of \cite{PSS}, we have only to consider the operator $\psi_a(H^a)$ restricted to the absolutely continuous subspace $\HH^a_{ac}(H^a)$ of $H^a$. When supp $\psi_a \downarrow \mathcal{T}_a$, the Lebesgue measure of supp $\psi_a$ goes to zero by the above definition. Thus for $h\in \HH^a_{ac}(H^a)$ we have $\Vert\psi_a(H^a)h\Vert\le\Vert E_{H^a}(\mbox{supp }\psi_a)h\Vert^2\to0$ when supp $\psi_a \downarrow \mathcal{T}_a$ uniformly with respect to the way of letting supp $\psi_a \downarrow \mathcal{T}_a$. This proves \eq{slimrelation}.
\begin{theorem}\label{emergentP^a}
Let the assumptions of Theorem \ref{Wapm=Sa1} be satisfied.
Let $B\Subset\R\setminus\mathcal{T}$. 
 Then for $f=E_H(B)f\in S_a^1(B)$ the following limit exists and we have
\begin{equation}
\begin{aligned}
f=
\lim_{t\to\infty}
e^{itH}P^aP_a^{{\varepsilon}_a}(t)e^{-itH}f.
\end{aligned}\label{asymptoticeigenprojectedevolutedstateSa1}
\end{equation}
In particular we have
\begin{equation}
S_a^1=S_a^0.\label{Sa1=Sa0}
\end{equation}
\end{theorem}
\begin{proof}
Let $t_m\to\infty$ and $M_{a}^m\to\infty$ be the sequences in Theorem \ref{quantumclassical} for $f\in S_a^1(B)\subset\HH_c(H)$. Let $\psi_a(\lambda)\in C_0^\infty(\R)$ satisfy \eq{psi_a}.
We decompose
\begin{align}
e^{-itH}f=\psi_a(H^a)e^{-itH}f+(I-\psi_a(H^a))e^{-itH}f.\label{10-1}
\end{align}
We will use the same convention as in Theorem \ref{eigenprojectionisemergent}: $x^a=(x_d^a,x^d)=(x_d-x_a,x^d)=(x^a-x^d,x^d)$ $(d\le a)$, $D_d=D_{x_d}$, $D_d^a=D_{x_d^a}$, $T_d^a=T_d-T_a$, $H_d^a=T_d^a+H^d$. The functions $\Phi^a_{b}(x^a)$ $(b<a)$ are defined as in section \ref{partitionofunity} for the variable $x^a\in X^a=\R^{\nu(N-|a|)}$ which form a partition of unity such that $\sum_{b<a}\Phi^a_{b}(x^a/t)=1$.
Now 
applying Theorem \ref{S_^1(B)asymptoticbehavior}-\eq{oneofdesiredones0copy} or \eq{oneofdesiredones0copy2} and Theorem \ref{quantumclassical}, we have the following asymptotic relation for $t=t_m$ and $M_d=M_d^m$ when $m\to\infty$.
\begin{equation}
\begin{aligned}
(I-\psi_a(H^a))e^{-itH}f
&\sim (I-\psi_a(H^a))\Phi_a(x_a/t)\sum_{d\le a}\tilde P^d_{M_d}e^{-itH}f\\
&=\Phi_a(x_a/t)(I-\psi_a(H^a))\sum_{d\le a}\tilde P^d_{M_d}e^{-itH}f.
\end{aligned}
\end{equation}
Noting that $(I-\psi_a(H^a))P^a_M=0$ we obtain from this
\begin{equation}
\begin{aligned}
(I-\psi_a(H^a))e^{-itH}f
& 
=\Phi_a(x_a/t)(I-\psi_a(H^a))\sum_{d<a}\tilde P^d_{M_d}e^{-itH}f\\&=\Phi_a(x_a/t)\sum_{b<a}\Phi_b^a(x^a/t)(I-\psi_a(H^a))\sum_{d< a}\tilde P^d_{M_d}e^{-itH}f\\
&\sim\Phi_a(x_a/t)\sum_{b<a}\Phi_b^a(x^a/t)\sum_{d< a}(I-\psi_a(H^a_d))\tilde P^d_{M_d}e^{-itH}f.
\end{aligned}\label{10-2}
\end{equation}
As $\psi_a(\lambda)=1$ in a neighborhood of each $\lambda\in \mathcal{T}_a$, $(I-\psi_a(H_d^a))\tilde P^d_{M_d}\ne0$ implies $T_d^a\ge d_1^2/2>0$ for some constant $d_1>0$. On the other hand Theorem \ref{chidonPdMd} gives $d_2^2/2\ge T_d$ for some $d_2>d_1>0$. Thus $d_2\ge |D_d|\ge |D_d^a|\ge d_1>0$ asymptotically on each state 
$(I-\psi_a(H^a_d))\tilde P^d_{M_d}e^{-itH}f$ $(d<a)$ as $t=t_m\to\infty$. Hence Theorem \ref{quantumclassical} gives $d_2\ge |x_d^a/t_m|\ge d_1$, $|x^a/t_m|\sim|x_d^a/t_m|$ and $|x^a|\ge d_1t_m$ asymptotically. So if $d<a$ satisfies $d\not\le b(<a)$, 
we can find a pair $(i,j)\le d$ with $(i,j)\not\le b$ so that by the factor $\Phi_b^a(x^a/t)$ and $|x^a|\ge d_1t_m$ we have $|x_{ij}|\ge ct$ for some constant $c>0$. Thus $|x^d|\ge|x_{ij}|\ge ct$, while the factor $\tilde P^d_{M_d}$ bounds $x^d$ by Theorem \ref{quantumclassical}-\eq{(2.6)}. Therefore the terms with $d\not\le b$ vanish asymptotically as $t=t_m\to\infty$ in the second sum of the RHS of \eq{10-2} and we get
\begin{equation}
\begin{aligned}
(I-\psi_a(H^a))e^{-itH}f
\sim \Phi_a(x_a/t)\sum_{b<a}\Phi_b^a(x^a/t)\sum_{d\le b}(I-\psi_a(H^a_d))\tilde P^d_{M_d}e^{-itH}f.
\end{aligned}\label{10-3}
\end{equation}
The RHS behaves like $\sum_{b<a}e^{-itH}g_b$ with $g_b\in S_b^{1\sigma{\mu}}$ $(b<a)$ for some $\sigma,{\mu}>0$ along the sequence $t=t_m\to\infty$. Thus we have from \eq{10-1}
\begin{equation}
\begin{aligned}
\Vert e^{-itH}f-\psi_a(H^a)e^{-itH}f\Vert^2
&=(e^{-itH}f-\psi_a(H^a)e^{-itH}f,(I-\psi_a(H^a))e^{-itH}f)\\
&\sim\sum_{b<a}(e^{-itH}f-\psi_a(H^a)e^{-itH}f,e^{-itH}g_b)
\end{aligned}\label{10-4}
\end{equation}
as $t=t_m\to\infty$ for some $g_b\in S_b^{1\sigma{\mu}}$ with $b<a$.
Theorem \ref{S_^1(B)asymptoticbehavior}-\eq{oneofdesiredones0copy} implies
\begin{equation}
\begin{aligned}
\psi_a(H^a)e^{-itH}f\sim \psi_a(H^a)\Phi_a(x_a/t)e^{-itH}f
=\Phi_a(x_a/t)\psi_a(H^a)e^{-itH}f
\end{aligned}\label{10-5}
\end{equation}
as $t=t_m\to\infty$, where $\Phi_a$ can have constants different from those for $\Phi_a$ in \eq{10-2}. This means that there exist $h\in S_a^{1\sigma{\mu}}$ and $\sigma>0$ such that $\psi_a(H^a)e^{-itH}f$ behaves like $e^{-itH}h$ for any ${\mu}>0$ along $t=t_m\to\infty$, which and $f\in S_a^1(B)$ imply that the factor $e^{-itH}f-\psi_a(H^a)e^{-itH}f$ in \eq{10-4} behaves like $e^{-itH}\tilde h$ with some $\tilde h\in S_a^{1\sigma{\mu}}$ and $\sigma>0$ for any ${\mu}>0$. Hence it is asymptotically orthogonal to $e^{-itH}g_b$ $(b<a)$, and the RHS of \eq{10-4} goes to $0$ as $t=t_m\to\infty$. We then have for $t=t_m\to\infty$
\begin{align}
e^{-itH}f\sim\psi_a(H^a)e^{-itH}f.
\end{align}
Thus
\begin{equation}
\begin{aligned}
f=\lim_{t_m\to\infty}e^{it_mH}\psi_a(H^a)e^{-it_mH}f.
\end{aligned}\label{intermediateexpression}
\end{equation}
From this and Theorem \ref{S_^1(B)asymptoticbehavior}-\eq{asymptotics0copy} follows that
\begin{align}
f=\lim_{t_m\to\infty}e^{it_mH}\psi_a(H^a)P_a^{{\varepsilon}_a}(t_m)e^{-it_mH}f.
\end{align}
In virtue of the factor $P_a^{{\varepsilon}_a}(t_m)$ this limit exists without taking the sequence $t_m$ as seen by the same argument as in the proof of Theorem \ref{adjointmodifiedwaveoperator} with using smooth operator technique.
\begin{align}
f=\lim_{t\to\infty}e^{itH}\psi_a(H^a)P_a^{{\varepsilon}_a}(t)e^{-itH}f.\label{psi_atoT_a}
\end{align}
Similarly by the same technique with using \eq{Ia(x)-Ia(xa,xa)} and $[H,P^a]=[H-\tilde H_a,P^a]$ where $\tilde H_a=H^a+T_a+I_a^L(x_a,0)$ and recalling that we assume that $P^a=P^a_M$ for some finite integer $M\ge0$, we can prove the existence of the following limit for $f\in \HH$.
\begin{align}
Qf:=\lim_{t\to\infty}e^{itH}P^aP_a^{{\varepsilon}_a}(t)e^{-itH}f.\label{P^a_MtoP^a}
\end{align}
We let $\varphi$ be a $C_0^\infty(\R)$-function such that $0\le\varphi(\lambda)\le 1$ and $\varphi(\lambda)=1$ for $\lambda\in\mbox{supp }\psi_a$, $P^a_M=\sum_{j=1}^MP^a_j$ with $P^a_j$ being a finite dimensional eigenprojection onto the eigenspace spanned by eigenfunctions of $H^a$ with eigenvalue $E_j$, and $\psi_{aj}(\lambda)$ be a $C_0^\infty$-function equal to $1$ near $\lambda=E_j$ such that $\sum_{j=1}^M\psi_{aj}=\psi_a$.
We set for $j=1,\dots,M$
\begin{equation}
\Omega_jf=\lim_{\mbox{\scriptsize supp }\psi_a\downarrow\mathcal{T}_a}\lim_{t\to\infty}e^{itH}\varphi(H^a)(\psi_{aj}(H^a)-P^a_j)P_a^{{\varepsilon}_a}(t)e^{-itH}f.\label{Omegaj}
\end{equation}
Then we have for $f=E_H(B)f\in S_a^1(B)$
\begin{equation}
\begin{aligned}
f-Qf&=\lim_{\mbox{\scriptsize supp }\psi_a\downarrow\mathcal{T}_a}\lim_{t\to\infty}e^{itH}(\psi_a(H^a)-P^a)P_a^{{\varepsilon}_a}(t)e^{-itH}f.\label{f-Qf}\\
&=\lim_{\mbox{\scriptsize supp }\psi_a\downarrow\mathcal{T}_a}\lim_{t\to\infty}e^{itH}\varphi(H^a)(\psi_a(H^a)-P^a)P_a^{{\varepsilon}_a}(t)e^{-itH}f\\
&=\sum_{j=1}^M\Omega_jf.
\end{aligned}
\end{equation}
Let $U_a(t)$ be the fundamental solution of the equation
\begin{equation}
i^{-1}\partial_tU_a(t)+H_a(t)U_a(t)=0,\quad U_a(0)=I,
\end{equation}
where $H_a(t)=H^a+T_a+I_a^L(t,x)$ and
\begin{equation}
I_a^L(t,x)=I_a^L(x_a,x^a)\prod_{(i.j)\not\le a}\chi_0(100x_{ij}/({\varepsilon}_at))
\end{equation}
with $\chi_0(x)$ $(x\in\R^\nu)$ being the function defined by \eq{5.4}.
Then as $t\to\infty$
\begin{equation}
\begin{aligned}
e^{-itH}(f-Qf)&=e^{-itH}\sum_{j=1}^M\Omega_jf\\
&\sim U_a(t)\lim_{\mbox{\scriptsize supp }\psi_a\downarrow\mathcal{T}_a}\lim_{t\to\infty}U_a(t)^*\varphi(H^a)(\psi_a(H^a)-P^a)P_a^{{\varepsilon}_a}(t)e^{-itH}f.
\end{aligned}
\end{equation}
Set
\begin{equation}
\tilde \Omega_jf=\lim_{\mbox{\scriptsize supp }\psi_a\downarrow\mathcal{T}_a}\lim_{t\to\infty}U_a(t)^*\varphi(H^a)(\psi_{aj}(H^a)-P^a_j)P_a^{{\varepsilon}_a}(t)e^{-itH}f.
\end{equation}
Then as $t\to\infty$
\begin{equation}
\Vert(H^a-E_j)U_a(t)\tilde \Omega_j f\Vert\to0\label{mugenendenoshuusoku}
\end{equation}
and
\begin{equation}
e^{-itH}(f-Qf)=e^{-itH}\sum_{j=1}^M\Omega_jf\sim U_a(t)\sum_{j=1}^M\tilde \Omega_j f.
\end{equation}
Let $g=\sum_{j=1}^M\tilde \Omega_j f$ and $\tilde U_{aj}(t)$ be the fundamental solution of
\begin{equation}
i^{-1}\partial_t\tilde U_{aj}(t)+(E_j+T_a+I_a^L(t,x))\tilde U_{aj}(t)=0,\quad \tilde U_{aj}(0)=I.
\end{equation}
\begin{lemma}\label{bootstrap}
For $g=\sum_{j=1}^M\tilde \Omega_j f$ there exists a constant $C>0$ such that for any $j=1,\dots,M$ and $t>1$
\begin{equation}
\Vert(H^a-E_j)U_a(t)g\Vert\le Ct^{-2\delta}.
\end{equation}
\end{lemma}
\begin{proof}
By $E_j\le0$ and $V^a\ge0$, we have
\begin{equation}
\begin{aligned}
\Vert |D^a|U_a(t)g\Vert^2&\le C (H^a_0U_a(t)g,U_a(t)g)\\
&\le C((H^a-E_j)U_a(t)g,U_a(t)g)+C((E_j-V^a)U_a(t)g,U_a(t)g)\\
&\le C\Vert(H^a-E_j)U_a(t)g\Vert\Vert g\Vert.
\end{aligned}\label{Dae-itHleHa-Eje-itH}
\end{equation}
Setting
\begin{equation}
f(t)=\Vert(H^a-E_j)U_a(t)g\Vert^2,
\end{equation}
we obtain
\begin{equation}
f'(t)=\frac{d}{dt}((H^a-E_j)^2U_a(t)g,U_a(t)g)=(i[I_a^L(t,x),(H^a-E_j)^2]U_a(t)g,U_a(t)g).
\end{equation}
The commutator equals
\begin{equation}
2\mbox{Re}\left((H^a-E_j)\biggl(\sum_{(i,j)\not\le a}\nabla_{x^a}V_{ij}^L(t,x_{ij})\cdot D^a+O(t^{-2-\delta})\biggr)\right).
\end{equation}
Thus
\begin{equation}
|f'(t)|\le Ct^{-1-\delta}\Vert(H^a-E_j)U_a(t)g\Vert \Vert |D^a|U_a(t)g\Vert+Ct^{-2-\delta}\Vert(H^a-E_j)U_a(t)g\Vert.\label{f7tleCt-1-delta}
\end{equation}
We have from \eq{Dae-itHleHa-Eje-itH}
\begin{equation}
|f'(t)|\le Ct^{-1-\delta}f(t)^{\frac{3}{4}}+Ct^{-2-\delta}f(t)^{\frac{1}{2}}.\label{f7tbound}
\end{equation}
This and \eq{mugenendenoshuusoku} give
\begin{equation}
f(t)=-\int_t^\infty f'(\tau)d\tau.\label{f7tintegration}
\end{equation}
Since $f(t)$ is uniformly bounded, this and \eq{f7tbound} yield
\begin{equation}
|f(t)|\le Ct^{-\delta}.
\end{equation}
Inserting this into \eq{f7tbound} and integrating by \eq{f7tintegration} we get
\begin{equation}
|f(t)|\le Ct^{-\delta(1+\frac{3}{4})}+Ct^{-1-\delta(1+\frac{1}{2})}.
\end{equation}
Repeating this procedure we finally arrive at the estimate
\begin{equation}
|f(t)|\le Ct^{-4\delta}.
\end{equation}
\end{proof}
By Lemma \ref{bootstrap}
\begin{equation}
\begin{aligned}
\left\Vert\frac{d}{dt}(\tilde U_{aj}(t)^*U_a(t)\tilde \Omega_jf)\right\Vert=\Vert i\tilde U_{aj}(t)^*(E_j-H^a)U_a(t)\tilde \Omega_j f\Vert\le Ct^{-2\delta}.
\end{aligned}
\end{equation}
Thus the following limit exists.
\begin{equation}
\begin{aligned}
h_j&=\lim_{\mbox{\scriptsize supp }\psi_a\downarrow\mathcal{T}_a}\lim_{t\to\infty}\tilde U_{aj}(t)^*U_a(t)\tilde \Omega_j f\\
&=\lim_{\mbox{\scriptsize supp }\psi_a\downarrow\mathcal{T}_a}\lim_{t\to\infty}\tilde U_{aj}(t)^*\varphi(H^a)(\psi_{aj}(H^a)-P^a_j)P_a^{{\varepsilon}_a}(t)e^{-itH}f.
\end{aligned}
\end{equation}
Since $\Vert\chi_{\{|x^a|>R\}}h_j\Vert\to0$ $(R\to\infty)$ and $[\chi_{\{|x^a|>R\}},\tilde U_{aj}(t)^*]=0$, we have as $R\to\infty$
\begin{equation}
\lim_{\mbox{\scriptsize supp }\psi_a\downarrow\mathcal{T}_a}\lim_{t\to\infty}\Vert\chi_{\{|x^a|>R\}}\varphi(H^a)(\psi_{aj}(H^a)-P^a_j)P_a^{{\varepsilon}_a}(t)e^{-itH}f\Vert\to0.
\end{equation}
On the other hand, since $\chi_{\{|x^a|<R\}}\varphi(H^a)$ is a compact operator we have for any $R>0$
\begin{equation}
\Vert \chi_{\{|x^a|<R\}}\varphi(H^a)(\psi_{aj}(H^a)-P^a_j)P_a^{{\varepsilon}_a}(t)e^{-itH}f\Vert\to0
\end{equation}
as supp $\psi_a\downarrow\mathcal{T}_a$. Thus when supp $\psi_a\downarrow\mathcal{T}_a$
\begin{equation}
\lim_{t\to\infty}\Vert\varphi(H^a)(\psi_{aj}(H^a)-P^a_j)P_a^{{\varepsilon}_a}(t)e^{-itH}f\Vert\to0.
\end{equation}
In particular we have when supp $\psi_a\downarrow\mathcal{T}_a$
\begin{equation}
\begin{aligned}
\Vert f-Qf\Vert=\bigl\Vert\lim_{t\to\infty}e^{itH}\varphi(H^a)(\psi_a(H^a)-P^a)P_a^{{\varepsilon}_a}(t)e^{-itH}f\bigr\Vert\to0.
\end{aligned}
\end{equation}
Namely
\begin{equation}
f=Qf=\lim_{t\to\infty}e^{itH}P^aP_a^{{\varepsilon}_a}(t)e^{-itH}f.
\end{equation}
\end{proof}
It is now obvious that the following theorem follows from Theorems \ref{sumG_a^+=H_c(H)}, \ref{rangeofW=Sa0} and \ref{emergentP^a}.
\begin{theorem}\label{rangeofwaveoperator}
Let the assumptions of Theorem \ref{Wapm=Sa1} be satisfied.
Let $a$ be a cluster decomposition with $2\le |a|\le N$. Similarly to the proof of Theorem \ref{shortrangecompletenesstheorem}, we extend the domain of the wave operator $W_a$ in Theorem \ref{adjointmodifiedwaveoperator} to the whole of $P^a\HH$, which we will denote by the same notation $W_a$. Then we have
\begin{align}
\mathcal{R}(W_a)=S_a^0=S_a^1.\label{rangeofwaveoperator0}
\end{align}
In particular we have the asymptotic completeness.
\begin{align}
\HH_c(H)=\bigoplus_{2\le|a|\le N}\mathcal{R}(W_a).\label{decompositionofH_cbyR(W_a)}
\end{align}
\end{theorem}

\section{Appendix}\label{calcusofpseuodandFourierintegraloperators}

We prove some lemmata.
\begin{lemma}\label{Fourierintegralpseudodifferential}
Let $a(x,\xi)$, $p(x,\xi)\in C^\infty(X\times X')$ satisfy for any multiindices $\alpha,\beta$
\begin{equation}
\begin{aligned}
&\sup_{(x,\xi)\in X\times X'}|\partial_x^\alpha\partial_\xi^\beta a(x,\xi)|<\infty,\\
&\sup_{(x,\xi)\in X\times X'}|\partial_x^\alpha\partial_\xi^\beta p(x,\xi)|<\infty.
\end{aligned}\label{Beinfty}
\end{equation}
Set for $f\in\mathcal{S}(X)$
\begin{equation}
\begin{aligned}
&a_{\varphi_a}(x,D_x)f(x)=(2\pi)^{-\nu(N-1)}\iint_{X\times X'}e^{i(\varphi_a(x,\xi)-y\xi)}a(x,\xi)f(y)dyd\xi,\\
&p(x,D_x)f(x)=(2\pi)^{-\nu(N-1)}\iint_{X\times X'}e^{i(x-y)\xi}p(x,\xi)f(y)dyd\xi.
\end{aligned}\label{Fourierpseudodifferentialopdef}
\end{equation}
Then we have
\begin{namelist}{8888}
\item[  {\rm 1)}]
\begin{equation}
a_{\varphi_a}(x,D_x)p(x,D_x)f(x)=(2\pi)^{-\nu(N-1)}\iint_{X\times X'}e^{i(\varphi_a(x,\xi)-y\xi)}s(x,\xi)f(y)dyd\xi,\label{productofFourierintegralpseudodiff}
\end{equation}
where
\begin{equation}
\begin{aligned}
&s(x,\xi)=(2\pi)^{-\nu(N-1)}\iint e^{-i(\eta-\xi)y}a(x,\eta)p(y+\nabla_\xi\varphi_a(\eta,x,\xi),\xi)dyd\eta,\\
&\nabla_\xi\varphi_a(\xi,x,\eta)=\int_0^1\nabla_\xi\varphi_a(x,\eta+\theta(\xi-\eta))d\theta.
\end{aligned}\label{symbolsofproduct}
\end{equation}
\item[  {\rm 2)}]
\begin{equation}
p(x,D_x)a_{\varphi_a}(x,D_x)f(x)=(2\pi)^{-\nu(N-1)}\iint_{X\times X'}e^{i(\varphi_a(x,\xi)-y\xi)}r(x,\xi)f(y)dyd\xi,\label{productofpseudodiffFourierintegral}
\end{equation}
where
\begin{equation}
\begin{aligned}
&r(x,\xi)=(2\pi)^{-\nu(N-1)}\iint e^{i(x-y)\eta}p(x,\eta+\nabla_x\varphi_a(x,\xi,y))a(y,\xi)dyd\eta,\\
&\nabla_x\varphi_a(x,\xi,y)=\int_0^1\nabla_x\varphi_a(y+\theta(x-y),\xi)d\theta.
\end{aligned}\label{symbolsofproduct-2}
\end{equation}
\item[  {\rm 3)}]
\begin{equation}
[a_{\varphi_a}(x,D_x),p(x,D_x)]f(x)=(2\pi)^{-\nu(N-1)}\iint e^{i(\varphi_a(x,\xi)-y\xi)}u(x,\xi)f(y)dyx\xi,\label{commutatorofFourierPseudo}
\end{equation}
where
\begin{equation}
\begin{aligned}
&u(x,\xi)=s(x,\xi)-r(x,\xi)\\
&=(2\pi)^{-\nu(N-1)}\iint e^{-i\eta y}\{a(x,\xi+\eta)p(y+\nabla_\xi\varphi_a(\xi+\eta,x,\xi),\xi)\\
&-p(x,\eta+\nabla_x\varphi_a(x,\xi,x+y))a(x+y,\xi)\}dyd\eta\\
&=a(x,\xi)\left(p(\nabla_\xi \varphi_a(x,\xi),\xi)-p(x,\nabla_x\varphi_a(x,\xi))\right)\\
&+\sum_{|\gamma|=1}(2\pi)^{-\nu(N-1)}\iint e^{-iy\eta}\\
&\times\biggl\{D_\eta^\gamma\{a(x,\xi+\eta)\int_0^1\partial_y^\gamma p(\theta y+\nabla_\xi\varphi_a(\xi+\eta,x,\xi),\xi)d\theta\}\\
&-D_y^\gamma\{\int_0^1\partial_\eta^\gamma p(x,\theta\eta+\nabla_x \varphi_a(x,\xi,x+y))d\theta a(x+y,\xi)\}\biggr\}dyd\eta.
\end{aligned}\label{differenceofsymbols}
\end{equation}
\end{namelist}
\end{lemma}
\begin{proof}
1)\ 
A direct calculation of oscillatory integrals yields
\begin{equation}
\begin{aligned}
a_{\varphi_a}(x,D_x)p(x,D_x)f(x)=(2\pi)^{-\nu(N-1)}\iint_{X\times X'}e^{i(\varphi_a(x,\eta)-z\eta)}s(x,\eta)f(z)dzd\eta,
\end{aligned}\label{productform}
\end{equation}
where
\begin{equation}
s(x,\eta)=(2\pi)^{-\nu(N-1)}\iint e^{i(\varphi_a(x,\xi)-\varphi_a(x,\eta)-(\xi-\eta)y)}a(x,\xi)p(y,\eta)dyd\xi.\label{intermediatesymbol}
\end{equation}
The following relation and a change of variable $y'=y-\nabla_\xi\varphi_a(\xi,x,\eta)$ give 1).
\begin{equation}
\varphi_a(x,\xi)-\varphi_a(x,\eta)=(\xi-\eta)\cdot\nabla_\xi\varphi_a(\xi,x,\eta).
\end{equation}\\
\noindent
2) Similarly we have
\begin{equation}
\begin{aligned}
p(x,D_x)a_{\varphi_a}(x,D_x)f(x)=(2\pi)^{-\nu(N-1)}\iint_{X\times X'}e^{i(\varphi_a(x,\eta)-z\eta)}r(x,\eta)f(z)dzd\eta,
\end{aligned}
\label{productform-2}
\end{equation}
where
\begin{equation}
r(x,\eta)=(2\pi)^{-\nu(N-1)}\iint e^{i(\varphi_a(y,\eta)-\varphi_a(x,\eta)+(x-y)\xi)}p(x,\xi)a(y,\eta)dyd\xi.
\label{intermediatesymbol-2}
\end{equation}
Noting
\begin{equation}
\varphi_a(y,\eta)-\varphi_a(x,\eta)+(x-y)\xi=(x-y)(\xi-\nabla_x\varphi_a(x,\eta,y)),
\end{equation}
we make a change of variable
\begin{equation}
\tilde \eta=\xi-\nabla_x\varphi_a(x,\eta,y)
\end{equation}
in \eq{intermediatesymbol-2}. Then we get
\begin{equation}
\begin{aligned}
&p(x,D_x)a_{\varphi_a}(x,D_x)f(x)\\
&=(2\pi)^{-\nu(N-1)}\iint e^{i(\varphi_a(x,\xi)-z\xi)}\\
&\times\iint e^{i(x-y)\eta}p(x,\eta+\nabla_x\varphi_a(x,\xi,y))a(y,\xi)dyd\eta f(z)dzd\xi.
\end{aligned}\label{formofproduct}
\end{equation}
The proof of 2) is complete.\\
\noindent
3) \ By 1) and 2) we have
\begin{equation}
\begin{aligned}
u(x,\xi)&=s(x,\xi)-r(x,\xi)\\
&=(2\pi)^{-\nu(N-1)}\iint e^{-i\eta y}a(x,\xi+\eta)p(y+\nabla_\xi\varphi_a(\xi+\eta,x,\xi),\xi)dyd\eta\\
&-(2\pi)^{-\nu(N-1)}\iint e^{-i\eta y}p(x,\eta+\nabla_x\varphi_a(x,\xi,x+y))a(x+y,\xi)dyd\eta.
\end{aligned}\label{differenceofsymbols-1}
\end{equation}

Taylor expanding the integrand of the first term with respect to $y$ around $y=0$ and that of the second term with respect to $\eta$ around $\eta=0$, we obtain
\begin{equation}
\begin{aligned}
&p(y+\nabla_\xi\varphi_a(\xi+\eta,x,\xi),\xi)\\
&=\sum_{|\alpha|<L}\frac{y^\alpha}{\alpha!}\partial_y^\alpha p(\nabla_\xi\varphi_a(\xi+\eta,x,\xi),\xi)\\
&+L\sum_{|\gamma|=L}\frac{y^\gamma}{\gamma !}\int_0^1(1-\theta)^{L-1}\partial_y^\gamma p(\theta y+\nabla_\xi\varphi_a(\xi+\eta,x,\xi),\xi)d\theta,\\
&p(x,\eta+\nabla_x\varphi_a(x,\xi,x+y))\\
&=\sum_{|\alpha|<L}\frac{\eta^\alpha}{\alpha!}\partial_\eta^\alpha p(x,\nabla_x\varphi_a(x,\xi,x+y))\\
&+L\sum_{|\gamma|=L}\frac{\eta^\gamma}{\gamma !}\int_0^1(1-\theta)^{L-1}\partial_\eta^\gamma p(x,\theta\eta+\nabla_x\varphi_a(x,\xi,x+y))d\theta.
\end{aligned}\label{Taylorexpansions}
\end{equation}
Substituting these with $L=1$ to \eq{differenceofsymbols-1}, integrating by parts and some calculation yield the following as desired (see \cite{[Kitada-Math-Anal]} for details).
\begin{equation}
\begin{aligned}
u(x,\xi)&=a(x,\xi)\left(p(\nabla_\xi \varphi_a(x,\xi),\xi)-p(x,\nabla_x\varphi_a(x,\xi))\right)\\
&+\sum_{|\gamma|=1}(2\pi)^{-\nu(N-1)}\\
&\times\iint e^{-iy\eta}\biggl\{D_\eta^\gamma\{a(x,\xi+\eta)\int_0^1\partial_y^\gamma p(\theta y+\nabla_\xi\varphi_a(\xi+\eta,x,\xi),\xi)d\theta\}\\
&-D_y^\gamma\{\int_0^1\partial_\eta^\gamma p(x,\theta\eta+\nabla_x \varphi_a(x,\xi,x+y))d\theta a(x+y,\xi)\}\biggr\}dyd\eta.
\end{aligned}
\end{equation}
\end{proof}
The following lemma follows from Theorem \ref{Theorem 5.2-2}.
\begin{lemma}\label{FourierconjugateFourierproduct}
Let $J_a$ be defined by \eq{5.17}.
\begin{equation}
\begin{aligned}
J_{a}f(x)&=(2\pi)^{-\nu(|a|-1)/2}\int_{\R^{\nu(|a|-1)}}e^{i\varphi_a(x_a,\xi_a)}\hat f(\xi_a,x^a)d\xi_a.
\end{aligned}\label{5.17-2}
\end{equation}
Then $P_a^{\varepsilon}(t)(J_a^*J_a-I)$, $(J_a^*J_a-I)P_a^{\varepsilon}(t)$, $P_a^{\varepsilon}(t)(J_aJ_a^*-I)$ and $(J_aJ_a^*-I)P_a^{\varepsilon}(t)$ satisfy the following estimates.
\begin{align}
&\Vert P_a^{\varepsilon}(t)(J_a^*J_a-I)\Vert \le C\langle t\rangle^{-\delta},\label{FourierconjugateFourieridentity}\\
&\Vert (J_a^*J_a-I)P_a^{\varepsilon}(t)\Vert \le C\langle t\rangle^{-\delta},\label{FourierconjugateFourieridentity-2}\\
&\Vert P_a^{\varepsilon}(t)(J_aJ_a^*-I)\Vert \le C\langle t\rangle^{-\delta},\label{FourierFourierconjugateidentity}\\
&\Vert (J_aJ_a^*-I)P_a^{\varepsilon}(t)\Vert \le C\langle t\rangle^{-\delta}.\label{FourierFourierconjugateidentity-2}
\end{align}
\end{lemma}
\begin{proof}
We note
\begin{eqnarray}
&&\begin{aligned}
J_{a}f(x)&=c_a\int_{\R^{2\nu(|a|-1)}}e^{i(\varphi_a(x_a,\xi_a)-y_a\xi_a)} f(y_a,x^a)dy_ad\xi_a,
\end{aligned}\label{5.17-2-2}\\
&&\begin{aligned}
J_{a}^*g(x)&=c_a\int_{\R^{2\nu(|a|-1)}}e^{i(x_a\xi_a-\varphi_a(y_a,\xi_a))} g(y_a,x^a)dy_ad\xi_a,
\end{aligned}\label{5.17-2-conjugate}
\end{eqnarray}
where $c_a=(2\pi)^{-\nu(|a|-1)}$. Then we have
\begin{equation}
\begin{aligned}
&\mathcal{F}_a(J_a^*J_a-I)\mathcal{F}_a^{-1}g(\xi_a,x^a)\\
&=
c_a\iint_{\R^{2\nu(|a|-1)}}e^{-i(\xi_a-\eta_a)z_a}
(|\det\nabla_{z_a}\nabla_{\xi_a}\varphi_a^{-1}(\xi_a,z_a,\eta_a)|-1)
g(\eta_a)d\eta_a dz_a,
\end{aligned}
\end{equation}
where $\nabla_{\xi_a}\varphi_a^{-1}(\xi_a,z_a,\eta_a)$ is the inverse map of $y_a\mapsto z_a=\nabla_{\xi_a}\varphi_a(\xi_a,y_a,\eta_a)$.
On the other hand Theorem \ref{Theorem 5.2-2}-ii) gives
\begin{equation}
|\partial_{z_a}^\alpha\partial_{\xi_a}^\beta\partial_{\eta_a}^\gamma(|\det\nabla_{z_a}\nabla_{\xi_a}\varphi_a^{-1}(\xi_a,z_a,\eta_a)|-1)|\le C_{\alpha\beta\gamma}\max_{1\le k\le k_a}(\langle z_{ak}\rangle^{-\delta-|\alpha|}).
\end{equation}
This and the factor $P_a^{\varepsilon}(t)$ give $\Vert P_a^{\varepsilon}(t)(J_a^*J_a-I)\Vert \le C\langle t\rangle^{-\delta}$. Other estimates are proved similarly.
\end{proof}
The following lemma is obvious, and the proof is left to the reader.
\begin{lemma}\label{differenceofpsuedoandadjpintpseudo}
\begin{equation}
\Vert P_a^{\varepsilon}(t)-P_a^{\varepsilon}(t)^*\Vert\le C\langle t\rangle^{-1}.
\end{equation}
\end{lemma}

\label{lastpage-01}

\begin{thebibliography}{88}

\bibitem{Agmon} {S.~Agmon}, { Spectral properties of Schr\"odinger operators and scattering theory}. \textit{Ann. Scuola Norm. Sup. Pisa } (4) \textbf{2} (1975), pp~151-218.

\bibitem{CV2} {A.~P.~Calder\'on} and {R.~Vaillancourt}, {A class of bounded pseudo-differential operators}. \textit{Proc. Nat. Acad. Sci. USA } \textbf{ 69} (1972), No. 5, pp~1185-1187.

\bibitem{[De]} {J.~Derezi\'nski}, { Asymptotic completeness of long-range
$N$-body quantum systems},  \textit{Annals of Math. }
 \textbf{\bf 138} (1993), pp~427-476.

\bibitem{[De-Ge]} {J.~Derezi\'nski} and {C.~Gerard}, {\textit{Scattering Theory of Classical and Quantum $N$-particle systems}}, Texts and Monographs in Physics, Springer, 1997.

\bibitem{[En]} {V.~Enss},
{ Introduction to asymptotic observables for multiparticle
 quantum scattering},  in \textit{Schr\"odinger Operators, Aarhus 1985} edited by
 E. Balslev, Lect.  Note in  Math., vol. 1218, Springer-Verlag,
1986, pp~61-92.

\bibitem{[En2]} {V.~Enss},
{ Separation of subsystems and clustered operators
for multiparticle quantum systems}, preprint Nr.213, Mathematics,
Freie Universitit Berlin, 1986.

\bibitem{[En3]} {V.~Enss},
{ Observables and asymptotic phase space localization of
N-body quantum scattering states}, unknown.

\bibitem{[E1]} V.~Enss, { Asymptotic completeness for quantum mechanical potential scattering I.  Short-range potentials}, \textit{Commun. Math. Phys.}, {\bf 61} (1978), pp~285-291.

\bibitem{[E1LR]} V.~Enss, { Asymptotic completeness for quantum mechanical potential scattering II.  singular and long-range potentials}, \textit{Ann. Physics} \textbf{ 119} (1979), pp~117-132.

\bibitem{[Enss4]} V.~Enss, { Quantum scattering theory for two- and three-body systems with potentials of short and long range} in \textit{Schr\"odinger operators}  ed. by S. Graffi, Lecture Notes in Math. 1159, Springer-Verlag, New York, 1985, pp~39-176.

\bibitem{[FH]} R.~Froese and I.~Herbst, { Exponential bounds and absence of positive eigenvalues for $N$-body Schr\"odinger operators}, \textit{Commun. Math. Phys.}
{\bf 87} (1982), pp~429-447.

\bibitem{[Gerard7]} C.~G\'erard, { Asymptotic completeness for 3-particle long-range systems}, \textit{Invent. Math.} \textbf{114} (1993), pp~333-397.

\bibitem{[G]} G.~M.~Graf, {Asymptotic completeness for $N$-body short-range quantum systems: A new proof}, \textit{Commun. Math. Phys.} {\bf 132} (1990), pp~73-101.

\bibitem{HMS} I.~Herbst, J.~S.~M{\o}ller and E.~Skibsted, { Spectral analysis of $N$-Body Stark Hamiltonians}, \textit{Commun. Math. Phys.} {\bf 174} (1995), pp~261-294.

\bibitem{Hor4} L.~H\"ormander, \textit{The Analysis of Linear Partial Differential Operators, IV Fourier Integral Operators}, Springer, 1985.

\bibitem{IK-1} H.~Isozaki and H.~Kitada, Asymptotic behavior of the scattering amplitude at high energies, in \textit{Differential Equations} ed. by I. W. Knowles and R. T. Lewis, North-Holland (1984), pp~329-334.

\bibitem{[IK]} H.~Isozaki and H.~Kitada, { Modified wave operators with time-independent modifiers}, \textit{Journal of the Fac. Sci, University of Tokyo},
Sec. IA, {\bf 32} (1985), pp~77-104.


\bibitem{Kato-wave} T.~Kato, { Wave operators and similarity for some non-selfadjoint operators}, \textit{Math. Annalen} {\bf 162} (1966), pp~258-279.


\bibitem{KK} T.~Kato and S.~T.~Kuroda, { Theory of simple scattering and eigenfunction expansions}, \textit{Functional Analysis and Related Fields}, Springer-Verlag, Berlin, Heidelberg, and New York, 1970, pp~99-131.

\bibitem{KK2} T.~Kato and S.~T.~Kuroda, {The abstract theory of scattering}, \textit{Rocky Mount. J. Math.} {\bf 1} (1971), pp~127-171.

\bibitem{Kitada-0} H.~Kitada, {On the completeness of modified wave operators}, \textit{Proc. of the Japan Academy} {\bf 52} (1976), pp~409-412.


\bibitem{KiI} H.~Kitada, {Scattering theory for Schr\"odinger operators with long-range potentials I, abstract theory}, \textit{J. Math. Soc. Japan}, {\bf 29} (1977), pp~665-691.

\bibitem{KiII} H.~Kitada, { Scattering theory for Schr\"odinger operators with long-range potentials II, spectral and scattering theory}, \textit{J. Math. Soc. Japan}, {\bf 30} (1978), pp~603-632.

\bibitem{[Ki-ScTime-Dep-Pot]} H.~Kitada, { Scattering theory for Schr\"odinger equations with time-dependent potentials of long-range type}, \textit{J. Fac. Sci., The Univ. Tokyo}, {\bf 29} (1982), pp~353-369. 



\bibitem{K3} H.~Kitada, { Asymptotic completeness for $N$-body Schr\"odinger operators I. Short-range potentials},
preprint (1984, February).

\bibitem{[K1]} H.~Kitada, { Asymptotic completeness of N-body wave operators  I. Short-range quantum systems}, \textit{Rev. in Math. Phys.} {\bf 3} (1991), pp~101-124.

\bibitem{[Ki($N$)]} H.~Kitada, { Asymptotic completeness of N-body wave operators  II. A new proof for the short-range case and the asymptotic clustering for long-range systems}, \textit{Functional Analysis and Related Topics, 1991}, Ed. by H. Komatsu, Lect. Note in Math., vol. 1540, Springer-Verlag, 1993, pp~149-189.

\bibitem{[Kitada-S]} H.~Kitada, { Scattering spaces and a decomposition of continuous spectral subspace of $N$-body quantum systems}, 1999. http://arxiv.org/abs/math/9912244

\bibitem{[Kitada-e-book]} H.~Kitada, \textit{Quantum Mechanics}, Lectures in Mathematical Sciences, vol. 23, The University of Tokyo, 2005, ISSN 0919-8180, ISBN 1-000-01896-2. http://arxiv.org/abs/quant-ph/0410061

\bibitem{KS} H.~Kitada, { A remark on simple scattering theory}, \textit{Commun. Math. Anal.} {\bf 11} (2011), No. 2, pp~124-138.

\bibitem{K2} H.~Kitada, {Fourier integral operators with weighted symbols and micro-local resolvent estimates}, \textit{J. Math. Soc. Japan}, {\bf 39} (1987), pp~455-476.

\bibitem{[Kitada-Math-Anal]} H.~Kitada,~\textit{Introduction to Mathematical Analysis}, Gendai-Suugaku-Sha, Oct. 7, 2012, x+591 pp. ISBN 978-4-7687-0407-3.


\bibitem{Kumano-go} H.~Kumano-go, \textit{Pseudo Differential Operators}, MIT Press, Cambridge, Massachusetts and London, England (1983).


\bibitem{Kuroda1} S.~T.~Kuroda, { Scattering theory for differential operators, I operator theory}, \textit{J. Math. Soc. Japan} {\bf 25} (1973), pp~75-104.

\bibitem{Kuroda2} S.~T.~Kuroda, { Scattering theory for differential operators, II self-adjoint elliptic operators}, \textit{J. Math. Soc. Japan} {\bf 25} (1973), pp~222-234.

\bibitem{MS} J.~S.~M{\o}ller and E.~Skibsted, { Spectral theory of time-periodic many-body systems}, \textit{Advances in Mathematics} {\bf 188} (2004), pp~137-221.

\bibitem{PSS} P.~Perry, I.~M.~Sigal and B.~Simon, { Spectral analysis of $N$-Body Schr\"odinger operators}, \textit{Ann. Math.} {\bf 114} (1981), pp~519-567.

\bibitem{RS} M.~Reed and B.~Simon, \textit{Methods of Modern Mathematical Physics III: Scattering Theory}, Academic Press (1979).

\bibitem{Saito1}  Y.~Sait{\= o}, { The principle of limiting absorption for second-order differential equations with operator-valued coefficients}, \textit{Publ. RIMS, Kyoto Univ.}, {\bf 7} (1971/72), pp~581-619.

\bibitem{Saito2} Y.~Sait{\= o}, { Spectral and scattering theory for second-order differential operators with operator-valued coefficients}, \textit{Osaka J. Math.}, {\bf 9} (1972), pp~463-498.



\bibitem{[SS]} I.~M.~Sigal and A.~Soffer, { The $N$-particle scattering problem: Asymptotic completeness for short-range systems}, \textit{Ann. Math.} {\bf 126} (1987), pp~35-108.


\bibitem{[XPWang6]} X.~P.~Wang, { On the three-body long-range scattering problems}, \textit{Lett. Math. Phys.} {\bf 25} (1992), pp~267-276.


\bibitem{[Yaf]} D.~R.~Yafaev, { New channels in three-body long-range scattering}, \textit{Equations aux deriv\'ees partielles}, Publ. Ecole Polytechnique, Palaiseau XIV
(1994), pp~1--11.

\bibitem{[Yaf-2]} D.~R.~Yafaev, { New channels of scattering for three-body
quantum systems with long-range potentials}, \textit{Duke Math. J.} {\bf 82} (1996), pp~553-584.

\bibitem{[Yaf-3]} D.~R.~Yafaev, \textit{Scattering theory: some old and new problems}, Lect. Notes Math., v. 1735, Springer-Verlag (2000).

\bibitem{[Yaf-4]} D.~R.~Yafaev, \textit{Mathematical Scattering Theory, Analytic Theory}, Mathematical Surveys
and Monographs Vol. 158, American Mathematical Society, Providence, Rhode Island (2010).

\end{thebibliography}
\end{document}